\documentclass[11pt]{article}
\usepackage[dvipsnames,svgnames,table]{xcolor}
\usepackage{enumitem}
\usepackage{graphicx}
\usepackage{fancybox}
\usepackage{comment}
\usepackage{xcolor}
\usepackage{nameref}
\usepackage{bm, bbm}
\usepackage[T1]{fontenc}

\definecolor{ForestGreen}{rgb}{0.1333,0.5451,0.1333}
\definecolor{DarkRed}{rgb}{0.8,0,0}
\definecolor{Red}{rgb}{1,0,0}
\usepackage[linktocpage=true,
pagebackref=true,colorlinks,
linkcolor=ForestGreen,citecolor=ForestGreen,
bookmarks,bookmarksopen,bookmarksnumbered]
{hyperref}
\usepackage[nottoc]{tocbibind}

\usepackage{amsmath}
\usepackage{amssymb}
\usepackage{amsthm}

\usepackage[ruled, noend, linesnumbered]{algorithm2e}
\usepackage[noend]{algpseudocode}

\usepackage{subfig}

\usepackage{thm-restate}

\usepackage{float}
\usepackage{cleveref}

\newtheorem{theorem}{Theorem}[section]
\newtheorem{informaltheorem}[theorem]{Informal Theorem}

\newtheorem{corollary}[theorem]{Corollary}
\newtheorem{lemma}[theorem]{Lemma}

\newtheorem{claim}[theorem]{Claim}

\newtheorem{fact}[theorem]{Fact}

\newtheorem{assumption}[theorem]{Assumption}
\newtheorem{problem}[theorem]{Problem}

\newtheorem{definition}[theorem]{Definition}

\newtheorem*{theorem*}{Theorem}
\newtheorem*{corollary*}{Corollary}
\newtheorem*{conjecture*}{Conjecture}
\newtheorem*{lemma*}{Lemma}
\newtheorem*{thm*}{Theorem}
\newtheorem*{prop*}{Proposition}
\newtheorem*{obs*}{Observation}
\newtheorem*{definition*}{Definition}

\newtheorem*{remark*}{Remark}
\newtheorem*{rec*}{Recommendation}

\newenvironment{fminipage}%
  {\begin{Sbox}\begin{minipage}}%
  {\end{minipage}\end{Sbox}\fbox{\TheSbox}}

\newcommand{\defeq}{:=}

\def\norm#1{\left\| #1 \right\|}

\def\calR{\mathcal{R}}

\newcommand\bbeta{\boldsymbol{\beta}}

\def\aa{\pmb{\mathit{a}}}
\newcommand\bb{\boldsymbol{\mathit{b}}}
\newcommand\cc{\boldsymbol{\mathit{c}}}
\newcommand\dd{\boldsymbol{\mathit{d}}}

\newcommand\ff{\boldsymbol{\mathit{f}}}
\renewcommand\gg{\boldsymbol{\mathit{g}}}
\newcommand\hh{\boldsymbol{\mathit{h}}}

\newcommand\uu{\boldsymbol{\mathit{u}}}
\newcommand\vv{\boldsymbol{\mathit{v}}}
\newcommand\ww{\boldsymbol{\mathit{w}}}
\newcommand\yy{\boldsymbol{\mathit{y}}}
\newcommand\zz{\boldsymbol{\mathit{z}}}
\newcommand\xx{\boldsymbol{\mathit{x}}}

\newcommand\veczero{\boldsymbol{0}}
\newcommand\vecone{\boldsymbol{1}}

\renewcommand\AA{\boldsymbol{\mathit{A}}}

\newcommand\DD{\boldsymbol{\mathit{D}}}

\newcommand\FF{\boldsymbol{\mathit{F}}}

\newcommand\II{\boldsymbol{\mathit{I}}}
\newcommand\JJ{\boldsymbol{\mathit{J}}}

\newcommand\MM{\boldsymbol{\mathit{M}}}

\newcommand\UU{\boldsymbol{\mathit{U}}}

\newcommand\VV{\boldsymbol{\mathit{V}}}
\newcommand\XX{\boldsymbol{\mathit{X}}}
\newcommand\YY{\boldsymbol{\mathit{Y}}}

\renewcommand\O{\widetilde{O}}

\newcommand\R{\mathbb{R}}

\DeclareMathOperator{\nnz}{nnz}

\DeclareMathOperator*{\argmin}{arg\,min}
\DeclareMathOperator*{\argmax}{arg\,max}

\DeclareMathOperator*{\ddiag}{\boldsymbol{\mathrm{Diag}}}

\DeclareMathOperator*{\im}{im}

\newcommand{\poly}{\mathrm{poly}}

\newcommand{\diam}{\mathrm{diam}}
\newcommand{\vvec}{\mathtt{vec}}

\newcommand\Z{\mathbb{Z}}

\newcommand\oracle{\mathcal{O}}

\newcommand{\eps}{\epsilon} 
\renewcommand{\O}{\widetilde{O}}

\renewcommand{\l}{\langle}
\renewcommand{\r}{\rangle}

\newcommand{\otilde}{\O}
\renewcommand{\forall}{\mathrm{\text{ for all }}}

\newcommand{\g}{\nabla}

\newcommand{\bDelta}{\boldsymbol{\Delta}}

\renewcommand{\hat}{\widehat}
\renewcommand{\tilde}{\widetilde}

\DeclareFontFamily{U}{mathb}{\hyphenchar\font45}
\DeclareFontShape{U}{mathb}{m}{n}{<5> <6> <7> <8> <9> <10> gen * mathb
<10.95> mathb10 <12> <14.4> <17.28> <20.74> <24.88> mathb12}{}
\DeclareSymbolFont{mathb}{U}{mathb}{m}{n}
\DeclareMathSymbol{\rcirclearrow}{\mathbin}{mathb}{'367}

\newif\ifrandom
\randomtrue

\renewcommand{\l}{\langle}
\renewcommand{\r}{\rangle}

\newcommand{\todolater}[1]{}

\newcommand{\cE}{\mathcal{E}}

\newcommand{\cT}{\mathcal{T}}

\newcommand{\cR}{\mathcal{R}}

\newcommand{\cF}{\mathcal{F}}

\DeclareUnicodeCharacter{2113}{$\ell$}

\newcommand{\dset}{\mathcal{D}}
\newcommand{\sset}{\mathcal{S}}
\newcommand{\xset}{\mathcal{X}}
\newcommand{\yset}{\mathcal{Y}}
\newcommand{\zset}{\mathcal{Z}}

\newcommand{\uset}{\mathcal{U}}
\newcommand{\vset}{\mathcal{V}}

\newcommand{\almostTime}{\hat{O}}

\newcommand{\imbal}{\mathrm{im}}
\newcommand{\congest}{\mathrm{cong}}

\newcommand{\linfbound}{\rho}

\newcommand{\scaledparam}{C}

\newcommand{\bestrsp}{\mathrm{BR}}

\newcommand{\anyX}{\Bar{\XX}}
\newcommand{\algoutput}[1]{\Hat{#1}}
\newcommand{\altdelta}{\upsilon}
\newcommand{\algo}{\mathcal{A}}
\newcommand{\OPT}{\mathrm{OPT}}

\newcommand{\Time}{\mathcal{T}}
\newcommand{\instance}{\mathcal{M}}
\newcommand{\barrier}{\phi}

\newcommand{\mcf}{MCF}
\newcommand{\Mcf}{MCF}

\newcommand{\edgenumber}{N}

\newcommand{\mwmcObj}{\cE_{\mathrm{mwmc}}}
\newcommand{\dcmcObj}{\cE_{\mathrm{dcmc}}}

\newcommand{\boxbound}{R}

\newcommand{\RdLower}{\boxbound_\ell}
\newcommand{\RdUpper}{\boxbound_\mu}
\newcommand{\usetSize}{\mathrm{Diam}(\uset)}
\newcommand{\usetBound}{\boxbound_{\uset}}

\newcommand{\fixedalpha}{\Bar{\alpha}}
\newcommand{\fixedbeta}{\Bar{\beta}}
\newcommand{\coefnorm}{\lambda}

\newcommand{\delSCO}{\delta_{\mathrm{SCO}}}
\newcommand{\delBR}{\delta_{\mathrm{BR}}}

\newcommand{\linfreg}{r_{1, \infty}}

\newcommand{\fullit}{\zz}
\newcommand{\halfit}{\zz_{+}}
\newcommand{\auxit}{\zz_{\mathrm{aux}}}

\newcommand{\BRconstant}{c}
\newcommand{\gradbound}{L}
\newcommand{\psibound}{\boxbound_{\psi}}

\newcommand{\scgeq}{>}
\newcommand{\offset}{\Delta}
\newcommand{\lowerend}{\theta_{\ell}}
\newcommand{\upperend}{\theta_{u}}
\newcommand{\apxfactor}{\omega}

\newcommand{\eff}{f}
\newcommand{\PhiBar}{\Bar{\Phi}}
\newcommand{\PhiHat}{\Hat{\Phi}}
\usepackage[margin=1in]{geometry}
\usepackage{orcidlink} 
\usepackage{lscape}

\usepackage{amssymb,amsthm,amsmath}
\usepackage[ruled,linesnumbered]{algorithm2e}

\SetCommentSty{mycommfont}
\SetKw{Continue}{continue}
\SetKw{Break}{break}
 \DontPrintSemicolon 
\usepackage{multirow}

\allowdisplaybreaks

\title{Accelerated Approximate Optimization of \\Multi-Commodity Flows on Directed Graphs}

\author{Li Chen\\Independent\\lichenntu@gmail.com
\and Andrei Graur\\Stanford University\\ agraur@stanford.edu
\and Aaron Sidford\\Stanford University\\ sidford@stanford.edu
}
\date{\today}

\begin{document}

\maketitle

\begin{abstract}
We provide $m^{1+o(1)}k\eps^{-1}$-time algorithms for computing multiplicative $(1 - \epsilon)$-approximate solutions to multi-commodity flow problems with $k$-commodities on $m$-edge directed graphs, including concurrent multi-commodity flow and maximum multi-commodity flow.

To obtain our results, we provide new optimization tools of potential independent interest. First, we provide an improved optimization method for solving $\ell_{q, p}$-regression problems to high accuracy. This method makes $\O_{q, p}(k)$ queries to a high accuracy convex minimization oracle for an individual block, where $\O_{q, p}(\cdot)$ hides factors depending only on $q$, $p$, or $\poly(\log m)$, improving upon the $\O_{q, p}(k^2)$ bound of [Chen-Ye, ICALP 2024]. As a result, we obtain the first almost-linear time algorithm that solves $\ell_{q, p}$ flows on directed graphs to high accuracy. Second, we present optimization tools to reduce approximately solving composite $\ell_{1, \infty}$-regression problems to solving $m^{o(1)}\eps^{-1}$ instances of composite $\ell_{q, p}$-regression problem. The method builds upon recent advances in solving box-simplex games [Jambulapati-Tian, NeurIPS 2023] and the area convex regularizer introduced in [Sherman, STOC 2017] to obtain faster rates for constrained versions of the problem. Carefully combining these techniques yields our directed multi-commodity flow algorithm.

\end{abstract}

\thispagestyle{empty}

\newpage
\tableofcontents
\thispagestyle{empty}
\newpage

\setcounter{page}{1}
\section{Introduction}
\label{sec:intro}

We consider multi-commodity flow (\mcf{}) problems on directed graphs. In these problems there is a capacitated, directed graph $G = (V, E, \uu)$ with $n$-nodes $V$, $m$-edges $E$, positive
edge capacities $\uu \in \Z^E_{> 0}$, and demands $\DD = (\dd_1,\ldots,\dd_k) \in \Z^{V \times k}$. Such \mcf{} problems optimize (directed) \mcf{} $\FF = (\ff_1, \ldots, \ff_k) \in \R^{E \times k}_{\geq 0}$ consisting of $k$ single-commodity flows $\ff_1,\ldots,\ff_k \in \R^E_{\geq 0}$, where the flow of commodity $i$ on edge $e$ is given by $[\ff_i]_e = [\ff_{ie}] = \FF_{ei}$. These problems seek to find a \mcf{} that optimizes a certain objective subject to constraints induced by the node demands and edge capacities. 

More precisely, we say that \emph{$\FF$ routes $\DD$}, denoted $\imbal(\FF) = \DD$, if the node imbalance of each flow $\ff_i$ is $\dd_i$, denoted $\imbal(\ff_i) = \dd_i$. (See \Cref{sec:prelim} for definitions and notation.) Additionally, we say that the \emph{congestion of $\FF$ is $C$}, if the maximum ratio of the total flow on an edge to its capacity is $C$, i.e., $\congest(\FF) = C$, where $\congest(\FF) \defeq \max_{e \in E} \sum_{i \in [k]} \ff_{ie} / \uu_e$. We consider the following canonical \mcf{} problems (broadly following the terminology of \cite{Madry10}).
\begin{itemize}
    \item \textbf{Concurrent \mcf{}}: compute \mcf{} $\FF \in \R^{E \times k}_{\geq 0}$ and $\alpha \in \R_{\geq 0}$ such that $\alpha$ is maximized, 
    $\FF$ routes $\alpha \DD$, and $\congest(\FF) \leq 1$.
    
    \item \textbf{Maximum \mcf{}}: compute \mcf{} $\FF \in \R^{E \times k}_{\geq 0}$ and $\bbeta \in \R_{\geq 0}^{k}$ such that $\sum_{i \in [k]} \bbeta_i$ is maximized, $\FF$ routes $(\bbeta_1 \dd_1,\ldots,\bbeta_k \dd_k)$, and $\congest(\FF) \leq 1$.
    
    \item \textbf{Maximum weighted \mcf{}}: compute \mcf{} $\FF \in \R^{E \times k}_{\geq 0}$ and $\bbeta \in \R_{\geq 0}^{k}$ such that $\sum_{i \in [k]} \ww_i\bbeta_i$ is maximized for input positive weights $\ww \in \Z^{k}_{\ge 0}$, 
    $\FF$ routes $(\bbeta_1 \dd_1,\ldots,\bbeta_k \dd_k)$, and $\congest(\FF) \leq 1$.
\end{itemize}
We assume throughout the paper that input parameters, e.g., $\DD, \uu, \ww$, are \textit{polynomially bounded}, i.e., the absolute value of each entry is in $\in \Z \cap [0, \poly(m)]$. 

\mcf{} problems are fundamental, ubiquitous problems in combinatorial optimization, theoretical computer science, and algorithmic graph theory. They generalize the well-studied maximum-flow problem (e.g., consider the above problems for $k = 1$ and $\dd_1 = \vecone_s - \vecone_t$ for $s,t\in V$), they have been studied for decades, and they have multiple applications. Correspondingly, there have been multiple algorithmic advances for optimizing MCFs (See \Cref{sec:related_work}). 

Unfortunately, despite the potential utility of efficient \mcf{} algorithms and advances in solving maximum flow and related single-commodity flow problem \cite{ds08, ckmst11, madry13, Mad16, ls20, kls20, AMV20, CKL+22, van2023deterministic, chen2023almost}, obtaining almost-linear time algorithms for solving \mcf{} to high-accuracy has been elusive.
\cite{dkz22} shed light on this challenge and shows that even when  $k = 2$ solving \mcf{} up to accuracy $\epsilon$ is as hard as solving an arbitrary linear program up to accuracy that is smaller than $\epsilon$ by at most a polynomial factor in the number of constraints and variables.
Though there have been MCF runtime improvements when $G$ is dense~\cite{bz23}, \cite{dkz22} implies that, barring a major breakthrough in linear programming and linear system solving, almost linear time, high-accuracy algorithms for \mcf{} are not to be expected. 

Consequently, we study the problem of computing approximately optimal \mcf{}s, i.e., feasible MCFs with value within a $(1\pm\epsilon)$ of the optimal value. 
We focus on developing algorithms that run in time almost linear in the naive representation of the output, i.e., $\almostTime(m k \cdot \poly(\epsilon^{-1}))$, where $\almostTime(\cdot)$ hides $m^{o(1)}$ factors. There is an interesting line of work on expander routing and the cut matching game~\cite{KRV09, KKOV07, ghaffari2017distributed, chuzhoy2021decremental, haeupler2023maximum,chen2023simple,haeupler2024low} which obtains better dependencies on $k$ (see \Cref{sec:related_work}). However, this large $k$ regime is not the focus of our paper; applying our techniques in this regime is an interesting direction for future research.

In certain cases, such almost linear time ($\almostTime(m k \cdot \poly(\epsilon^{-1}))$) algorithms can be obtained. Classic work of \cite{radzik1996fast,grigoriadis1996approximate} shows that certain \mcf{} problems, e.g., concurrent \mcf{} and cost-constrained \mcf{}, can be reduced to solving $\tilde{O}(k \epsilon^{-2})$ single-commodity minimum cost flow problems. Applying these results with recent advances in solving minimum cost flow in almost linear time \cite{CKL+22, van2023deterministic, chen2023almost}, yield state-of-the-art $\almostTime(m k \epsilon^{-2})$ time algorithms for these problems (see \Cref{sec:related_work} for a more comprehensive survey of related results). Additionally, for undirected graphs, \cite{sherman2017area} showed that concurrent \mcf{} can be solved in $\almostTime(m k \epsilon^{-1})$.

In light of these results, we ask whether the $\epsilon$ dependence obtained via \cite{radzik1996fast,grigoriadis1996approximate} can be improved to match that obtainable for undirected graphs via \cite{sherman2017area}. The $\epsilon^{-2}$ dependence in \cite{radzik1996fast,grigoriadis1996approximate} stems from the application of a multiplicative weight update (MWU) framework \cite{AHK12MWU} and out-performing MWU is notoriously challenging and setting specific \cite{GPS16}. On the other hand, \cite{sherman2017area} and follow up work \cite{JST19,boob2019faster,CST20,assadi2022semi,JT23} open the door to improving the $\epsilon^{-1}$ to $\epsilon^{-2}$. Unfortunately, directly applying the approach of \cite{sherman2017area} for directed graphs faces immediately obstacles (as we discuss at the end of \Cref{sec:related_work}). 

\subsection{Our Results}
\label{sec:results}

The main results in this paper are a new almost linear time algorithms for solving \mcf{} problems on directed graphs. For the concurrent, maximum (weighted) \mcf{} problems, and more, we provide algorithms that compute $(1\pm\epsilon)$ approximate solutions in $\almostTime(m k \epsilon^{-1})$ time.
This improves upon the prior state-of-the-art runtimes which either had a worse dependence on one of $m, n, \epsilon^{-1}$ \cite{CLS19STOC, CKL+22, bz23} or applied only to undirected graphs \cite{sherman2017area}. 
We obtain these results by reducing these problems to a more general problem, \emph{(directed) composite \mcf{}s} (\Cref{prob:convex_mcflow}), which we show how to solve approximately by solving $\almostTime(k \eps^{-1})$ single-commodity convex flow problems.

Broadly, our method uses, extends, and applies the optimization framework underlying \cite{sherman2017area} to improve the efficacy of methods like those in \cite{radzik1996fast,grigoriadis1996approximate}. Our method builds upon \cite{sherman2017area}, follow up work~\cite{boob2019faster,JST19,CST20,assadi2022semi,JT23}, and optimization techniques of \cite{akps19, AKPS22, chen2023high}, to reduce \mcf{} problems to approximately solving $\almostTime(\epsilon^{-1})$ sub-problems. By using methods similar to those of \cite{radzik1996fast,grigoriadis1996approximate}, these subproblems can be solved in $\tilde{O}(k)$ iterations of solving single commodity convex flow problems.
These single commodity flow problems are perhaps more complicated and general then simply solving minimum cost flow, however they are still solvable to high accuracy in almost linear time due to \cite{CKL+22,van2023deterministic}. Consequently, altogether, we essentially show how we can accelerate the approximate optimization of \mcf{} problems.

\paragraph{Composite \mcf{}.}
More precisely, to obtain our results, we develop almost linear time algorithms for solving a more general \mcf{} problem which we call \emph{(directed) composite \mcf{}}. Composite  \mcf{} problems are defined on what we call an \emph{\mcf{} instance}, $\instance$, which consists of a directed graph $G$, a number of commodities $k$, and a set of demands $\DD \in \Z^{V \times k}$. For a vector $\bbeta \in \R^k$, we say that a flow $\FF \in \R_{\geq 0}^{E \times k}$ routes demand $\DD \ddiag(\bbeta) = (\bbeta_1 \dd_1, \ldots \bbeta_k \dd_k)$, and write $\imbal(\FF) = \DD \ddiag(\bbeta)$, if 
\begin{equation}\label{eq:node_imbalance}
    \sum_{e = (v, w), w \in V} \FF_{ei} - \sum_{e = (w, v), w \in V} \FF_{ei} = \beta_i (\dd_i)_v, \forall i \in [k], v \in V \,,
\end{equation}
meaning that the node imbalance for commodity $i \in [k]$ and vertex $v \in V$ is equal to $\beta_i (\dd_i)_v$. The formal definition of a \mcf{} instance is provided below. 

\begin{definition}[\mcf{} Instance]
\label{def:instance}
A \emph{\mcf{} instance} $\instance = (G = (V, E, \uu), k, \DD)$ on $k$ commodities consists of a directed graph $G$ with vertices $V$, edges $E$, capacities $\uu \in \Z_{>0}^{E}$, and demands $\DD \in \Z^{V \times k}$ with $\dd_i \ne \veczero^m, \forall i \in [k]$.
We call $\instance$ \emph{poly-regular} if $\|\DD\|_{\infty}, \|\uu\|_{\infty} \le \poly(m)$. We define 
\[
S_{\instance}
\defeq
\left\{
(\FF, \bbeta) ~ | ~ \im(\FF) = \DD \ddiag(\bbeta) \text{, and } \sum_{i \in [k]} \FF_{ei} \le \uu_e, \forall e \in E
\right\}\,.
\]
\end{definition}

Having defined a MCF instance, we now describe an instance of the composite MCF problem (\Cref{prob:convex_mcflow} defined below) and the goal of the problem. 
An instance for composite MCF,
in addition to the input for a MCF instance, receives
range bounds $0 \leq \RdLower \leq \RdUpper$. 
The feasible region for this problem, which we denote by $S_{\instance}(\RdLower, \RdUpper)$, consists
of all MCFs $\FF \in \R_{\geq 0}^{E \times k}$ and $\bbeta \in \R^k$ such that $\FF$ routes $(\bbeta_1 \dd_1, \ldots, \bbeta_k \dd_k) = \DD \ddiag(\bbeta)$, and each $\bbeta_i \in [\RdLower, \RdUpper]$. The goal of composite MCF is then to minimize the sum of convex costs on the entries of $\FF$ and $\bbeta$ over points in $S_{\instance}(\RdLower, \RdUpper)$. 

\begin{restatable}[(Directed) Composite \mcf{}]{problem}{cvxKCommFlow}\label{prob:convex_mcflow}
Given \mcf{} instance  $\instance$ 
and convex costs $\{c_{ei}\}$ and $\{v_i\}$ for each edge $e \in E$ and commodity $i \in [k]$, the \emph{composite \mcf{}} problem asks to solve
\begin{align*}
	\min_{(\FF, \bbeta) \in S_{\instance}(\RdLower, \RdUpper)}
	\dcmcObj(\FF,\bbeta)
	\text{ where }
	 \dcmcObj(\FF,\bbeta) \defeq
	c(\FF, \bbeta) + \congest(\FF)
\end{align*}
and 
\[
 c(\FF, \bbeta) \defeq  c(\FF) + v(\bbeta)
 \text{ where  }
 c(\FF) \defeq \sum_{e \in E, i \in [k]} c_{ei}(\FF_{ei})
 \text{  and }
 v(\bbeta) \defeq \sum_{i \in [k]} v_i(\bbeta_i)\,.
\]
\end{restatable}

Under mild assumptions on the instance and the $c_{ei}$ and $v_i$ we show that composite \mcf{} can solved to accuracy $\epsilon$ in $\almostTime(m k \epsilon^{-1})$ time. In particular, we assume functions $\{c_{ei}\}_{i \in V, e \in E}$ and $\{v_i\}_{i \in V}$ are what we call \emph{$m$-decomposable}, which includes a broad range of convex functions (see \Cref{def:decomposable} for more details).
The result is formalized in \Cref{thm:mainKCommFlow} and proved in \Cref{sec:kcommflow}.

\begin{restatable}[Composite MCF Algorithm]{theorem}{mainKCommFlowThm}\label{thm:mainKCommFlow}
Consider the setting of \Cref{prob:convex_mcflow}. 
There is an algorithm that, given as input a poly-regular \mcf{} instance $\instance$, $m$-decomposable functions $\{c_{ei}\}_{i \in V, e \in E}$ and $\{v_i\}_{i \in V}$, where $\{v_i\}_{i \in V}$ are $\poly(m)$-Lipschitz, and $\eps, \delta = \Omega(1/\poly(m))$, in $\almostTime(mk \eps^{-1})$ time outputs $(\algoutput{\FF}, \algoutput{\bbeta}) \in S_{\instance}(\RdLower, \RdUpper)$ such that 
\begin{align*}
\dcmcObj(\algoutput{\FF}, \algoutput{\bbeta}) \le \dcmcObj(\FF^{\star}, \bbeta^{\star}) + \eps\cdot\congest(\FF^{\star}) + \delta,
\end{align*}
where $(\FF^{\star},\bbeta^{\star}) \in \argmin_{(\FF, \bbeta) \in S_{\instance}(\RdLower, \RdUpper)} \dcmcObj(\FF, \bbeta)$.
\end{restatable}

Leveraging \Cref{thm:mainKCommFlow}, we obtain $\almostTime(m k \epsilon^{-1})$-time algorithms for the concurrent MCF (\Cref{coro:approxConcurrentFlow}), the maximum MCF, as well as the maximum weighted MCF (\Cref{coro:approxMaxKCommFlow}).
The reduction of these problems to the composite MCF problem is shown in \Cref{sec:kcommflow}. 

\paragraph{Composite $\ell_{1,\infty}$-regression.} 
Our $\almostTime(m k \epsilon^{-1})$ time algorithm for \Cref{prob:convex_mcflow} makes key use in recent advances in single commodity flow \cite{CKL+22, van2023deterministic, chen2023almost}. In particular, we leverage that given any fixed commodity $i$ it is possible to minimize to additive error $\epsilon$, in $\almostTime(m \log \epsilon^{-1})$ time, a convex function 
that is comprised of $m$-decomposable (\Cref{def:decomposable}) costs on the edges and a convex, $\poly(m)$-Lipschitz function of $\beta$, the amount of demand routed. 
The formal definition of $m$-decomposable functions involves the notion of \emph{computable functions} \cite{CKL+22, van2023deterministic, chen2023almost} which describes functions that admit a self-concordant barrier for their epigraph with natural bounds.
$m$-decomposable functions constitute a class of convex functions 
that can be written as a sum of constantly many functions that are represented as $\min_{a+b = x, a \in [\lowerend, \upperend]} c_{1}(a) + c_{2}(b)$, where $c_1, c_2$ are  $(m, \log^{O(1)} c)$-computable. 
Leveraging this single-commodity flow subroutine, we show that to solve \Cref{prob:convex_mcflow}, it suffices to solve a problem which we call \emph{composite $\ell_{1,\infty}$-regression} provided we have suitable access to what we call \emph{$(\Time, \delta)$-separable convex minimization oracles}. 

Composite $\ell_{1,\infty}$-regression problems are defined on a set $\xset \subseteq \R^{m \times k}_{\ge 0}$ which is decomposable as $\xset = \oplus_{j \in [k]} \xset_j$ where  each $\xset_j$ is a convex subset of $\R^{m}$ and comes with a convex cost $\psi_j : \xset_j \rightarrow \R$.
The goal of the composite $\ell_{1,\infty}$-regression problem is to find $\XX \in \xset$ that minimizes the sum of convex costs $\psi_j(\XX_{j})$ and $\|\XX\|_{1, \infty}$, the $\ell_{1, \infty}$-norm of $\XX$ which is defined as $\max_{i\in[m]} \sum_{j\in[k]}|\XX_{ij}|$ (see \Cref{prob:comp_ell1inf_regr}). 
We consider this problem when
we are given access to what we call \emph{separable convex minimization oracles} for each $\xset_j, j \in [k]$ (\Cref{def:sepCvxOracle}). Such an oracle for $\xset_j$, when given $m$-decomposable costs on each entry can find $\xx_j \in \xset_j$ that minimizes the sum of the given entry-wise costs and $\psi_j(\xx_j)$.
\begin{definition}[$(\Time, \delta)$-SCO]
\label{def:sepCvxOracle}
We call $\oracle$ a \emph{$(\Time, \delta)$-separable convex minimization oracle ($(\Time, \delta)$-SCO)} for closed, non-empty $\xset \subseteq \R^m$ and convex $\psi: \xset \to \R$ if when queried with a collection of $m$-decomposable functions $\{c_{i}(\cdot)\}_{i \in [m]}$,  
and error $\delta > 0$ if  
it outputs, in $\Time$ time, $\algoutput{\xx} \in \xset$ with 
\begin{align*}
 \psi(\algoutput{\xx}) + c(\algoutput{\xx}) \leq
    \min_{\xx \in \xset} \psi(\xx) + c(\xx) + \delta
\text{ where }
c(\xx) \defeq \sum_{i \in [m]} c_i(\xx_i)\,.
\end{align*}
\end{definition}

We now define the composite $\ell_{1,\infty}$-regression problem. 

\begin{problem}[Composite $\ell_{1,\infty}$-Regression]\label{prob:comp_ell1inf_regr}
Given $\xset \subseteq \R^{m \times k}_{\ge 0}$ decomposable as $\xset \defeq \otimes_{j \in [k]} \xset_j$, where each $\xset_j$ is a nonempty, convex, compact subset of $\R^m_{\ge 0}$ and convex costs $\{\psi_j\}$ on each $\xset_j, j \in [k]$, the \emph{composite $\ell_{1,\infty}$-regression problem} asks to solve
\begin{align}
\label{eq:L1infMin}
\min_{\XX \in \xset} \cE_{1, \infty}(\XX) 
\text{ where }
	\cE_{1, \infty}(\XX) \defeq \psi(\XX) + \norm{\XX}_{1, \infty}
\text{ and }
	\psi(\XX) \defeq \sum_{j \in [k]} \psi_j(\XX_j)
\,. 
\end{align}
\end{problem}

Given a $(\Time_j, \delSCO)$-SCOs, for sufficiently small $\delSCO = \frac{1}{\poly(m, k)}$, for each component $(\xset_j, \psi_j)$ for $j \in [k]$, we provide an algorithm that solves \Cref{prob:comp_ell1inf_regr} to $(1 + \eps)$-approximation in $\almostTime(\sum_{j \in [k]} (\Time_j + m) \epsilon^{-1})$ time (stated informally below). 

\begin{informaltheorem}
\label{informal:approxL1InfMin}
Consider the setting of \Cref{prob:comp_ell1inf_regr}. 
Given $\eps, \delta = \Omega(1 / \poly(m))$ and a $(\Time_j, \delSCO)$ SCO for each component $(\xset_j, \psi_j)$ for $j \in [k]$ with $\delSCO = \delta / \poly(m)$,
\Cref{alg:composite_ell1inf_regression}, in $\almostTime(\sum_{j \in [k]} (\Time_j + m) \epsilon^{-1})$ time, outputs $\algoutput{\XX} \in \xset$ such that for $\XX^{\star} = \argmin_{\XX \in \xset} \cE_{1, \infty}(\XX)$, 
\begin{align*}
    \cE_{1, \infty}(\algoutput{\XX}) \le \min_{\XX \in \xset} \cE_{1, \infty}(\XX) + \epsilon \|\XX^{\star}\|_{1,\infty} + \delta
    \,.
\end{align*}
\end{informaltheorem}

As mentioned earlier, we show how to use this result  to obtain our $\almostTime(mk\eps^{-1})$ time algorithm for solving the composite \mcf{} problem (\Cref{prob:convex_mcflow}).
Here, we decompose an MCF instance among the $k$ commodities and the feasible region for each commodity corresponds to a typical single commodity flow instance.
We
implement a $(\almostTime(m), \delSCO)$-SCO for each commodity $i \in [k]$ using the convex flow solver from \cite{CKL+22}. 
This yields our composite MCF algorithm that runs in $\almostTime(mk\eps^{-1})$ time (\Cref{thm:mainKCommFlow}). 

Our composite $\ell_{1, \infty}$ regression algorithm is formalized as \Cref{thm:approxL1InfMin} and proved in \Cref{sec:L1infMin}. We develop and apply two optimization tools to obtaining this result. The first tool is a faster (by a factor of roughly $k$ compared to \cite{chen2023high}) high-accuracy algorithm for composite $\ell_{q, p}$-regression. The second tool is an optimization framework which allows for fine-grained rates for constrained, composite box-simplex games. We apply the first tool to implement iterates in the second tool which, when carefully applied, yields \Cref{thm:approxL1InfMin}. 

\paragraph{High-accuracy composite $\ell_{q, p}$-regression.}
To obtain our $\almostTime(\sum_{j \in [k]} (\Time_j + m) \epsilon^{-1})$-time algorithm for solving \Cref{prob:comp_ell1inf_regr}, we introduce optimization tools of potential independent interest.
First, generalizing the  aforementioned $\ell_{1, \infty}$ norm, we define the $\ell_{q, p}$ norm for any $\XX \in \R^{m \times k}$ and $q, p \ge 1$ as $\|\XX\|_{q,p} \defeq (\sum_{i\in[m]} (\sum_{j\in[k]} |\XX_{ij}|^q)^{p})^{\frac{1}{pq}}$ and define the \emph{composite $\ell_{q,p}$-regression} problem as follows:

\begin{restatable}[Composite $\ell_{q, p}$-Regression]{problem}{compLqpRegrProb}\label{prob:comp_ellqp_regr}
Given $\xset \subseteq \R^{m \times k}$ decomposable as $\xset \defeq \otimes_{j \in [k]} \xset_j$ where each $\xset_j$ is a nonempty, convex, compact subset of $\R^{m}$, convex costs $\{\psi_{j}\}$ on each $\xset_j, j \in [k]$, the \emph{composite $\ell_{q,p}$-regression problem} asks to solve
\begin{align}
\label{eq:comp_ellqp_regr}
\min_{\XX \in \xset} \cE_{q, p}(\XX)
\text{ where }
\cE_{q, p}^{\coefnorm}(\XX) \defeq \psi(\XX) + \coefnorm \norm{\XX}_{q, p}^{pq}
\text{ and }
\psi(\XX) \defeq \sum_{j \in [k]} \psi_j(\XX_j)
\,.
\end{align}
\end{restatable}

We show how to combine the multiplicative weight update (MWU) procedure of \cite{radzik1996fast} and the iterative refinement framework from \cite{chen2023high} to solve the $\ell_{q, p}$-regression problem to additive error $\delta$ in 
$O_{q, p}(\sum_{j \in [k]} (\Time_j \allowbreak + m) \log \frac{1}{\delta})$ time, where $O_{q, p}(\cdot)$ is used to hide constants that depend only on $q$ and $p$. 
The iterative refinement approach of \cite{chen2023high}
allows us to reduce this problem to solving $O_{q, p}(\log \frac{1}{\delta})$ \emph{residual problems} up to $O_{q, p}(1)$-approximation.
By leveraging this
much relaxed requirement on the approximation, we can afford using the $\O(k \eps^{-2})$-query algorithmic framework of \cite{radzik1996fast}.
In fact, we show that $k$ queries are enough to solve the residual problem to $O_{q, p}(1)$-approximation (\Cref{lem:approxResSolver}), which yields an improvement over \cite{chen2023high} by a factor of $k$.
The result is formalized as \Cref{thm:comp_ellqp_regr} and proved in \Cref{sec:comp_ellqp_regr}.

\begin{informaltheorem}
\label{informal:comp_ellqp_regr}
Consider the setting of \Cref{prob:comp_ellqp_regr} where $1 < q \le 2 \le p$ and $\coefnorm = \exp(\O(1))$. 
Given a $(\Time_j, \delSCO)$ SCO for each component $(\xset_j, \psi_j)$ for $j \in [k]$
there is an algorithm that, in $O_{q, p}(\sum_{j \in [k]} (\Time_j + m))$ time, outputs $\algoutput{\XX} \in \xset$ such that 
\begin{align*}
\cE_{q, p}^{\coefnorm}(\algoutput{\XX}) \le \min_{\XX \in \xset} \cE_{q, p}^{\coefnorm}(\XX) + 3 k \delSCO\,.
\end{align*}
\end{informaltheorem}

While we defer the exact expression $O_{q, p}(1)$ to \Cref{sec:comp_ellqp_regr} (see \Cref{thm:comp_ellqp_regr}), we note that when $p = O(\sqrt{\log m})$ and $q = 1+1/p$, our algorithm for solving \Cref{prob:comp_ellqp_regr} runs in $\almostTime(\sum_{j \in [k]} (\Time_j + m))$ time. 
As a result, for this regime of $p$ and $q$, 
we obtain the first almost-linear time algorithm that solves $\ell_{q, p}$ flows on directed graphs to high accuracy in $\almostTime(m k \log \epsilon^{-1})$ time, improving upon 
\cite{chen2023high} which obtained an $\almostTime(m k^2 \log \epsilon^{-1})$ rate for an undirected variant of the problem that didn't have composite terms.\footnote{It seems that the algorithm of \cite{chen2023high} can be straightforwardly extended to the directed case even with the composite term present. Our focus is on improving the quadratic dependency on $k$.} 

\paragraph{Fine-grained rates for constrained, composite box-simplex games.}
Towards obtaining this reduction, we build upon recent advances~\cite{sherman2017area, JST19, CST20, JT23} in solving \emph{box-simplex games},
a class of problems that include
$\ell_\infty$-regression over the $[0, 1]^n$ box, to solve \Cref{prob:comp_ell1inf_regr} efficiently. More precisely, we consider the following problem, which we call \emph{constrained, composite box-simplex games} 
\begin{equation}\label{eqprob:gen_box_simplex}
    \min_{\uu \in \uset} \max_{\vv \in \Delta^m} \vv^\top \AA \uu + \psi(\uu)\,,
\end{equation}
with arbitrary $\uset \subseteq [0, 1]^n$ and composite convex function $\psi$. 
Naively extending \cite{JT23} to solve this problem would get additive error $\epsilon$ in $\otilde(\|\AA\|_{\infty} \epsilon^{-1})$ steps. 
Mapping this to
the context of composite $\ell_\infty$ regression, with $\uset = \xset$ and $\AA = \AA_{1, \infty}$, which is the matrix satisfying $\|\AA_{1, \infty} \XX\|_{\infty} = \|\XX\|_{1, \infty}$, this would result in a query complexity (in calls to SCOs) of $\frac{k}{\|\XX^{\star}\|_{1, \infty}} k \epsilon^{-1}$ where $\XX^{\star}$ is the optimal solution (since $\norm{\AA_{1,\infty}}_{\infty} = k$).
When applied to \mcf{} problems (through careful binary search) this would result in a runtime of $\almostTime(m k^2 \epsilon^{-1})$ rather than $\almostTime(m k \epsilon^{-1})$. 

To improve upon this query complexity, 
as we explain in more detail in \Cref{sec:short_overview}, we build upon \cite{sherman2017area, JST19, CST20, JT23} to establish a method that solves \eqref{eqprob:gen_box_simplex} to additive error in  $\otilde(\max_{\uu \in \uset} \||\AA| \uu\|_{\infty} \epsilon^{-1})$ steps.
Each step of the method involves a suitable subproblem approximately optimizing over $\uset$. To obtain this method, we provide a straightforward generalization of\cite{JST19} to handle the constraint set $\uu$, the composite term $\psi$, and approximate error in the subproblems. Perhaps more interestingly, we also introduce a regularizer, resembling that of \cite{sherman2017area}, which enables this fine-grained improvement from $\|\AA\|_\infty$ to $\max_{\uu \in \uset} \||\AA| \uu\|_{\infty}$ in the convergence. 

To apply this new algorithm for \eqref{eqprob:gen_box_simplex} to obtain our results for solving composite $\ell_{1,\infty}$ one more insight is required. Rather than working with $\xset$, for which even $\max_{\uu \in \uset} \||\AA| \uu\|_{\infty}$ is too large to obtain our result, we work with a carefully selected subset $\sset \subseteq \xset$, which is known to contain the optimal solution $\XX^{\star}$. We then solve $\min_{\XX \in \sset} \max_{\yy \in \Delta^m} \yy^\top \AA \XX + \psi(\XX)$. We then set $\uset = \sset$ in \eqref{eqprob:gen_box_simplex}. This choice of $\sset$ improves $\max_{\uu \in \uset} \||\AA| \uu\|_{\infty}$ but makes $\uset$ no longer decomposable. However, we show that our improved high-accuracy algorithm for composite $\ell_{q, p}$-regression can still solve the necessary optimization sub-problems. 

Altogether, setting $p = 2 \lceil \sqrt{\log m} \rceil + 1$ and $q = 1 + \frac{1}{p}$, 
this reduces $\ell_{1, \infty}$-regression problem to $\almostTime(\eps^{-1})$ instances of the $\ell_{q, p}$-regression problem, 
each of which we solve to additive error $\delta = O(1/\poly(m))$ in $\almostTime(\sum_{j \in [k]} (\Time_j + m))$ time. 
For more details, see \Cref{sec:short_overview}.

\paragraph{Open problems.}
This work improves upon the best known running times for solving a range of \mcf{} problems. Further improving the runtime for any of the three \mcf{} problems considered here would be exciting progress. Related to this, determining the optimal dependence on $\eps^{-1}$ when $k = O(1)$ is an interesting open question in light of the hardness result of \cite{dkz22}. 
Furthermore, a natural follow-up question is to obtain algorithms for \mcf{} problems that only use $\otilde(k \epsilon^{-1})$ single-commodity flow optimizations.
More generally, 
determining the optimal query complexity for \Cref{prob:convex_mcflow}, or  finding further applications of our framework for composite $\ell_{1, \infty}$-regression or simpler ways of obtaining our result in \Cref{thm:approxL1InfMin} would be a fruitful and timely direction, given the recent attention to extragradient methods.

\subsection{Overview of approach}
\label{sec:short_overview}

Here, we provide a high-level overview of our approach.
First, in \Cref{subsec:shorter_overview}, we discuss how our methods apply to the concurrent \mcf{} problem. Then, in \Cref{subsec:extragrad_overview} and 
\Cref{subsec:high_acc_ellpq_overview} we explain in greater detail the techniques used in obtaining our result regarding efficiently solving \Cref{prob:comp_ell1inf_regr} (\Cref{thm:approxL1InfMin}), which,  
as discussed in \Cref{sec:intro}, enables us to obtain all of our results. Finally, in \Cref{subsec:obstacle} we provide further context for this approach discussing challenges in applying \cite{sherman2017area} directly to achieve our results for MCF.
 
\subsubsection{Concurrent \mcf{}}
\label{subsec:shorter_overview}

Here we present an overview of how our algorithm computes an $(1+\eps)$-approximate solution to the concurrent \mcf{} problem (\Cref{coro:approxConcurrentFlow}) on unit-capacitated, directed graphs.
For simplicity, we assume that the optimal congestion is $1$, whereas to prove \Cref{coro:approxConcurrentFlow} in its full generality, we binary search to obtain a value that is multiplicatively close to the optimal congestion.
To obtain our result, we reduce concurrent \mcf{} to $\almostTime(k \eps^{-1})$ single commodity flow computations and use the results for optimizing single-commodity flows (\Cref{coro:convexFlow}) to achieve an $\almostTime(mk\eps^{-1})$ running time.

First, we view concurrent \mcf{} as a \emph{box-simplex} game (see \Cref{subsec:extragrad_overview} for description) and use an extragradient method to compute an approximate solution for:
\begin{align*}
    \min_{\imbal(\FF) = \DD, \FF \ge 0} \congest(\FF) = \min_{\imbal(\FF) = \DD, \FF \ge 0} \max_{e \in E} \sum_{i \in [k]} \FF_{ei}= \min_{\imbal(\FF) = \DD, \FF \ge 0} \max_{\yy \in \Delta_E} \sum_{e \in E, i \in [k]} \yy_e \FF_{ei}
\end{align*}
Our convergence lemma (\Cref{lem:convergence}) 
reduces this problem to sequentially implementing $\O(\linfbound \eps^{-1})$ iterations of our extragradient method, where $\linfbound$ is the maximum congestion of a feasible flow. Each iteration consists of solving an instance of the following sub-problem: 
\begin{align*}
\min_{\imbal(\FF) = \DD, \FF \ge 0} \sum_{e \in E, i \in [k]} c_{ei}(\FF_{ei})\,,
\end{align*}
for some $m$-decomposable costs $\{c_{ei}(\cdot)\}$. 
By our assumption that the optimal congestion is $1$, we may restrict the space of flows by $\FF_{ei} \le 1, \forall e, i$ and since each sub-problem is decomposable among commodities, each iteration can be solved using $k$ single commodity flow computations, taking $\almostTime(mk)$-time. However, unfortunately, this only allows us to bound $\linfbound$ by $k$.
This is due to the fact that, without additional structure, the optimal solution to any subproblem above, could have the property that every commodity puts $1$ unit of flow on some edge $e$. 
This approach then results in a $\almostTime(mk^2\eps^{-1})$-time algorithm for solving the concurrent \mcf{} problem.

To improve the runtime by a factor of $k$, we constrain the domain to control the congestion. 
Specifically, we refine the feasible set to contain flows such that $\|\FF\|_{q, p} = \almostTime(1)$
for $p = 2 \lceil \sqrt{\log m} \rceil + 1$ and $q = 1 + \frac{1}{p}$. 
Any flow in this new feasible set will have a congestion at most $\almostTime(1)$ due to our choice of $q$ and $p.$
Now, the extragradient method, working over this set of flows of smaller congestion, takes only $\almostTime(\eps^{-1})$ iterations, but each iteration involves a more complex family of sub-problems:
\begin{align}\label{prob:restricted_step_concurrent}
\min_{\imbal(\FF) = \DD, \FF \ge 0, \|\FF\|_{q, p} \le R} \sum_{e \in E, i \in [k]} c_{ei}(\FF_{ei})
\end{align}
for some $R = \almostTime(1).$
The main difficulty of solving \eqref{prob:restricted_step_concurrent} is the fact that we have to optimize over a constraint set that is not decomposable over the commodities, in contrast with the set $[0, 1]^{E \times k}$. However, 
using binary search (see \Cref{lem:unconstrained_red}), this is equivalent, for a suitable coefficient $\coefnorm$, to solving to high accuracy the following composite $\ell_{q,p}$-norm \mcf{} problem:
\begin{align*}
\min_{\imbal(\FF) = \DD, \FF \ge 0} \sum_{e \in E, i \in [k]} c_{ei}(\FF_{ei}) + \coefnorm \|\FF\|_{q, p}^{pq}
\end{align*}

We then show that the composite $\ell_{q,p}$-norm \mcf{} problem can be solved to high accuracy in $\almostTime(mk)$-time, thus improving upon the $\otilde(m k^2)$ runtime of \cite{chen2023high}. 
Note that the objective function $\cE^{\coefnorm}_{q, p}(\FF) \defeq \sum_{e, i}c_{ei}(\FF_{ei}) + \coefnorm \|\FF\|_{q, p}^{pq}$ in the problem above is no longer separable over the commodities. 
Consequently, we follow the iterative refinement framework for $\ell_{q, p}$-regression developed in \cite{chen2023high}. Specifically, 
our algorithm starts with an initial feasible solution $\FF^{(0)}$ and, 
at the $t$-th iteration, given the current feasible solution $\FF^{(t)}$, we define a residual proxy $\cR(\bDelta) \approx \cE^{\coefnorm}_{q, p}(\FF^{(t)} + \bDelta) - \cE^{\coefnorm}_{q, p}(\FF^{(t)})$. 
Our algorithm then finds an update $\algoutput{\bDelta}$ that approximately solves the residual problem:
\begin{align*}
    \min_{\bDelta: \FF^{(t)} + \bDelta \text{ is feasible}} \cR(\bDelta)
\end{align*}

In contrast to \cite{chen2023high}, which solves the subproblems to high-accuracy, we only 
obtain an approximate solution of $(\alpha, \delta/3)$ multiplicative-additive error for each residual problem, (i.e., $\cR\left(\algoutput{\bDelta}; \XX\right) \le  \frac{1}{\alpha}  \min_{\bDelta: \XX + \bDelta \in \xset}\cR(\bDelta; \XX) + \delta/3$). Fortunately, by \Cref{coro:convRate}, this leads 
to an $\almostTime(\alpha \log \frac{1}{\delta})$\footnote{The $m^{o(1)}$ factor is due in part to our choice of $p$ and $q$. The asymptotic dependency on $p$ and $q$ is $(\frac{p}{q-1})^{O(\frac{1}{q-1})}$ (\Cref{thm:comp_ellqp_regr}).} bound on the iteration count for computing a $(1+\delta)$-approximation to the composite $\ell_{q,p}$-norm \mcf{} problem.
To obtain a solution of $(m^{o(1)}, \delta/3)$ multiplicative-additive error to the residual problem in $\almostTime(mk)$ time, 
we update one commodity at a time, minimizing a dynamic objective, which is reminiscent of the MWU framework used in \cite{radzik1996fast,grigoriadis1996approximate}.
Each single-commodity update is done via the almost-linear time single-commodity flow solver. 

In the remainder of the section, we elaborate on each of the ingredients of this approach, in the greater generality to which it applies to proving \Cref{thm:approxL1InfMin}. At the end, we explain why extending the approach of \cite{sherman2017area} for undirected graphs to directed graphs encounters roadblocks.

\subsubsection{Extragradient Methods}
\label{subsec:extragrad_overview}

To solve $\ell_{1, \infty}$-regression problem using $\almostTime(k\eps^{-1})$ queries to the SCOs (\Cref{thm:approxL1InfMin}), it suffices to work with the alterative formulation 
\[\min_{\XX \in \xset} \max_{\yy \in \Delta^m} \yy^\top \AA_{1, \infty} \XX + \psi(\XX)\,,\]
where $\xset$ is a convex, compact set, $\psi:\xset \to \R$ is a convex function, and
$\AA_{1, \infty} \in \R^{m \times mk}$ 
is the matrix with the property that $\|\XX\|_{1, \infty} = \max_{\yy \in \Delta^m} \yy^\top \AA_{1, \infty} \XX$ for every $\XX \in \R^{mk}$. 
A more general problem, which we call \emph{constrained, composite box-simplex games}, is 
\begin{equation}\label{eqprob:constr_box_simplex}
    \min_{\uu \in \uset} \max_{\vv \in \Delta^m} \vv^\top \AA \uu + \psi(\uu)\,,
\end{equation}
where $\uset \subseteq [0, 1]^n, \psi:\uset \to \R$ is convex, $\Delta^m = \{\yy \in \R^m_{\ge 0}: \sum_i \yy_i = 1\}$ is the $m$-dimensional simplex and $\AA \in \R^{m \times n}$. Later in \Cref{sec:all_of_4}, we design algorithms for solving this problem. 

Constrained, composite box simplex games generalize a problem known in the literature as \emph{box-simplex games}, where $\psi(\uu) = 0 \forall \uu \in [0, 1]^n$.
Specifically, the formulation is 
\begin{equation}\label{eqprob:box_simplex}
    \min_{\uu \in [0, 1]^n} \max_{\vv \in \Delta^m} \vv^\top \AA \uu + \psi(\uu)\,,
\end{equation}
State-of-the-art works that solve 
the box-simplex formulation \eqref{eqprob:box_simplex} apply extragradient methods. 
In particular, extragradient methods work by solving, at each step, two proximal gradient steps with respect to a convex regularizer $r:[0, 1]^n \times \Delta^m
\to \R$, of the form $\min_{(\uu, \vv) \in [0, 1]^n \times \Delta^m} \langle \ww, (\uu, \vv) \rangle + r(\uu, \vv)$ for some vector $\ww \in \R^n \times \R^m$. 
Working with a regularizer $r$ that is area-convex (see \Cref{def:area_convexity}), prior works~\cite{sherman2017area, JST19, CST20, JT23} show how to solve box-simplex games up to an additive error $\epsilon$ in 
$\otilde(\mathrm{range}(r, [0, 1]^n) \epsilon^{-1})$ steps, where $\mathrm{range}(r, \uset)\defeq \max_{\uu \in \uset, \vv \in \Delta^m} |r(\uu, \vv)|$ for any $\uset 
\subseteq [0, 1]^n$. 
Additionally these papers provide regularizers where $\mathrm{range}(r, [0, 1]^n) = \O(\|\AA\|_{\infty})$ and each iteration can be implemented in $O(\nnz(\AA))$ (number of non-zero entries of $\AA$) time, leading to $\O(\nnz(\AA) \|\AA\|_{\infty} \epsilon^{-1})$ runtimes for the problem. 

To obtain our improved results for \Cref{prob:comp_ell1inf_regr}, we develop methods that obtain a more fine-grained convergence rate in terms of the constraint set $\uset$ and handle an arbitrary $\uset \subseteq [0, 1]^n$ and a composite term $\psi$. 
While extending the framework of \cite{JT23} to work over $\uset \subseteq [0, 1]^n$ and with general convex $\psi$ is somewhat straightforward, identifying a more fine-grained condition and showing that it is obtainable requires more insight. 
For our purposes, we wish to pick a regularizer $r$ for which $\mathrm{range}(r, \uset)$ is $O(\max_{\uu \in \uset} \||\AA| \uu\|_{\infty})$ (up to logarithmic factors in $\|\AA\|_{\infty}$, $\max_{\uu \in \uset} \||\AA| \uu\|_{\infty}$, and  $\diam(\uset)$). 
Yet, unfortunately, the regularizers and analysis in \cite{sherman2017area, JST19, CST20, JT23} (naively) all only give bounds of $\mathrm{range}(r, \uset) = \Omega(\|\AA\|_{\infty})$.

To motivate the need for a more fine-grained convergence rate for solving \eqref{eqprob:constr_box_simplex}, consider our accuracy requirement for \Cref{prob:comp_ell1inf_regr}. We require additive error $\epsilon \|\XX^{\star}\|_{1, \infty}$ and to obtain it the number of steps suggested by prior methods is $\frac{\mathrm{range}(r, \xset)}{\|\XX^{\star}\|_{1, \infty}} \epsilon^{-1}$. This can be prohibitively large if $\|\XX^{\star}\|_{1, \infty}$ could be much smaller than $\mathrm{range}(r, \xset)$. In particular, in our example in \Cref{subsec:shorter_overview}, we had $\|\XX^{\star}\|_{1, \infty} = 1$ while $\mathrm{range}(r, \xset) = \O(\|\AA\|_{\infty}) = \O(k)$. 

To get around this issue we apply two techniques. 
First, we run our extragradient method on a smaller restriction $\sset \subseteq \xset$ that is known to contain the optimal solution, in the hopes that $\mathrm{range}(r, \sset)$ will be on the order of $\|\XX^{\star}\|_{1, \infty}$ up to factors of $\almostTime(1)$. 
In particular, we instead solve
the problem 
\begin{equation}\label{eqprob:comp_ell1inf_reform}
    \min_{\XX \in \sset} \max_{\yy \in \Delta^m} \yy^\top \AA \XX + \psi(\XX)\,,
\end{equation}
where 
\begin{align*}
    \sset \defeq \{\XX \in \xset \cap [0, 1]^{m \times k}: \|\XX\|_{q, p} \le \linfbound^{\star} \cdot m^{\frac{1}{pq}}\}
\end{align*}
for $p = 2 \lceil \sqrt{\log m} \rceil + 1$, $q = 1 + \frac{1}{p}$, and $\linfbound^{\star} = \|\XX^{\star}\|_{1, \infty}$, where $\XX^{\star} = \argmin_{\XX \in \xset} \cE_{1, \infty}(\XX)$ (i.e., the optimal solution for \Cref{prob:comp_ell1inf_regr}). 
Working over $\sset$ is favorable as we can obtain better bounds on the ``width'' of $\sset$, $\linfbound \defeq \max_{\XX \in \sset} \allowbreak \|\XX\|_{1, \infty}$. 
In particular, \Cref{obs:qp1inf} gives us that $\linfbound := \max_{\XX \in \sset} \|\XX\|_{1, \infty}$ is small relative to $\linfbound^{\star}$ (i.e., $\linfbound = \linfbound^{\star} \cdot m^{o(1)}$). 
Second, we use our \emph{doubly-entropic} regularizer, which we discuss towards the end of this subsection, when 
working over $\sset$ and obtain
an additive error $O(\epsilon \linfbound)$ in $\almostTime(\epsilon^{-1})$ iterations. 

To implement the steps of our extragradient method, which reduces to 
implementing the proximal gradient steps in a way that handles a composite convex function $\psi:\uset \to \R$ in the objective and a constraint set $\uset$, 
in \Cref{subsec:implement_oracles}, we provide a condition for what optimization operations we can perform over $\uset$.
Specifically, 
\Cref{defn:apx_best_resp} defines an oracle (termed \textit{approximate best response}) for approximately solving subproblems on $\uset$ that involve minimizing a linear function in $\uu$ plus a term $r(\uu, \vv)$ for a fixed $\vv \in \Delta^m$ over a certain subset $\uset \subseteq [0, 1]^n$. 
This slightly extends the implementability of the (extra)gradient steps comparing to \cite{JT23}, as they assume the ability to exactly minimize such subproblems over the $[0, 1]^{n}$ box, while computing high-accuracy solutions is enough for our purposes. 
In the context of \eqref{eqprob:comp_ell1inf_reform}, 
the calls to this approximate best response oracle have the form, for some fixed $\Bar{\yy}$,
\[\min_{\XX \in \sset} \langle \hh, \XX \rangle + r(\XX, \Bar{\yy}) + \psi(\XX)\,,\]
and can indeed be implemented, as we discuss more in \Cref{subsec:high_acc_ellpq_overview}. 

Technically, each modification primarily involves carefully tweaking \cite{JT23} and generalizing their analysis. However, we think the observation that we can 
choose a regularizer which allows $\mathrm{range}(r, \uset)$ to be bounded (up to log factors) by $\max_{\uu \in \uset} \||\AA| \uu\|_{\infty}$ rather than $\|\AA\|_{\infty}$ 
can be particularly powerful (as our MCF results illustrate). 

We conclude by discussing the regularizer we leverage in greater detail. 
To solve \eqref{eqprob:box_simplex}, 
\cite{JST19} presents a regularizer $r^{\text{AC}}$,\footnote{Here, AC stands for area-convexity.}
that is area-convex with respect to $\gg$ over the entire space of $[0, 1]^n \times \Delta^m:$
\begin{align}
\label{eq:sherman}
r^{\text{AC}}(\uu, \vv) = \sum_{i \in [m], j \in [n]} \vv_i |\AA_{ij}| \uu_j^2 + \O(\|\AA\|_{\infty}) \sum_{i \in [m]} \vv_i \log \vv_i = \l\vv, |\AA| \uu^2\r + \O(\|\AA\|_{\infty}) \sum_{i \in [m]} \vv_i \log \vv_i.
\end{align}
To obtain a regularizer $r$ with $\mathrm{range}(r, \uset) \le \linfbound = \max_{\uu \in \uset} \||\AA| \uu\|_{\infty}$, one might attempt to leverage that $\linfbound$ could be smaller than $\|\AA\|_{\infty}$ and scale down  $r^{\text{AC}}$ by a factor of $\|\AA\|_{\infty} / \linfbound$. 
However, scaling $r^{\text{AC}}$'s range down to $\O(\linfbound)$ does not preserve area-convexity with respect to the gradient operator $\gg$. 
This comes from the fact that $\l\vv, |\AA| \uu^2\r$ differs from $\l\vv, |\AA| \uu\r$ by a factor of $\|\AA\|_{\infty}$ no matter how small $\||\AA| \uu\|_{\infty}$ is.
Instead, in \Cref{sec:regularizer}, we present a regularizer that is area-convex and has range $\O(\linfbound)$.
This reduces our algorithm's runtime for solving \Cref{prob:comp_ell1inf_regr} by a factor of $\frac{\|\AA\|_\infty}{\linfbound}$. Our regularizer is defined as follows:

\begin{definition}[Doubly Entropic Regularizer, a complete definition is in \Cref{def:regularizer_complete}]
\label{def:regularizer}
Consider the setting of \Cref{prob:comp_ell1inf_regr} and 
let $\AA \in \R^{m \times n}$.
Define the family of regularizers $r^{\alpha, \xi}: \xset \times \Delta^m \to \R$ parametrized by $\alpha \ge 0$ and $\xi > 0$, as follows: 
\begin{align*}
    r^{\alpha, \xi}(\uu, \vv) 
    &\defeq \sum_{i \in [m], j \in [n]} \left(\vv_i + \xi\right) |\AA_{ij}| (\uu_j + \xi) \log (\uu_j + \xi) + \alpha \sum_{i \in [m]} \vv_i \log \vv_i \,.
\end{align*}
\end{definition}

Our regularizer is related to a regularizer proposed by \cite{sherman2017area}, which also had a term of form $\sum_{i \in [m], j \in [n]} \vv_i |\AA_{ij}| \uu_j \log \uu_j$, and
has a similar form to $r^{\text{AC}}$ \eqref{eq:sherman}.
Our regularizer introduces a $\xi$ term in the expression $\l\vv + \xi, |\AA| (\uu + \xi) \log (\uu + \xi)\r$ and is carefully designed to hadnle several constraints. First, we wish to ensure that $\l\vv + \xi, |\AA| (\uu + \xi) \log (\uu + \xi)\r$ is a good approximation (within a $\poly(\log m, \log n)$ factor) of $\l\vv, |\AA| \uu\r$. Second, we want to ensure that the entries of $\nabla^2 r^{\alpha, \xi}(\uu, \vv)$ are bounded well in terms of $\max_{\uu \in \uset} \||\AA|\uu\|_{\infty}$. Lastly, for our applications to \Cref{prob:comp_ell1inf_regr}, we want the terms $\log(\vv_i + \xi)$ to be bounded in absolute by $O(\log m)$ everywhere. 

Note that our doubly entropic regularizer $r^{\alpha, \xi}$ is defined in terms of not only $\alpha$ and $\xi$, but also the matrix $\AA$. However, to simplify notation, we omit $\AA$ from the notation of $r^{\alpha, \xi}$ and always use $r^{\alpha, \xi}$ in contexts where the underlying matrix $\AA$ is implicit. 

\subsubsection{High-Accuracy Composite \texorpdfstring{$\ell_{q, p}$}{Lqp}-Regression in \texorpdfstring{$\almostTime(k)$}{OHat(k)}-queries}
\label{subsec:high_acc_ellpq_overview}

As discussed in \Cref{subsec:extragrad_overview}, we need to implement some steps of form \begin{equation}\label{eq:best_resp_over_sset}
    \min_{\XX \in \sset} \langle \hh, \XX \rangle + \phi(\XX)
\end{equation} 
for implementing our best response oracle (\Cref{defn:apx_best_resp}).
Here $\phi(\XX)$ encodes a linear term and a regularizer term from the extragradient method, as well as $\psi(\XX)$, the given composite convex cost from \Cref{prob:comp_ell1inf_regr}.

To solve \eqref{eq:best_resp_over_sset},
via binary search (\Cref{lem:unconstrained_red}), this task can be reduced to minimizing the following objective to a small, say $1/\poly(m)$, additive error $\delta > 0$:
\begin{align}
\label{eq:overviewPQFlow}
    \min_{\XX \in \xset} \cE_{q, p}^{\coefnorm}(\XX) \defeq \langle h, \XX \rangle + \phi(\XX) + \coefnorm \norm{\XX}_{q, p}^{pq}\,,
\end{align}
for some value of $\coefnorm > 0$. 
Solving \eqref{eq:overviewPQFlow} up to $1/\poly(m)$ additive error is done by our algorithm for solving high-accuracy composite $\ell_{q, p}$-regression (\Cref{prob:comp_ellqp_regr}) to high-accuracy in $O_{q, p}(k)$-queries (see \Cref{informal:comp_ellqp_regr}). 
Without the $\coefnorm \|\XX\|_{q, p}^{pq}$ term, \eqref{eq:overviewPQFlow} can be easily solved in $k$ oracle calls, as the objective would be separable, so we can optimize over every $\xset_j$. However, the $\coefnorm \|\XX\|_{q, p}^{pq}$ term is not separable and makes solving \eqref{eq:overviewPQFlow} one of our main challenges. 

Our algorithm for solving \eqref{eq:overviewPQFlow} follows the iterative refinement framework from the $\ell_p$ regression literatures~\cite{akps19, AKPS22}.
At each iteration, the algorithm maintains a feasible point $\XX$ and tries to find a feasible update direction $\bDelta$ (meaning $\XX + \bDelta \in \xset$) so that $\cE_{q, p}^{\coefnorm}(\XX + \bDelta)$ is smaller. 
However, minimizing $\cE_{q, p}^{\coefnorm}(\XX + \bDelta)$ over all feasible updates is as hard as solving the original problem.
The idea of iterative refinement is to find a proxy objective $\cR(\bDelta; \XX)$ that approximates the residual $\cE_{q, p}^{\coefnorm}(\XX + \bDelta) - \cE_{q, p}^{\coefnorm}(\XX)$ and each iteration computes the update $\bDelta$ that only minimizes $\cR(\bDelta; \XX)$ approximately.
For our choice of $p = \Theta(\sqrt{\log m})$ and $q=1+1/p$, we design the proxy $\cR(\bDelta; \XX)$ so that $\almostTime(\alpha \log \frac{1}{\delta})$ iterations suffices to compute a $(1+\delta)$-approximate solution to \eqref{eq:overviewPQFlow} if at every iteration we compute an approximate solution of $(\alpha, \delta')$ multiplicative-additive error to minimizing $\cR$, i.e., $\cR\left(\algoutput{\bDelta}; \XX\right) \le \delta' + \frac{1}{\alpha}  \min_{\bDelta: \XX + \bDelta \in \xset}\cR(\bDelta; \XX)$, for $\delta' \le \frac{\delta}{3}$ (see \Cref{coro:convRate}). 
Additionally, we show that for $\delSCO = \frac{\delta}{\poly(k)}$, a solution of $(m^{o(1)}, O(k \delSCO))$ multiplicative-additive error to the proxy objective can be computed using $\almostTime(k)$ queries (\Cref{lem:approxResSolver}). 

Our design for the proxy $\cR(\bDelta; \XX)$ comes from approximating $\|\XX + \bDelta\|_{q, p}^{pq}$ effectively, as done in \cite{chen2023high}.
For any $p > 1$, \cite{AKPS22} shows that $|1+x|^p$ can be approximated by a term linear in $x$ plus an error term $\gamma_p(x)$ (\Cref{def:gamma}).
$\gamma_p(x)$ behaves like $|x|^2$ when $|x|$ is small and $|x|^p$ otherwise.
Similar to \cite{chen2023high}, we show, in \Cref{lem:qpIRLB} and \Cref{lem:qpIRUB},
\begin{align*}
    \norm{\XX + \bDelta}_{q, p}^{pq} - \norm{\XX}_{q, p}^{pq} - \l\gg, \bDelta\r \approx \sum_{i \in [m], j \in [k]} \hh_{ij} \gamma_q(\bDelta_{ij}; \XX_{ij}) + \sum_{i \in [m]} \left(\sum_{j \in [k]} \gamma_q(\bDelta_{ij}; \XX_{ij})\right)^p
\end{align*}
for $\gg \in \R^{m \times k}$ and $\hh \in \R^{m \times k}_{\ge 0}$ corresponding to the gradient and the Hessian of the function $\|\XX + \bDelta\|_{q, p}^{pq}$ respectively.

Via iterative refinement, \eqref{eq:overviewPQFlow} can be reduced to solving $m^{o(1)}\log \frac{1}{\delta}$ instances of the following problem to $(m^{o(1)}, \delta / 3)$ multiplicative-additive approximation:
\begin{align}
\label{eq:overviewResidualProb}
\min_{\bDelta \in \xset - \FF} \langle \gg, \bDelta\rangle + \phi(\XX + \bDelta) - \phi(\XX) + \sum_{i \in [m], j \in [k]} \hh_{ij} \gamma_q(\bDelta_{ij}; \XX_{ij}) + \sum_{i \in [m]} \left(\sum_{j \in [k]} \gamma_q(\bDelta_{ij}; \XX_{ij})\right)^p
\end{align}

To solve \eqref{eq:overviewResidualProb} approximately in $k$ SCO queries, which is an improvement over \cite{chen2023high} by a factor of $k$, we utilize a sequential block minimization algorithm.
The algorithm makes 1 SCO call for each of the $k$ blocks in a sequential manner.
In its application for the $\ell_{q,p}$ flow problem, this implies an $\almostTime(mk)$-time algorithm which improves on the previous $\almostTime(mk^2)$ runtime~\cite{chen2023high} and is key to our $\almostTime(mk\eps^{-1})$-time algorithm for approximate MCFs.

The sequential block minimization algorithm resembles the multiplicative weight update algorithm for the approximate \mcf{}s presented in \cite{grigoriadis1996approximate, radzik1996fast}, as it cycles through all $k$ commodities, minimizing a dynamic objective.

More precisely, starting from $[\mathbf{0}_m; \ldots; \mathbf{0}_m] \in \R^{m \times k}$
we cycle through all indices $j \in [k]$ in order one by one, at each step computing only one $\algoutput{\bDelta}_j$ at a time in order to minimize the proxy discussed in \Cref{subsec:high_acc_ellpq_overview}. 
At step $j$, we compute $\algoutput{\bDelta}_j$ that minimizes the objective based on the $j-1$ updates we have so far. 
Suppose we write the objective of \eqref{eq:overviewResidualProb} as $\cR([\bDelta_1; \ldots; \bDelta_k])$, $\algoutput{\bDelta}_j$ minimizes the following problem:
\begin{align*}
    \algoutput{\bDelta}_j \in \argmin_{\bDelta_j: \XX_j + \bDelta_j \in \xset_j} \cR(\left[\bDelta_1; \ldots; \bDelta_{j-1}; \bDelta_j; \mathbf{0}_m; \ldots; \mathbf{0}_m\right])
\end{align*}
As we show in \Cref{sec:comp_ellqp_regr}, 
$\algoutput{\bDelta}_j$ can be computed by calling the $(\Time_j, \delSCO)$-SCO for set $\xset_j$ and function $\psi_j$. Hence, our sequential block minimization procedure runs in $k$ SCO calls and $\almostTime(mk)$ additional time. 
In \Cref{sec:resP}, we show that the final output $\algoutput{\bDelta}$ of this procedure is an approximate solution of $(m^{o(1)}, O(k \delSCO))$ multiplicative-additive error (see \Cref{lem:approxResSolver}).

When applying our generic $\ell_{1,\infty}$-regression algorithm to \mcf{} problems, we implement the SCOs using \Cref{coro:convexFlow}.
For each commodity $j \in [k]$, the objective function corresponding to
the residual problem \eqref{eq:overviewResidualProb} is $\psi_j(\XX_j) + \sum_{i \in [m]} c_{ij}(\XX_{ij})$, for some $m$-decomposable costs $c_{ij}$. Since the expression $\sum_{i \in [m]} c_{ij}(\XX_{ij})$ is
separable over the edges, we are able to implement a $(\almostTime(m), \delSCO)$-SCO for commodity $j$, for $\delSCO = 1/\poly(m)$, via the algorithm given by \Cref{coro:convexFlow}.

\subsubsection{Discussion: Obstacles for Applying \texorpdfstring{\cite{sherman2017area}}{Sherman17}}
\label{subsec:obstacle}

One may be tempted to recover our results by more directly extending the results of \cite{sherman2017area}. 
However, there are several obstacles which the preceding approach overcame.

A straightforward application of the approach in \cite{sherman2017area} requires what is known as an $m^{o(1)}$-competitive linear oblivious routing scheme. In in undirected graphs, $\O(1)$-competitive linear oblivious routing exists and can be computed in $\O(m)$ time~\cite{peng16,li2025congestion}. However, unfortunately, for directed graphs such schemes do not exist; there are directed graphs for which any oblivious routing scheme has a competitive ratio of at least $\Omega(\sqrt{n})$~\cite{ACF2003}.

Another approach combines Sherman's numerical outer loop (Theorem 1.3 in \cite{sherman2017area}) with its area-convex regularizer and the almost linear-time single-commodity flow algorithm to solve the concurrent \mcf{} problem $\min_{\im(\FF) = \DD, \FF \ge 0} \congest(\FF).$ The number of iterations for the outer loop would be $\O(k \eps^{-1})$, with each iteration solving $k$ single-commodity flow problems. 
As discussed in the beginning of \Cref{sec:short_overview}, this results in an $\almostTime(mk^2 \eps^{-1})$-time algorithm, leading to a quadratic dependence on $k$ in the runtime.
Unlike the undirected case, where the iteration count is $\O(\eps^{-1})$, an additional factor of $k$ is required for directed graphs because the $k$ flows computed in each iteration can collectively have congestion as large as $\Omega(k)$. Additionally, the $\ell_{\infty}$-operator norm of $\congest(\cdot)$ is $k$ when viewed as a linear operator. These two factors necessitate increasing the range of the regularizer to $\O(k)$ in order to preserve \emph{area-convexity}. In contrast, in the undirected case, the congestion of the flows computed in each iteration is $\O(1)$ due to the existence of $\O(1)$-competitive oblivious routing.

In our work, we resolve these issues with an efficient $\ell_{q,p}$-norm flow algorithm for directed graphs, which, in each iteration, finds a \mcf{} with congestion as small as $m^{o(1)}$. We also design a new area-convex regularizer, whose range depends only on the congestion of the flows computed in each iteration, rather than on the $\ell_{\infty}$-operator norm. 
Additionally, we develop a more general optimization outer loop for solving a broader class of \mcf{} problems while Sherman's approach handles only the concurrent \mcf{} problem.

\subsection{Related Work}
\label{sec:related_work}

\paragraph{\Mcf{}s.}

\Mcf{} problems have been studied extensively and have numerous applications~\cite{kennington1978survey,ahuja1988network,ouorou2000survey,barnhart2009multicommodity,wang2018multicommodity}.
It is known to be equivalent to linear programming in certain computational models and there have been advances using linear programming algorithms~\cite{itai1978two, dkz22,khachiyan1980polynomial,karmarkar1984new,renegar1988polynomial,cls21}.
Recently, \cite{bz23} gives evidence that graph structure could help with an $\O(k^{2.5} \sqrt{m} n^{\omega - 1/2})$-time algorithm.
This is faster than state-of-the-art general LP solvers when the graph is not sparse, i.e., $m = n^{1+\Omega(1)}.$

Many prior works focus on computing $(1+\eps)$-approximate solutions for MCF problems~\cite{shahrokhi1990maximum,lsmtpt91,goldberg1992natural,grigoriadis1994fast,klein1994faster,karger1995adding,plotkin1995fast,Young95,grigoriadis1996approximate,radzik1996fast,shmoys1997cut,Fleischer00,Kar02,GargK07,Madry10}.
In particular, \cite{radzik1996fast} shows that $(1+\eps)$-approximate maximum concurrent \mcf{}s can be computed using $\O(k\eps^{-2})$ min-cost flow computations.
\cite{grigoriadis1996approximate} reduces $(1+\eps)$-approximate \emph{minimum-cost \mcf{}s} to $\O(k\eps^{-2})$ min-cost flow computations.
Combining with the recent development of min-cost flow algorithms, this yields a deterministic $\almostTime(mk\eps^{-2})$-time algorithm for this problem~\cite{ls19, bllsssw21, CKL+22, van2023deterministic}.
Another line of work focuses on obtaining algorithms whose runtime dependency on $\eps$ is better than $\eps^{-2}.$
\cite{bienstock2004solving} and \cite{Nesterov09} give runtime depending on $O(\eps^{-1} \log \eps^{-1})$ and $O(\eps^{-1})$ respectively although the runtimes are at least quadratic in $m$.

When the input graph is undirected, there are faster approximate MCF algorithms that do not require min-cost flow computations~\cite{KMP12, klos14, sherman2017area}.
In particular, \cite{sherman2017area} introduced the idea of \emph{area convexity} and obtained the first $\O(mk\eps^{-1})$-time max concurrent flow algorithm.
If $\omega(1)$-approximation ratio is allowed, there are algorithms that run in $\almostTime(m+k)$-time via dynamic shortest path data structures~\cite{chuzhoy2021decremental, haeupler2024low}. 
Additionally, if the graph is an expander, \cite{li2025local} provided a $O(m + \epsilon^{-3} k^3 D) n^{o(1)}$ algorithm for computing a $(1+\epsilon)$-approximate solution for ``source-sink'' demands of form $\dd_j = \vecone_{s_j} - \vecone_{t_j}$ (for pairs of vertices $(s_j, t_j)$), where $D = \sum_{j \in [k]} \|\dd_j\|_{\infty}$ is the total demand. 

\paragraph{Accelerated cyclic block coordinate methods.}
To implement the steps of our extragradient method to high accuracy, we use a cyclic block coordinate algorithm, as described in \Cref{sec:short_overview}. In some sense, our approach is an instance of boosting cyclic block coordinate methods to obtain faster runtimes. 
In tandem, there 
has been recent work on designing cyclic coordinate methods for convex optimization \cite{song2023cyclic,lin2023accelerated}. These methods work by partitioning the set of coordinates into blocks and cyclically updating each block. 
Furthermore, \cite{lin2023accelerated} has provided an accelerated cyclic block coordinate algorithm for minimizing convex functions that are smooth with respect to the $\ell_2$ norm. 
Their algorithm works by using an extragradient update rule to set the new value of a given block of coordinates. 
By contrast, our method uses an extragradient method in the outer loop, while the cyclic block coordinate method is used as an inner loop to implement the extragradient update steps.
Unfortunately, directly applying
the result of \cite{lin2023accelerated} to \mcf{} problems is prohibitively expensive. In particular, any objective function of $h(\XX)$ that approximates $\|\XX\|_{1, \infty}$ (which corresponds to the objective for concurrent \mcf{}) well and is $\ell_2$-smooth will have a smoothness coefficient that is $\Omega(\sqrt{m})$. Since the algorithm in \cite{lin2023accelerated} has a runtime depends linearly on the smoothness coefficient, naively applying it to concurrent \mcf{} would result in a runtime of $\Omega(m^{1.5} / \epsilon)$. 

\paragraph{$\ell_p$-regression.}

The problem of $\ell_p$-regression asks for $\xx \in \R^d$ minimizing $\|\AA \xx - \bb\|_p$ for some $p > 0$, $\AA \in \R^{n \times d}$ and $\bb \in \R^n.$
This interpolates between linear regression ($p = 2$) and linear program ($p \in \{1, \infty\}$).
Low-accuracy algorithms have been studied for over-constrained problems where $n = \omega(d)$ and have achieved runtime $O_p(\nnz(\AA) + \poly(d, \eps^{-1}))$~\cite{dasgupta2009sampling,meng2013robust,woodruff2013subspace,clarkson2016fast,clarkson2017low,clarkson2019dimensionality}.

There is also a line of work on high-order methods for finding high-accuracy solutions.
\cite{ls19} developed interior point method for convex optimization which can be used to achieved $\O(\sqrt{\min\{n, d\}})$ iterations for $p \in \{1,\infty\}$ \cite{LeeS15}.
\cite{bubeck2018homotopy} proposed homotopy methods whose iteration count is $\O(n^{|\frac{1}{2}-\frac{1}{p}|}).$
The iterative refinement framework, which is used in our algorithm, can solve $\ell_p$ regression via solving $\O_p(d^{(p - 2) / (3p - 2)})$ linear systems for $p \ge 2$~\cite{akps19,APS19,AdilS20,GPV21,AKPS22,jambulapati2022improved,jambulapati2023sparsifying}.

A class of $\ell_p$-regression problems, named $\ell_p$ flows, has been studied extensively and has been used to design faster algorithms for problems such as maxflows and min-cost flows~\cite{ls20, AMV20}.
$\ell_p$ flows interpolate between transshipment ($p=1$), Laplacian systems ($p=2$), and maxflows ($p=\infty$).
\cite{SpielmanTengSolver:journal} gave the first near-linear time algorithm for the case $p=2$ on undirected graphs.
For $p \in [\omega(1), o(\log m)]$, \cite{kpsw19} gave the first almost-linear time algorithm on unweighted undirected graphs.
\cite{CKL+22} gave the first almost-linear time algorithm for any $p \ge 1$ on weighted directed graphs.
Our algorithm is based on solving its multi-commodity variation, the $\ell_{q, p}$ flow, which was proposed and was solved to high-accuracy in $\almostTime_{p, q}(mk^2)$ time in \cite{chen2023high}.

\subsection{Paper Organization.}
In \Cref{sec:all_of_4}, we present and prove our general framework for solving composite bilinear optimization problems.
We state and show the iterative refinement framework for solving composite $\ell_{q, p}$-regression problems in \Cref{sec:comp_ellqp_regr}.
\Cref{sec:L1infMin} concludes our algorithm framework for approximately solving composite $\ell_{1, \infty}$-regression problems and proves our $\almostTime(mk\eps^{-1})$-time \mcf{} algorithms. 
In \Cref{apx:implementing_steps}, we present our extension of the framework in \cite{JT23} for solving constrained box-simplex games. In \Cref{subsec:unconstrained_red}, we prove the results regarding the binary search procedure that is key to reducing the steps for solving $\ell_{1, \infty}$ regression problem to solving $\ell_{q, p}$ regressions. 
Finally, in \Cref{sec:proving_computability}, we prove computability and decomposability of the functions we work with.

\section{Preliminaries}
\label{sec:prelim}

\paragraph{General notation.}
We denote vectors as lower case bold letters $\xx$, and matrices as upper case bold letters $\XX.$
When a vector $\xx \in \R^d$ is clear from the context, we write $\XX$ to denote $\ddiag(\xx)$, the diagonal matrix whose $i$-th element on the diagonal is $\xx_i.$
We also write $\ddiag(\MM)$ for a matrix $\MM$ to denote the matrix obtained from taking the elements from $\MM$'s diagonal and putting $0$'s everywhere else.
We define $\vecone^n$ to be the $n$-dimensional vector with all entries equal to $1$ and $\vecone_i$ to be the vector where the $i$-th component is $1$ and all others are zeros.
Given two vectors $\xx, \yy$, we write $\l\xx, \yy\r$ to denote their inner dot product. 
Additionally, given two matrices $\XX$ and $\YY$, we write $\l\XX, \YY\r$ to denote $\l\vvec(\XX), \vvec(\YY)\r$ where $\vvec(\XX)$ flattens the matrix into a vector, meaning that it returns the vector formed by concatenating the columns of $\XX$. 
We denote by $\nnz(\XX)$ the number of non-zero entries of matrix $\XX$. 

We define the $m$-dimensional simplex $\Delta^m = \{\yy \in \R^m_{\ge 0}: \sum_{i \in [m]} \yy_i = 1\}$.
For $\AA \in \R^{n \times d}$ and $p > 0$, we define $\|\AA\|_{p} = \max_{\|\xx\|_p \le 1} \|\AA \xx\|_p$. Consequently, $\|\AA\|_{\infty} = \max_i \sum_j |\AA_{ij}|.$
For $P \subseteq \R^n$, we denote $\mathrm{diam}(P) \defeq \argmax_{\xx \in P} \|\xx\|_{\infty}$. 
For $q, p \ge 1$ and $\XX \in \R^{n\times d}$, we define the $\ell_{q, p}$ norm of $\XX$ as $\|\XX\|_{q, p} = (\sum_i (\sum_j |\XX_{ij}|^q)^{p})^{\frac{1}{pq}}.$
Consequently, $\|\XX\|_{1, \infty}$ is $\max_i \sum_j |\XX_{ij}|.$

Given a differentiable function $f$ and any two points in its domain $\xx, \yy \in \mathrm{dom}(f)$, we denote $V^f_{\xx}(\yy)$ to be its Bregman divergence from $\xx$ to $\yy$, i.e., $V^f_{\xx}(\yy) = f(\yy) - \l\g f(\xx), \yy-\xx\r - f(\xx).$

\paragraph{Graphs and flows.}
In this paper, we work with directed graphs $G = (V, E)$ with $n$ vertices, $m$ edges, and capacity vectors $\uu \in \R^{E}$.
A flow $\ff \in \R^E$ is a vector indexed by edges of $G$ and 
a vertex demand vector $\dd \in \R^V$ is a vector indexed by the vertices.
We define the vector of node imbalances $\imbal(\ff)$ as:
\begin{align*}
    \imbal(\ff)_a \defeq \sum_{e = (a, b) \in E} \ff_e - \sum_{e = (b, a) \in E} \ff_e, \forall a \in V
\end{align*}
We say a flow $\ff$ routes a demand $\dd$ if $\imbal(\ff) = \dd.$
Given edge capacities $\uu \in \R^{E}_{\ge 0}$, a flow $\ff$ is feasible if $0 \le \ff_e \le \uu_e$ holds for any edge.
A \mcf{} $\FF = \{\ff_1, \ldots, \ff_k\} \in \R^{E \times k}$ is a collection of $k$ flows and its congestion is defined as
\begin{align*}
    \congest(\FF) \defeq \max_{e \in E} \frac{\sum_{i \in [k]} |\ff_{ie}|}{\uu_e} = \|\UU^{-1}\FF\|_{1, \infty}
\end{align*}

\paragraph{Computable and decomposable functions.} Throughout the paper, we assume that we can minimize to high-accuracy a broad class of convex costs on the flow values $\ff_e$ on the edges of a graph, which we call $m$-decomposable (and define formally in this section). A more restrictive class of functions that was used in \cite{CKL+22} is that of \textit{computable functions}, formally defined below. 

\begin{definition}[Computable Functions]\label{def:computable}
    Given a single variable convex function $c: \R \to \R \cup \{+\infty\}$, $c$ is $(m, K)$-\emph{computable} if there exists a barrier function $\barrier_c(x, t)$ defined on the domain $\dset_c \defeq \{(x, t) |~ c(x) \le t\}$ such that
\begin{enumerate}
\item\label{item:costQuasiPoly}
The function value is quasi-polynomially bounded, i.e., $|c(x)| = O(m^K + |x|^K)$ for all $x \in \R.$ 
\item\label{item:costSC}
$\barrier_c$ is a \emph{$\nu$-self-concordant barrier} for some $\nu \le K$, that is, the following holds
\begin{align*}
    \barrier_c(\uu) &\to \infty, \text{ as $\uu$ approaches the boundary of $\dset_c$} \\
    \left|\g^3 \barrier_c(\uu) [\vv, \vv, \vv]\right| &\le 2 \left(\g^2 \barrier_c(\uu) [\vv, \vv]\right)^{3/2}, \forall \uu \in \dset_c, \vv \in \R^2 \\
    \l\g \barrier_c(\uu), \vv\r^2 &\le \nu \cdot \g^2 \barrier_c(\uu) [\vv, \vv]
\end{align*}
\item\label{item:costHessian}
The Hessian is quasi-polynomially bounded as long as the function value is $\O(1)$ bounded, i.e., for all points $|x|, |t|\le m^K$ with $\barrier_c(x, t)\le \O(1)$, we have $\g^2 \barrier_c(x, t) \preceq \exp(\log^{O(1)} m) \II$. 
\item\label{item:costHessianCompute}
Both $\g \barrier_c$ and $\g^2 \barrier_c$ can be computed and accessed in $\O(1)$-time.
\end{enumerate}
\end{definition}

Note that \Cref{def:computable} includes functions such as $x^p$ for $p \ge 1$, $x \log x$, and indicator functions. 
The definition of $(m, K)$-computable functions is equivalent to the one in \cite{CKL+22} (see Assumption 10.2). 
To obtain our results for \mcf{} problems, we work with a broader class of functions, which we call \textit{self-concordant decomposable}, or, in short, \textit{decomposable}. The building blocks of decomposable functions are computable functions. The precise definition is provided below. 

\begin{definition}[Decomposable Functions]
\label{def:decomposable}
Let $m \in \Z_{>0}$. 
A single variable convex $c: \R \to \R \cup \{+\infty\}$ is $m$-self-concordant decomposable ($m$-decomposable) if it can be written as $c(x) = \sum_{\ell \in [O(1)]} c_{i}(x)$ such that each $c_{i}: \R \to \R \cup \{+\infty\}$ can be written, for $0\le\lowerend^{(i)}\le\upperend^{(i)}$, as $c_{i}(x) = \min_{a+b = x, a \in [\lowerend^{(i)}, \upperend^{(i)}]} c_{i, 1}(a) + c_{i, 2}(b)$ with both $c_{i, 1}, c_{i, 2}$ being $(m, O(\log^{O(1)} m))$-computable. 
\end{definition}

Note that every $(m, \log^{O(1)})$-computable function $c$ is also $m$-decomposable, as we can write, 
for $c_1(x) = 0 \forall x, \lowerend = \upperend = 0, c_2(x) := c(x), \forall x$, 
$c(x) = \min_{a+b = x, a \in [\lowerend, \upperend]} c_1(x) + c_2(x)$. 
Additionally, note that if $c(x)$ is $(m, \otilde(1))$-computable, $c(\vartheta x)$ is also $(m, \otilde(1))$-computable computable for $\vartheta \ge 0, \vartheta = \otilde(1)$. 
Consequently, if $c(x)$ is $m$-decomposable, $c(\vartheta x)$ is also $m$-decomposable for $\vartheta \ge 0, \vartheta = \otilde(1)$.

As mentioned in \Cref{sec:intro}, \cite{CKL+22, van2023deterministic} gave an almost-linear time algorithm for finding a flow that routes a given demand and minimizes a separable convex cost. The formal theorem is stated below. 
\begin{theorem}[Theorem 10.13, \cite{CKL+22}]
\label{thm:convexFlow}
Let $\eta = \exp(-\log^{O(1)} m)$ be a granularity parameter and 
$G$ a graph with $m$ edges, capacity vector $\uu \in \eta\Z^E_{>0}$ and demand vector $\dd \in \eta\Z^V$.
Given $(m, \log^{O(1)} m)$-computable costs $\{c_e(\cdot)\}_e$ on edges and $C > 0$, there is an algorithm that runs in $\almostTime(m)$ time and outputs a flow $\algoutput{\ff} \in \R^E$ with $\imbal(\algoutput{\ff}) = \dd$ and, 
\begin{align*}
    c(\algoutput{\ff}) \le \min_{\imbal(\ff^{\star}) = \dd, 0 \le \ff^{\star} \le \uu} c(\ff^{\star}) + \exp(-\log^C m)
        \text{ where }
    c(\ff) \defeq \sum_{e\in E} c_e(\ff_e)\,.
\end{align*}
\end{theorem}

We apply \Cref{thm:convexFlow} to obtain a class of single commodity flow optimization results more easily applicable for our purposes, \Cref{coro:convexFlow} below. 

\begin{corollary}
\label{coro:convexFlow}
Let $\eta = \exp(-\log^{O(1)} m)$ be a granularity parameter and 
$G$ a graph with $m$ edges, capacity vector $\uu \in \eta\Z^E_{>0}$ and demand vector $\dd \in \eta\Z^V$.
Given $0 \le \RdLower \le \RdUpper \le \poly(m)$, fixed constant $C>0$, $m$-decomposable costs $\{c_e(\cdot)\}_e$ on edges, 
and $\poly(m)$-Lipschitz $v: \R \to \R$, there is an algorithm that, in $\almostTime(m)$ time, outputs a flow $\algoutput{\ff} \in \R^E$ and $\algoutput{\beta} \in \eta \Z_{> 0}$ so that $\imbal(\algoutput{\ff}) = \algoutput{\beta} \dd$ and
\begin{align*}
    c(\algoutput{\ff}, \algoutput{\beta}) \le \min_{\beta^* \in [\RdLower, \RdUpper], \imbal(\ff^{\star}) = \beta^{\star} \dd, 0 \le \ff^{\star} \le \uu} c(\ff^{\star}, \beta^{\star}) + \exp(-\log^C m)
        \text{ where }
    c(\ff, \beta) \defeq \sum_{e\in E} c_e(\ff_e) + v(\beta)\,.
\end{align*}
\end{corollary}

\Cref{coro:convexFlow} is used to solve \Cref{prob:convex_mcflow}, as for we are interested in finding a flow that routes $\beta \dd$ for some $\beta \in [\RdLower, \RdUpper]$ and minimizes the sum of the total cost on edges and the cost $v(\beta)$.  
We include a proof of \Cref{coro:convexFlow} in \Cref{subsec:single_comm_unconstrained}.

\section{Composite Constrained Box-Simplex Games}
\label{sec:all_of_4}

In this section, we present our framework for solving constrained box-simplex games with a composite function. Recall the setting of \Cref{prob:comp_ell1inf_regr}, where our goal is to solve \eqref{eqprob:comp_ell1inf_reform}. 
As mentioned in \Cref{subsec:extragrad_overview}, our approach is to instead solve the problem 
\begin{equation}\label{eq:prob_constr_minmax}
    \min_{\XX \in \sset} \max_{\yy \in \Delta^m} \yy^\top \AA_{1, \infty} \vvec(\XX) + \psi(\XX)\,,
\end{equation}
where $\sset$ is a smaller subset of $\xset$, the set of feasible $\XX$, and $\AA_{1, \infty}$ is the matrix with $(\AA_{1, \infty})_{i,\ell} = 1$ if and only if $\ell \mod k = i$,
which has the property that $\max_{\yy \in \Delta^m} \yy^\top \AA_{1, \infty} \vvec(\XX) = \|\XX\|_{1,\infty}, \forall \allowbreak \XX \in \xset$. 
Here, we provide a framework that solves the following more general problem:
\begin{equation}\label{eq:prob_u_v_formulation}
    \min_{\uu \in \uset} \max_{\vv \in \vset} \vv^\top \AA \uu + \psi(\uu)\,,
\end{equation}
where $\AA \in \R^{m \times n}$, $\uset \subseteq \R_{\ge 0}^n$ and $\vset = \Delta^m$. 
For $\uset \subseteq \R_{\ge 0}^n$ and a fixed matrix $\AA$, we call Problem \ref{eq:prob_u_v_formulation} a \textit{constrained, composite box-simplex game} $(\uset, \Delta^m, \AA)$. 
To approximately solve \eqref{eq:prob_u_v_formulation}, we apply an extragradient method framework that generalizes that of \cite{JT23}, which we develop throughout this section. 

To implement our framework, we require the ability to approximately minimize over the constraint set $\uset$ functions of $\uu \in \uset$ that involve the composite term $\psi$, a linear term and the regularizer used for implementing the composite proximal gradient steps. Specifically, we work with an \textit{approximate composite best response oracle} (which we abbreviate as ABRO), whose precise definition is provided below. 

\begin{definition}[Approximate Composite Best Response Oracle]\label{defn:apx_best_resp}
    Let $\uset \subseteq \R_{\ge 0}^n, \vset = \Delta^m$,
    $\psi: \uset \to \R, r:\uset \times \Delta^m \to \R$ convex, and $\delta > 0$. 
    A $\delta$-approximate (composite) best response oracle ($\delta$-ABRO) with respect to tuple $(\uset, \psi, r)$, $\oracle_{\bestrsp}^{\uset, \psi, r}$, takes as input a point $\vv \in \Delta^m$, vector $\hh^{\uset} \in \R^n$ and constant $\BRconstant \in [1/2, 1]$ and outputs $\algoutput{\uu} \in \uset, \algoutput{\uu} = \oracle_{\bestrsp}^{\uset, \psi, r}(\vv, \gg, \BRconstant)$ with the following properties: 
     \begin{enumerate}
        \item $\langle \hh^{\uset} + \nabla_{\uset} r(\algoutput{\uu}, \vv), \algoutput{\uu}-\uu \rangle \le \BRconstant (\psi(\uu) - \psi(\algoutput{\uu})) + \delta, \forall \uu \in \uset.$
        \item $\|\nabla_{\vset} r(\uu^{\star}, \vv) - \nabla_{\vset} r(\algoutput{\uu}, \vv)\|_{\infty} \le \frac{\delta}{2}$, where $\uu^{\star} = \argmin_{\uu \in \uset} \langle \hh^{\uset}, \uu \rangle + r(\uu, \vv) + \psi(\uu)$. 
    \end{enumerate}
\end{definition}
It is straightforward to show that the point outputted by 
such an oracle is a solution of additive error $\delta$ for the following optimization problem:
\begin{equation}\label{prob:prox_composite_step}
    \min_{\uu \in \uset}  \langle \hh^{\uset}, \uu \rangle + r(\uu, \vv) + \psi(\uu)\,. 
\end{equation} 
The regularizer $r$ we will use for solving \eqref{prob:prox_composite_step} is the doubly-entropic regularizer in \Cref{def:regularizer}. 
Given access to such an oracle, as well as exact minimization over $\Delta^m$ for tasks of form $\min_{\vv \in \Delta^m} \langle \hh^{\vset}, \allowbreak \vv \rangle + \sum_{i \in [m]} \vv_i \log \vv_i$, for any $\hh^{\vset} \in \R^m$, we can efficiently approximately solve \eqref{eq:prob_u_v_formulation}, as stated below. 

\begin{theorem}\label{thm:box_simplex_solver}
    Let $(\uset, \Delta^m, \AA)$ be a constrained box-simplex game, $\linfbound \ge \max_{\uu \in \uset} \||\AA| \uu\|_{\infty}$, $\usetBound \ge \usetSize$, and
    $r = r^{\alpha, \xi}$ be the doubly entropic regularizer in \Cref{def:regularizer} with parameters $\xi = \min(\frac{1}{m}, \frac{\rho}{m n \|\AA\|_{\infty}})$ and $\alpha = 16 \linfbound B$, where $B = 1 + \log (\max(\frac{1}{\xi}, \usetBound + \xi))$. 
    Assuming access to a $\delBR$-ABRO $\oracle_{\bestrsp}^{\uset, \psi, r}$ with respect to $(\uset, \psi, r)$, 
    in $\otilde((\nnz(\AA) + m+n) B \epsilon^{-1})$ time and $\otilde(B \epsilon^{-1})$ calls to $\oracle_{\bestrsp}^{\uset, \psi, r}$ taking as input $(\vv, \hh^{\uset})$ with $\|\hh^{\uset}\|_{\infty} \le 3 n B \|\AA\|_{\infty}$, $\Cref{alg:box_simplex}$ outputs $(\algoutput{\uu}, \algoutput{\vv})$ with 
    \[\langle \vv, \AA \algoutput{\uu} \rangle + \psi(\algoutput{\uu}) \le \langle \algoutput{\vv}, \AA \uu \rangle + \psi(\uu) + \epsilon \linfbound + 6 \delBR, \forall \uu \in \uset, \vv \in \vset \,.\]
\end{theorem}

We defer the proof of \Cref{thm:box_simplex_solver} to \Cref{subsec:3_lastsubsec}. 
The rest of the section is organized as follows. 
In \Cref{sec:extragrad_framework}, we present a framework that we use to solve \textit{variational inequalities} with composite terms, which generalizes Problem \ref{eq:prob_u_v_formulation}. 
The algorithm for approximately solving variational inequalities, 
\Cref{alg:main}, 
is an extragradient method and it implements its \textit{gradient} and \textit{extragradient} steps approximately. The formal notions of approximately implementing these steps are defined in \Cref{sec:extragrad_framework}, while their implementability, under a series of assumptions, is discussed in \Cref{subsec:implement_oracles}. 
In \Cref{sec:regularizer}, we present our doubly-entropic regularizer, introduced in \Cref{subsec:extragrad_overview}, along with the technical lemmas regarding its properties. Finally, in \Cref{subsec:3_lastsubsec}, we present the algorithm that obtains the guarantee in \Cref{thm:box_simplex_solver}. This algorithm is an instantiation of \Cref{alg:main}, 
using the regularizer in \Cref{sec:regularizer}, which implements the gradient and extragradient steps needed following the results of \Cref{subsec:implement_oracles}. 

\subsection{Composite Bilinear Optimization}
\label{sec:extragrad_framework}
In this subsection, we present our framework for solving the composite variational inequality problem (\Cref{prob:generic_regret}). Before presenting its formal statement, we start with some notation. For vectors $\zz, \ww \in \zset$, where $\zset$ is a convex set with $\zset \subseteq \R^N$ for some $N \in \Z_{>0}$ and $\hh \in \R^N$, and function $\psi:\zset \to \R$, we define \[
\mathrm{regret}^{\zset}_{\hh, \psi}(\zz; \ww) = \hh^{\top}(\zz - \ww) + \psi(\zz) - \psi(\ww),
\]
and  
\[\mathrm{regret}^{\zset}_{\hh, \psi}(\zz) = \sup_{\ww \in \zset} \mathrm{regret}^{\zset}_{\hh, \psi}(\zz; \ww)\,.\]

Using the notation above, we now provide the formal definition of \Cref{prob:generic_regret}. 
\begin{problem}[Composite Variational Inequality]\label{prob:generic_regret}
    Let $\zset \subseteq \R^{N}$ be a convex set,  
    $\gg:\zset \to \zset^*$ be a gradient operator and $\psi:\zset \to \R$ be a convex function. Given a target accuracy $\epsilon$, our task is to find a sequence of points $\zz^{(1)},\ldots,\zz^{(T)}$ such that 
\[\frac{1}{T} \sum_{t \in [T]} \mathrm{regret}^{\zset}_{\gg(\zz^{(t)}), \psi}(\zz^{(t)}; \ww) \le \epsilon, \forall \ww \in \zset \,.\]
\end{problem}

Hence, the goal of \Cref{prob:generic_regret} is to find a sequence of points $\zz^{(1)},\ldots,\zz^{(T)}$ such that 
\[
\frac{1}{T} \sum_{t \in [T]} \left[\langle \gg(\zz^{(t)}) , \zz^{(t)} - \ww \rangle + \psi(\zz^{(t)}) - \psi(\ww)\right] \le \epsilon, \forall \ww \in \zset\,. 
\] 
The motivation for solving \Cref{prob:generic_regret} is that if the gradient operator $\gg$ is bilinear (meaning $\gg(\sum_{t} \alpha_t \zz^{(t)}) = \sum_{t} \alpha_t \gg(\zz^{(t)})$ for any $\{\zz^{(t)}\}_t$ in $\zset$ and $\alpha_t$ scalars), then solving \Cref{prob:generic_regret} up to target accuracy $\epsilon$ yields the point $\Bar{\zz} \in \zset, \Bar{\zz} = \frac{1}{T} \sum_{t \in [T]} \zz^{(t)}$ so that $\mathrm{regret}^{\zset}_{\gg(\Bar{\zz}), \psi}(\Bar{\zz}) \le \epsilon$(see proof of \Cref{lem:convRate}).

In the context of \Cref{thm:box_simplex_solver}, $\zset = \uset \times \Delta^m$, where $\uset \subseteq \R^n$, and the composite function $\psi$ only depends on the $\uset$ variable (i.e., for every $\uu \in \uset, \vv \in \Delta^m$, the value $\psi(\uu, \vv)$ depends only on the value of $\uu$). 
For some matrix $\AA \in \R^{m \times n}$, 
the gradient operator $\gg:\uset \times \Delta^m \to \R^{n + m}$ is defined as $\gg(\uu, \vv) := [\AA^\top \vv; -\AA\uu], \forall \uu \in \uset, \vv \in \Delta^m$.

To obtain an $\epsilon$-approximate solution to \Cref{prob:generic_regret}, we design an extragradient method (\Cref{alg:main}) which generates iterates $\fullit^{(1)}, \ldots \fullit^{(T)}$ and $\halfit^{(1)}, \ldots \halfit^{(T)}$ (in the extragradient method literature $\fullit^{(t)}$'s are called \textit{full iterates} and $\halfit^{(t)}$'s are called \textit{half iterates}) in the order $\halfit^{(t)}, \fullit^{(t+1)}, \halfit^{(t+1)}$, etc, and use the collection $\{\halfit^{(1)}, \ldots \halfit^{(T)}\}$ for \Cref{prob:generic_regret}. 
Specifically, each step of the algorithm consists of approximately solving two composite proximal gradient steps with respect to a regularizer $r:\zset \to \R$ (which we refer to as the \textit{gradient step} and the \textit{extragradient step}). 

\Cref{alg:main} is converges under the assumption that the regularizer $r$ used is \textit{relaxed relative Lipschitz} with respect to the gradient operator $\gg$ (see \Cref{def:relaxed_RelLip}). It also assumes that we can approximately implement the composite gradient step and the extragradient step (see \Cref{defn:gen_apx_grad} and \Cref{defn:gen_apx_exgrad}). 
In the rest of the section, we present the definition of relaxed relative Lipschitzness, as well as the steps mentioned, and then provide the pseudocode of \Cref{alg:main}, along with its convergence guarantee (\Cref{lem:convergence}) and its proof. 

\begin{definition}[Relaxed Relative Lipschitzness, Definition 1 from \cite{JT23}]\label{def:relaxed_RelLip}
    Let $\zset \subseteq \R^{N}$ be a convex set.
    We say that $r:\zset \to \R$ is $\eta$-relaxed relative Lipschitz with respect to a gradient 
    operator $\gg:\zset \to \zset^*$ if for all $(\zz^{(1)}, \zz^{(2)}, \zz^{(3)}) \in \zset \times \zset \times \zset$, we have 
    \[
    \eta \langle \gg(\zz^{(2)}) - \gg(\zz^{(1)}), \zz^{(2)} - \zz^{(3)} \rangle \le V_{\zz^{(1)}}^{r}(\zz^{(2)}) + V_{\zz^{(2)}}^{r}(\zz^{(3)}) + V_{\zz^{(1)}}^{r}(\zz^{(3)})\,.
    \]
\end{definition}
Throughout the rest of the section,
we say $r:\zset \to \R$ is \textit{relaxed relative Lipschitz} with respect to a gradient operator $\gg$ if it is $1/3$-relaxed relative Lipschitz with respect to $\gg$. As we 
will work with Bregman divergence with respect to some convex function $r:\zset \to \R$, we define the following quantity which we use throughout the section:
\begin{equation}\label{def_3_breg_opt}
    \Delta^{r}_{\zz}(\zz',\ww) \defeq V_{\zz}^{r}(\ww) - V_{\zz'}^{r}(\ww) - V_{\zz}^{r}(\zz') = \langle \nabla V^{r}_{\zz}(\zz'), \ww - \zz' \rangle = \langle \nabla r(\zz') - \nabla r(\zz), \ww - \zz' \rangle \,.
\end{equation}
We next provide precise definitions for the notions of approximately solving the gradient and extragradient steps. 

\begin{definition}[Composite Gradient Step]\label{defn:gen_apx_grad}
    Let $\zset \subseteq \R^{N}$ be convex, compact 
    and $\psi, r, s: \zset \to \R$ be convex functions. 
    We say $\zz' = \oracle_{\mathrm{grad}}^{\psi, r, s}(\zz, \hh)$ is a $\delta$-approximate composite gradient step ($\delta$-CGS) on input $(\zz, \hh)$ with respect to function $\psi$ and regularizers $r, s$, if  
    \[\mathrm{regret}^{\zset}_{\hh, \psi}(\zz'; \ww) \le \Delta^{r}_{\zz}(\zz',\ww) + V_{\zz}^{s}(\ww) + \delta, \forall \ww \in \zset.\]
\end{definition}

\begin{definition}[Composite Extragradient Step]\label{defn:gen_apx_exgrad}
    Let $\zset \subseteq \R^{N}$ be convex, compact 
    and $\psi, r, s: \zset \to \R$ be convex functions. 
    We say an $(\zz', \vv') = \oracle_{\mathrm{xgrad}}^{\psi, r, s}(\zz, \vv, \hh)$ is a $\delta$-approximate composite extragradient step ($\delta$-CES) on input $(\zz, \vv, \hh)$ with respect to function $\psi$ and regularizers $r, s$ if
    \[\mathrm{regret}^{\zset}_{\hh, \psi}(\zz'; \ww) \le \Delta^{r}_{\zz}(\zz',\ww) + V_{\vv}^{s}(\ww) - V_{\vv'}^{s}(\ww) + \delta,
    \text{ for all } \ww \in \zset\,.
    \]
\end{definition}

Implementing a composite gradient step essentially comes down to computing approximate solutions to the problem $\min_{\ww \in \zset} \hh^\top \ww + \psi(\ww) + V_{\zz}^{r}(\ww)$. To see this, note that 
for $\delta = 0$ and $s(\zz) = 0, \forall \zz \in \zset$, the condition in \Cref{defn:gen_apx_grad} would imply that the point $\zz'$ satisfies
\[\langle \hh + \nabla V_{\zz}^{r}(\zz'), \zz - \ww \rangle + \psi(\zz') - \psi(\ww) \le 0, \forall \ww \in \zset,\]
which is equivalent to requiring $\zz' = \argmin_{\ww \in \zset} \hh^\top \ww + \psi(\ww) + V_{\zz}^{r}(\ww)$. Similarly, implementing composite extragradient steps is equivalent to computing approximate solutions to $\min_{\ww, \ww' \in \zset} \hh^\top \ww + \psi(\ww) + V_{\zz}^{r}(\ww) + V_{\vv}^{r}(\ww')$. 

We now present our algorithm for solving \Cref{prob:generic_regret}, which assumes implementability of (extra)gradient steps (\Cref{defn:gen_apx_grad,defn:gen_apx_exgrad}). This algorithm is almost identical to Algorithm 1 in \cite{JT23}. The only difference is that we replace the gradient and extragradient steps defined in \cite{JT23} with the more general notions (i.e., involving a constraint set and composite function) of gradient and extragradient steps in \Cref{defn:gen_apx_grad,defn:gen_apx_exgrad}.

\begin{algorithm}[H]
\caption{Composite Bilinear Optimization}\label{alg:main}
\KwData{Convex $\psi, r, s:\zset \to \R$, access to gradient operator $\gg:\zset \to \zset^*$, $T \in \Z$, starting points $\fullit^{(0)}, \auxit^{(0)} \in \zset$}
\KwResult{A point $\Bar{\zz} \in \zset$ with guarantees in \Cref{lem:convergence}}
\SetKwProg{Fn}{Function}{:}{}
\SetKwFunction{compblopt}{CompositeBilinearOptimization}
\Fn{\compblopt{$\psi, r, s, \gg, T, \fullit^{(0)}, \auxit^{(0)}$}}{
    \For{$t \in [T]$}{
        \tcp{Sufficient step: $\halfit^{(t)} = \argmin_{\zz \in \zset} \eta \gg(\fullit^{(t)})^\top \zz + \psi(\zz) + V_{\fullit^{(t)}}^{r}(\zz)$}
        $\halfit^{(t)} \gets \oracle_{\mathrm{grad}}^{\psi, r, s}(\fullit^{(t)}, \eta \gg(\fullit^{(t)}))$ \label{line:gradient_step}\;
        \BlankLine
        \tcp{Idealized step: $\fullit^{(t+1)} = \argmin_{\ww \in \zset} \eta \gg(\halfit^{(t)})^\top \ww + \psi(\ww) + V_{\zz}^{r+s}(\ww)$ and $\auxit^{(t+1)} = \argmin_{\yy \in \zset} V_{\auxit^{(t)}}^s(\yy)$}
        $(\fullit^{(t + 1)}, \auxit^{(t+1)}) \gets \oracle_{\mathrm{xgrad}}^{\psi / 2, r+s, s}(\fullit^{(t)}, \auxit^{(t)}, \frac{\eta}{2} \gg(\halfit^{(t)}))$ \label{line:xgradient_step}\;
    }
    \Return $\Bar{\zz} = \frac{1}{T} \sum_{t \in [T]} \halfit^{(t)}$
}
\end{algorithm}

\begin{lemma}\label{lem:convergence} 
    Let $\zset \subseteq \R^{N}$ be convex, compact, 
    $\psi:\zset \to \zset$ be a convex function, 
    $r, s:\zset \to \R$ be convex regularizers 
    so that $r$ is $\eta$-relaxed relative Lipschitz with respect to gradient operator $\gg:\zset \to \zset^*$. 
    Then, assuming implementability of $\delta$-CGSs $\oracle_{\mathrm{grad}}^{\psi, r, s}$ in \Cref{line:gradient_step} $\delta$-CESs $\oracle_{\mathrm{xgrad}}^{\psi / 2, r+s, s}$ in \Cref{line:xgradient_step},
    starting from $\fullit^{(0)}, \auxit^{(0)} \in \zset$, \Cref{alg:main} produces, after $T$ CGSs and CESs, 
    iterates that satisfy: \begin{equation}\label{eq:regret_guar}
        \frac{1}{T} \sum_{t \in [T]} \mathrm{regret}^{\zset}_{\gg(\halfit^{(t)}), \psi}(\halfit^{(t)}; \ww) \le \frac{2}{T \eta} (V_{\fullit^{(0)}}^r(\ww) + V_{\auxit^{(0)}}^{s}(\vv)) + 2 \delta,
    \text{ for all }
    \ww, \vv\in \zset
    \,.
    \end{equation}
\end{lemma}

\begin{proof}
    Fix $t \in [T]$. Since $\oracle_{\mathrm{grad}}^{\psi, r, s}$ is a $\delta$-approximate composite gradient step with respect to $\psi$ and regularizers $r, s$, we obtain (by setting $\fullit = \fullit^{(t+1)}$)  
    \begin{equation}\label{eq:grad_cond}
        \mathrm{regret}^{\zset}_{\eta \gg(\fullit^{(t)})}(\halfit^{(t)}; \fullit^{(t+1)}) \le V_{\fullit^{(t)}}^{r+s}(\fullit^{(t+1)}) - V_{\halfit^{(t)}}^{r}(\fullit^{(t+1)}) - V_{\fullit^{(t)}}^{r}(\halfit^{(t)}) + \delta.
    \end{equation}
    Since $\oracle_{\mathrm{xgrad}}^{\psi / 2, r+s, s}$ is a $\delta$-approximate composite extragradient step with respect to $\psi / 2$ and regularizers $r+s, s$, we obtain (after multiplying by $2$)
    \begin{equation}\label{eq:xgrad_cond}
        \mathrm{regret}^{\zset}_{\eta \gg(\halfit^{(t)}), \psi}(\fullit^{(t+1)}; \ww) \le 2 \Delta^{r+s}_{\fullit^{(t)}}(\fullit^{(t+1)}, \ww) + 2 V_{\auxit^{(t)}}^{s}(\ww) - 2 V_{\auxit^{(t+1)}}^{s}(\ww) + 2 \delta, \forall \ww \in \zset. 
    \end{equation}
    Since $g$ is $\eta$-relaxed relative Lipschitz with respect to $r$, by setting $(\zz^{(1)}, \zz^{(2)}, \zz^{(3)}) = (\fullit^{(t)}, \halfit^{(t)}, \fullit^{(t+1)})$ in \Cref{def:relaxed_RelLip}, we obtain: 
    \begin{equation}\label{ineq:area_conv_prop}
        \eta \langle \gg(\halfit^{(t)}) - \gg(\fullit^{(t)}), \halfit^{(t)} - \fullit^{(t+1)} \rangle \le V_{\fullit^{(t)}}^{r}(\halfit^{(t)}) + V_{\fullit^{(t)}}^{r}(\fullit^{(t+1)}) + V_{\halfit^{(t)}}^{r}(\fullit^{(t+1)}). 
    \end{equation}
    
    Summing up \eqref{eq:grad_cond}, \eqref{eq:xgrad_cond} and \eqref{ineq:area_conv_prop}, we obtain, after cancellation of terms $\eta \langle \gg(\halfit^{(t)}) - \gg(\fullit^{(t)}), \halfit^{(t)} - \fullit^{(t+1)} \rangle$ on the LHS and $V_{\fullit^{(t)}}^{r}(\halfit^{(t)})$ on the RHS, 
    \begin{align}
        \mathrm{regret}^{\zset}_{\eta \gg(\halfit^{(t)}), \psi}(\halfit^{(t)}; \ww) & \le 2 V_{\fullit^{(t)}}^{r+s}(\ww) - 2 V_{\fullit^{(t+1)}}^{r+s}(\ww) \\
        & + (V_{\fullit^{(t)}}^{r}(\fullit^{(t+1)}) - V_{\fullit^{(t)}}^{r+s}(\fullit^{(t+1)})) \\
        & + 2 V_{\auxit^{(t)}}^{s}(\ww) - 2 V_{\auxit^{(t+1)}}^{s}(\ww), \forall \ww \in \zset. 
    \end{align}
    Bounding $V_{\fullit^{(t)}}^{r+s}(\fullit^{(t+1)}) - V_{\fullit^{(t)}}^{r}(\fullit^{(t+1)}) \ge 0$, we thus obtain \[\mathrm{regret}^{\zset}_{\eta \gg(\halfit^{(t)}), \psi}(\halfit^{(t)}; \ww) \le 2 V_{\fullit^{(t)}}^{r+s}(\ww) - 2 V_{\fullit^{(t+1)}}^{r+s}(\ww) + 2 V_{\auxit^{(t)}}^{s}(\ww) - 2 V_{\auxit^{(t+1)}}^{s}(\ww).\]
    Summing up for all indices $t \in [T]$ and telescoping yields the desired conclusion. 
\end{proof}

\subsection{Implementing Composite (Extra)gradient Steps}
\label{subsec:implement_oracles}

In this section, we provide a result regarding sufficient conditions for implementing composite gradient and extragradient steps (\Cref{defn:gen_apx_grad,defn:gen_apx_exgrad}). We begin with  
several assumptions regarding our domain space $\zset$, as well as the composite function $\psi$ and the regularizers $r, s$ corresponding to the gradient and extragradient steps. This is captured by \Cref{assm:4.2space} below. This captures the setups considered by prior works \cite{sherman13,JST19,CST20,JT23}. 

\begin{assumption}\label{assm:4.2space}
    We assume that the set $\zset$ is decomposable as $\zset = \uset \times \vset$ with $\uset \subseteq \R^n, \vset \subseteq \R^m$ being convex, compact sets, and write use interchangeably $\zz = (\uu, \vv)$ with $\uu \in \uset, \vv \in \vset$. Secondly, we assume that the composite function $\psi$ only depends on the $\uset$ side variable, meaning $\psi(\uu, \vv) = \psi(\uu, \vv')$ for any $\uu \in \uset, \vv, \vv' \in \vset$. 
\end{assumption}

Throughout the section, we will work with a special class of regularizers, namely \textit{block compatible regularizers}. This is a concept that we introduce to generalize the types of regularizers used in \cite{JT23}. Specifically, the regularizers there satisfy these properties. 

\begin{definition}[Block Compatible Regularizers]\label{def:param_reg}
    Consider the setting of \Cref{assm:4.2space}.
    A family of regularizers $r^{\alpha}:\zset \to \R$, indexed by parameter $\alpha \ge 0$ and defined as \[r^{\alpha}(\uu, \vv) = r_1(\uu, \vv) + \alpha r_2(\vv),\] 
    where $r_1:\vset \to \vset$, and $r_2: \zset \to \R$, is block compatible if the following properties hold for blocks $r_1, r_2$:
    \begin{enumerate}
        \item $r_1(\uu, \cdot)$ is linear for every $\uu \in \uset$
        \item $r_1(\cdot, \vv)$ is convex for every $\vv \in \vset$
        \item $r_2$ is convex
        \item there exists $\alpha_0$ so that $r^{\alpha}$ is jointly convex over $\zset$ for every $\alpha \ge \alpha_0$. 
    \end{enumerate}
\end{definition}

Note that the regularizers defined in \Cref{def:regularizer} are a family of block-compatible regularizers. 
We now present the notion of $(\alpha, \beta; \delta)$-solvability, defined below, which involves assuming certain optimization procedures we can perform over sets $\uset$ and $\vset$ respectively. 

\begin{definition}\label{def:opt_oracles}
    Consider the setup of \Cref{assm:4.2space} and let $r^{\alpha}(\uu, \vv) := r_1(\uu, \vv) + \alpha r_2(\vv)$ be a family of regularizers with blocks $r_1:\zset \to \R, r_2:\vset \to \R$. We say $\{r^{\alpha}\}_{\alpha \ge 0}$ is $(\fixedalpha, \fixedbeta; \delta)$ solvable if we have
    access to a $\delta$-ABRO $\oracle_{\bestrsp}^{\uset, \psi, r^{\fixedalpha}}$ with respect to $(\uset, \psi, r^{\fixedalpha})$, 
    and, additionally, for any $\hh^{\vset} \in \R^m$, and $\vv\in \vset$, we can exactly solve \begin{equation}\label{prob:vset_opt}
    \min_{\yy \in \vset} \langle \hh^{\vset}, \yy \rangle + V^{\fixedbeta r_2}_{\vv}(\yy) \,.
\end{equation}
\end{definition}

$(\fixedalpha, \beta; \delta)$-solvability with respect to a family of block compatible regularizers enables the implementation of the composite gradient and extragradient steps, as shown in the lemma below. 

\begin{lemma}\label{lem:xgrad_steps_implementability}
    Consider the setting of \Cref{assm:4.2space} and 
    let 
    $\{r^{\alpha}\}_{\alpha \ge 0}$ be a family of block compatible regularizers with $\fixedalpha > 0$ so that $r^{\fixedalpha}$ is jointly convex and $r^{\fixedalpha}$ is $(\fixedalpha, \fixedbeta; \delta)$-solvable for some $\fixedbeta \ge \fixedalpha$. 
    We can implement a $3 \delta$-CGS $\oracle_{\mathrm{grad}}^{\psi, r^{\fixedalpha}, \fixedbeta r_2}((\uu, \vv), \hh)$ and a $3 \delta$-CES $\oracle_{\mathrm{xgrad}}^{\psi, r^{\fixedalpha+\fixedbeta}, \fixedbeta r_2}((\uu, \vv), \vv', \hh)$ by performing $O(1)$ operations that include (1) calling the optimization procedures on $\uset$ and $\vset$ specified in \Cref{def:opt_oracles} and (2) computing gradients of $r^{\alpha}$. 
\end{lemma}

The proof of \Cref{lem:xgrad_steps_implementability} is deferred to \Cref{apx:implementing_steps}. \Cref{lem:xgrad_steps_implementability} gives sufficient conditions to implement our composite gradient and extragradient steps using \Cref{alg:grad_step} and \Cref{alg:exgrad_step} respectively, which are needed to implement \Cref{alg:main}. 
To conclude the section, we note that with a proper choice of a family of block compatible regularizers $\{r^{\alpha}\}_{\alpha \ge 0}$, we can use the framework developed in \Cref{sec:extragrad_framework} to implement \Cref{alg:main} for approximately solving Problem \ref{eq:prob_u_v_formulation}, yielding the guarantee in \Cref{thm:box_simplex_solver}. 
Picking a suitable block compatible regularizer is the topic of our next subsection. 

\subsection{Area-Convex Regularizer for Constrained Box-Simplex Game}
\label{sec:regularizer}

In this section, we present a new relaxed relative Lipschitz regularizer for the constrained box-simplex game. To show that it is relaxed relative Lipschitz, we show it is \textit{area convex}, which is a stronger property (see \Cref{lem:area_to_relaxed}). Area convexity is defined below.

\begin{definition}[Area Convexity, \cite{sherman2017area}]\label{def:area_convexity}
    Let $\zset \subseteq \R^{N}$ be a convex set.
    We say that $r:\zset \to \R$ is $\eta$-area convex with respect to a gradient 
    operator $\gg:\zset \to \zset^*$ if for all $(\zz^{(1)}, \zz^{(2)}, \zz^{(3)}) \in \zset \times \zset \times \zset$ and  $\cc = \frac{1}{3} (\zz^{(1)} + \zz^{(2)} + \zz^{(3)})$,
    \[\eta \langle \gg(\zz^{(2)}) - \gg(\zz^{(1)}), \zz^{(2)} - \zz^{(3)} \rangle \le r(\zz^{(1)}) + r(\zz^{(2)}) + r(\zz^{(3)}) - 3 r(\cc).\]
\end{definition}

\paragraph{Notation.}
Following the notation of \Cref{subsec:implement_oracles}, in this subsection we let
$\uset \subseteq \R_{\ge 0}^n$ be a subset of the positive quadrant in $n$ dimensions, and $\vset = \Delta^m$ is the $m$-dimensional simplex. For a constrained box-simplex game $(\uset, \Delta^m, \AA)$, we 
define gradient operator as $\gg(\uu, \vv) \defeq [\AA^\top \vv; -\AA\uu]$. Additionally, given a bound $\linfbound \ge \max_{\uu \in \uset} \||\AA| \uu\|_{\infty}$ and a bound $\usetBound \ge \usetSize$, we define the parameters $\xi, B$ associated with $\linfbound, \usetBound$ by $\xi := \min(1, \frac{\linfbound}{nm\|\AA\|_{\infty}})$ and $B \defeq \log (\min(\frac{1}{\xi}, \usetSize + \xi))$.  

Note that for a constrained box-simplex game $(\uset, \Delta^m, \AA)$, approximately solving an instance of \Cref{prob:generic_regret} 
with $\zset = \uset \times \Delta^m$ and gradient operator $\gg$ yields an approximate solution for Problem \ref{eq:prob_u_v_formulation}. 
As explained in \Cref{subsec:extragrad_overview},
one of the main innovations of our work is designing a regularizer whose range 
has a favorable dependence on the structure of set $\uset$, 
while still retaining the area convexity property with respect to the gradient operator $\gg$. 
The formal definition of our regularizer, which was informally defined in \Cref{def:regularizer}, is provided below.

\begin{definition}[Doubly Entropic Regularizer]
\label{def:regularizer_complete}
Consider a constrained box-simplex game $(\uset \subseteq \R_{\ge 0}^n, \Delta^m, \AA \in \R^{m \times n})$ 
and a bound $\linfbound \ge \max_{\uu \in \uset} \||\AA| \uu\|_{\infty}$. 
Let $\xi = \min(1, \frac{\linfbound}{nm\|\AA\|_{\infty}})$ and 
define blocks $r_1:\zset \to \R, r_2:\vset \to \R$ by \[r_1(\uu, \vv) := \sum_{i \in [m], j \in [n]} (\vv_i + \xi) |\AA_{ij}| (\uu_j + \xi) \log (\uu_j + \xi), \quad r_2(\vv) := \sum_{i \in [m]} \vv_i \log \vv_i,\] 
and a family of regularizers parametrized by $\alpha \ge 0$,
$r^{\alpha}: \uset \times \Delta^m \to \R$, as follows: 
\begin{align*}
    r^{\alpha}(\uu, \vv) 
    &:= r_1 + \alpha r_2 = \sum_{i \in [m], j\in[n]} (\vv_i + \xi) |\AA_{ij}| (\uu_j + \xi) \log (\uu_j + \xi) + \alpha \sum_{i\in[m]} \vv_i \log \vv_i \\
\end{align*}
\end{definition}

For the rest of the subsection, whenever we write $r^{\alpha}$, we refer to the \emph{doubly-entropic regularizer} defined in \Cref{def:regularizer_complete}. 
Our main lemma of this section argues that $r^{\alpha}$ is $O(1)$-area convex, jointly-convex on $\uset \times \Delta^m$, and it has range $\O(\alpha)$ when setting $\alpha \ge 4 B \linfbound$, where $\linfbound$ is a parameter with $\linfbound \ge \max_{\uu \in \uset} \||\AA|\uu\|_{\infty}$. 
\begin{lemma}
\label{lem:regularizer}
Consider a constrained box-simplex game $(\uset, \Delta^m, \AA)$, $\linfbound \ge \max_{\uu \in \uset} \||\AA|\uu\|_{\infty}$ and $\usetBound \ge \usetSize$. For any $\alpha \ge 16 B \linfbound$, where $B = 1 + \log (\max(\frac{1}{\xi}, \usetBound + \xi))$, $r^{\alpha}(\cdot, \cdot)$ is $\frac{1}{3}$-area-convex w.r.t. the gradient operator $\gg$. 
Additionally, $r^{\alpha}(\cdot, \cdot)$ is convex and $\max_{\uu \in \uset, \vv \in \Delta^m} |r^{\alpha}(\uu, \vv)| \le 4 B \linfbound + \O(\alpha).$
\end{lemma}

We will prove the area-convexity using the 2nd-order condition presented in \cite{JST19} and \cite{sherman2017area}. The formal statement is provided below. 
\begin{lemma}[2nd-Order Area Convexity Condition, Theorem 1.6~\cite{sherman2017area}]
\label{lem:2ndOrderAreaCvx}
Consider a constrained box-simplex game $(\uset, \Delta^m, \AA)$, and some regularizer $r: \uset \times \Delta^m \to \R$ which is twice differentiable on $\mathrm{int}(\uset \times \Delta^m)$. 
Then, 
$r$ is $\frac{1}{3}$-area-convex w.r.t.\ $\gg$ if
\begin{align*}
\begin{pmatrix}
    \g^2 r(\zz) & -\JJ \\
    \JJ & \g^2 r(\zz)
\end{pmatrix} \succeq 0,
\text{ where $\JJ$ is the Jacobian of $\gg(\cdot, \cdot)$, i.e., }
    \JJ = \begin{pmatrix}
        0 & \AA^\top \\
        -\AA & 0
    \end{pmatrix}\,.
\end{align*}
\end{lemma}

The following lemma bounds $|\bb^\top \AA \aa|$ in terms of $\linfbound$ and is helpful in proving the area-convexity of $r^{\alpha}$. 
\begin{lemma}[Upper bound on $\bb^\top \AA \aa$]
\label{lem:acHelper}
Consider a constrained box-simplex game $(\uset, \Delta^m, \AA)$ and let $\linfbound \ge \max_{\uu \in \uset} \||\AA|\uu\|_{\infty}$. 
For any $\aa \in \R^n, \bb \in \R^m, \uu \in \uset, \vv \in \Delta^m, \tau > 0$, we have
\begin{align*}
    \left|\bb^\top \AA \aa\right| \le \sqrt{\linfbound} \tau \sum_{i\in[m]} \frac{\bb_i^2}{\vv_i} + \frac{\sqrt{\linfbound}}{\tau} \sum_{j\in[n]} \left(\sum_{i\in[m]} |\AA_{ij}|\vv_i\right) \frac{\aa_j^2}{\uu_j + \xi}\,.
\end{align*}
\end{lemma}
\begin{proof}[Proof of \Cref{lem:acHelper}]
Let $\uu$ be any element of $\uset.$
We have
\begin{align*}
\bb^\top \AA \aa 
= \sum_{i\in[m], j\in[n]} \bb_i \AA_{ij} \aa_j 
\le \sum_{i\in[m], j\in[n]} |\bb_i| |\AA_{ij}| |\aa_j|
\end{align*}

Let $\vv$ be any element of $\Delta^m$, we have, by Cauchy-Schwartz, that, for any $i$,
\begin{align*}
|\bb_i| \sum_{j\in[n]}  |\AA_{ij}| |\aa_j|
&\le |\bb_i| \left(\sum_{j\in[n]} |\AA_{ij}| (\uu_j + \xi)\right)^{\frac{1}{2}} \left(\sum_{j\in[n]} |\AA_{ij}| \frac{\aa_j^2}{\uu_j + \xi}\right)^{\frac{1}{2}} \\
&\le \sqrt{2\linfbound} |\bb_i| \left(\sum_{j\in[n]} |\AA_{ij}| \frac{\aa_j^2}{\uu_j + \xi}\right)^{\frac{1}{2}} \\
&\le \sqrt{\linfbound} \tau \frac{\bb_i^2}{\vv_i} + \frac{\sqrt{\linfbound}}{\tau}\vv_i \left(\sum_{j\in[n]} |\AA_{ij}| \frac{\aa_j^2}{\uu_j + \xi}\right)
\end{align*}
where the 2nd inequality comes from $\||\AA|(\uu + \xi \vecone^n)\|_{\infty} \le \||\AA|\uu\|_{\infty} + \xi \|\AA\|_{\infty} \le 2\linfbound$, and the 3rd inequality comes from $ab \le \frac{\tau}{2y}a^2 + \frac{y}{2\tau}b^2$ for any $a, b, y, \tau > 0.$
The lemma concludes by summing up all $i$'s.
\end{proof}

The next lemma shows that for $\alpha$ large enough, meaning $\alpha \ge 8 \linfbound (1+ \max \{\log \frac{1}{\xi}, \log (\usetSize + \xi)\})$, $r^{\alpha}$ is convex and $\g^2 r^{\alpha}$ can be approximated by its diagonal part. 
\begin{lemma}[$r^{\alpha}(\cdot, \cdot)$ is convex]
\label{lem:regCvx}
Consider a constrained box-simplex game $(\uset, \Delta^m, \AA)$. 
Let $\linfbound \ge \max_{\uu \in \uset} \||\AA|\uu\|_{\infty}$ and $\usetBound \ge \usetSize$. 
If $\alpha \ge 2 B \linfbound$, where $B = 1 + \log (\max(\frac{1}{\xi}, \usetBound + \xi))$, $r^{\alpha}$ is jointly-convex. 
Moreover, for any $\aa \in \R^n, \bb \in \R^m, \uu \in \uset, \vv \in \Delta^m$, we have:
\begin{align*}
\begin{pmatrix}
    \aa \\
    \bb
\end{pmatrix}^\top
\g^2 r^{\alpha}(\uu, \vv)
\begin{pmatrix}
    \aa \\
    \bb
\end{pmatrix}
\ge \frac{1}{2}
\begin{pmatrix}
    \aa \\
    \bb
\end{pmatrix}^\top
\ddiag(\g^2 r^{\alpha}(\uu, \vv))
\begin{pmatrix}
    \aa \\
    \bb
\end{pmatrix}\,.
\end{align*}
\end{lemma}
\begin{proof}[Proof of \Cref{lem:regCvx}]
Be definition of $r^{\alpha}$, we know
\begin{align*}
    \g^2 r^{\alpha}(\uu, \vv) = \begin{pmatrix}
        \ddiag(|\AA|^\top (\vv + \xi))\ddiag\left(\frac{1}{\uu + \xi \vecone^n }\right) & \ddiag(1 + \log(\uu + \xi \vecone^n )) |\AA|^\top \\
        |\AA| \ddiag(1 + \log(\uu + \xi \vecone^n )) & \alpha \ddiag\left(\frac{1}{\vv}\right)
    \end{pmatrix}
\end{align*}
For any $\aa \in \R^n$ and $\bb \in \R^m$, we have
\begin{align*}
\begin{pmatrix}
    \aa \\
    \bb
\end{pmatrix}^\top
\g^2 r^{\alpha}(\uu, \vv)
\begin{pmatrix}
    \aa \\
    \bb
\end{pmatrix}
& = \alpha \sum_{i\in[m]} \frac{\bb_i^2}{\vv_i} + \sum_{i\in[m], j\in[n]} \bb_i |\AA_{ij}| (1 + \log(\uu_j + \xi)) \aa_j \\
& + \sum_{j\in[n]} \left(\sum_{i\in[m]} |\AA_{ij}| (\vv_i + \xi) \right) \frac{\aa_j^2}{\uu_j + \xi}
\end{align*}
By \Cref{lem:acHelper} and the fact that $|1 + \log(\uu_j + \xi)| \le 1 + \log (\max_{\uu \in \uset} \|\uu\|_{\infty} + \xi) \le B$ for any $j$, we can bound the scale of the cross term by
\begin{align*}
\left|\sum_{i\in[m], j\in[n]} \bb_i |\AA_{ij}| (1 + \log(\uu_j + \xi)) \aa_j\right|
&\le B
\left(\sqrt{\linfbound} \tau \sum_{i\in[m]} \frac{\bb_i^2}{\vv_i} + \frac{\sqrt{\linfbound}}{\tau} \sum_{j\in[n]} \left(\sum_{i\in[m]} |\AA_{ij}|\vv_i\right) \frac{\aa_j^2}{\uu_j + \xi}\right)
\end{align*}
for any $\tau > 0.$
By setting $\tau = \sqrt{\linfbound}$ and observing that $\alpha \ge 2 B \linfbound$, we conclude the proof by having
\begin{align*}
\begin{pmatrix}
    \aa \\
    \bb
\end{pmatrix}^\top
\g^2 r^{\alpha}(\uu, \vv)
\begin{pmatrix}
    \aa \\
    \bb
\end{pmatrix}
&= \alpha \sum_{i\in[m]} \frac{\bb_i^2}{\vv_i} + \sum_{i\in[m], j\in[n]} \bb_i |\AA_{ij}| (1 + \log(\uu_j + \xi)) \aa_j \\
& + \sum_{j\in[n]} \left(\sum_{i\in[m]} |\AA_{ij}| (\vv_i + \xi) \right) \frac{\aa_j^2}{\uu_j + \xi} \\
&\ge \frac{\alpha}{2}\sum_{i\in[m]} \frac{\bb_i^2}{\vv_i} + \frac{1}{2}\sum_{j\in[n]} \left(\sum_{i\in[m]} |\AA_{ij}| (\vv_i + \xi) \right) \frac{\aa_j^2}{\uu_j + \xi} \\
&\ge \frac{1}{2}
\begin{pmatrix}
    \aa \\
    \bb
\end{pmatrix}^\top
\ddiag(\g^2 r^{\alpha}(\uu, \vv))
\begin{pmatrix}
    \aa \\
    \bb
\end{pmatrix}\,.
\end{align*}
This proved the ``Moreover'' part of the lemma. To prove joint convexity of $r^{\alpha}$ over $\uset \times \Delta^m$, note that terms $\frac{\bb_i^2}{\vv_i} , \vv_i + \xi, \frac{\aa_j^2}{\uu_j + \xi}$ for all $i \in [m], j \in [n]$ are all positive, which implies that $\g^2 r^{\alpha}(\uu, \vv) \succeq 0$ holds for any $\uu \in \uset, \vv \in \Delta^m.$
\end{proof}

Now, we are ready to show that $r^{\alpha}$ is area-convex and prove \Cref{lem:regularizer}.
\begin{proof}[Proof of \Cref{lem:regularizer}]
First, note that the convexity of $r^{\alpha}$ follows from \Cref{lem:regCvx}. 
To prove that $r^{\alpha}$ is area-convex, it suffices, due to \Cref{lem:2ndOrderAreaCvx}, to prove that
\begin{align*}
0
\preceq
\begin{pmatrix}
\g^2 r^{\alpha}(\zz) & -\JJ \\
\JJ & \g^2 r^{\alpha}(\zz)
\end{pmatrix}
\end{align*}
for any $\zz = (\uu, \vv) \in \uset \times \Delta^m.$

For every $(\aa, \bb, \cc, \dd) \in \R^n \times \R^m \times \R^n \times \R^m$, we have
\begin{align*}
\begin{pmatrix}
    \aa \\
    \bb \\
    \cc \\
    \dd
\end{pmatrix}^\top
\begin{pmatrix}
\g^2 r^{\alpha}(\zz) & -\JJ \\
\JJ & \g^2 r^{\alpha}(\zz)
\end{pmatrix}
\begin{pmatrix}
    \aa \\
    \bb \\
    \cc \\
    \dd
\end{pmatrix}
&=
\begin{pmatrix}
    \aa \\
    \bb
\end{pmatrix}^\top
\g^2 r^{\alpha}(\zz)
\begin{pmatrix}
    \aa \\
    \bb
\end{pmatrix}
+
\begin{pmatrix}
    \cc \\
    \dd
\end{pmatrix}^\top
\g^2 r^{\alpha}(\zz)
\begin{pmatrix}
    \cc \\
    \dd
\end{pmatrix}
\\
&-
\begin{pmatrix}
    \aa \\
    \bb
\end{pmatrix}^\top
\JJ
\begin{pmatrix}
    \cc \\
    \dd
\end{pmatrix}
+
\begin{pmatrix}
    \cc \\
    \dd
\end{pmatrix}^\top
\JJ
\begin{pmatrix}
    \aa \\
    \bb
\end{pmatrix}
\end{align*}
By definition of $\JJ$, we have
\begin{align}
\label{eq:JPart}
-
\begin{pmatrix}
    \aa \\
    \bb
\end{pmatrix}^\top
\JJ
\begin{pmatrix}
    \cc \\
    \dd
\end{pmatrix}
+
\begin{pmatrix}
    \cc \\
    \dd
\end{pmatrix}^\top
\JJ
\begin{pmatrix}
    \aa \\
    \bb
\end{pmatrix}
&= 2\bb^\top \AA \cc -2\dd^\top \AA \aa
\end{align}
By \Cref{lem:acHelper} and \Cref{lem:regCvx}, we can bound the absolute value of \eqref{eq:JPart} by
\begin{align*}
2|\bb^\top \AA \cc| + 2|\dd^\top \AA \aa|
&\le 2\sqrt{\linfbound}\tau  \sum_{i\in[m]} \frac{\bb_i^2 + \dd_i^2}{\vv_i} + \frac{2\sqrt{\linfbound}}{\tau}\sum_{j\in[n]} \left(\sum_{i\in[m]} |\AA_{ij}\vv_i|\right) \frac{\aa_j^2 + \cc_j^2}{\uu_j + \xi} \,.
\end{align*}
for any $\tau > 0$. Setting $\tau = 4 \sqrt{\linfbound}$, we thus obtain 
\begin{align*}
2|\bb^\top \AA \cc| + 2|\dd^\top \AA \aa|
&\le \frac{\alpha}{2} \sum_{i\in[m]} \frac{\bb_i^2 + \dd_i^2}{\vv_i} + \frac{1}{2}\sum_{j\in[n]} \left(\sum_{i\in[m]} |\AA_{ij}\vv_i|\right) \frac{\aa_j^2 + \cc_j^2}{\uu_j + \xi} \\
&= \frac{1}{2}\begin{pmatrix}
    \aa \\
    \bb
\end{pmatrix}^\top
\ddiag(\g^2 r^{\alpha}(\zz))
\begin{pmatrix}
    \aa \\
    \bb
\end{pmatrix}
+
\frac{1}{2}\begin{pmatrix}
    \cc \\
    \dd
\end{pmatrix}^\top
\ddiag(\g^2 r^{\alpha}(\zz))
\begin{pmatrix}
    \cc \\
    \dd
\end{pmatrix}
\end{align*}
by using that $\alpha \ge 16 \linfbound.$
Hence, we obtain that 
\begin{equation}
    \begin{pmatrix}
    \aa \\
    \bb \\
    \cc \\
    \dd
\end{pmatrix}^\top
\begin{pmatrix}
\g^2 r^{\alpha}(\zz) & -\JJ \\
\JJ & \g^2 r^{\alpha}(\zz)
\end{pmatrix}
\begin{pmatrix}
    \aa \\
    \bb \\
    \cc \\
    \dd
\end{pmatrix}
\ge \frac{1}{2}\begin{pmatrix}
    \aa \\
    \bb
\end{pmatrix}^\top
\ddiag(\g^2 r^{\alpha}(\zz))
\begin{pmatrix}
    \aa \\
    \bb
\end{pmatrix}
+
\frac{1}{2}\begin{pmatrix}
    \cc \\
    \dd
\end{pmatrix}^\top
\ddiag(\g^2 r^{\alpha}(\zz))
\begin{pmatrix}
    \cc \\
    \dd
\end{pmatrix} \,.
\end{equation}
Since $\ddiag(\g^2 r^{\alpha}(\zz))$ is positive definite,
the RHS term above is always positive. 
Consequently, we obtain that $r^{\alpha}$ is $\frac{1}{3}$-area-convex by \Cref{lem:regCvx}.

Next, we show that the range of $r^{\alpha}$ is $\le 4 \linfbound B + \O(\alpha).$
To bound the entropy part ($r_2$), note that $\max_{\vv \in \Delta^m} |\alpha r_2(\vv)| \le \alpha \log m$. 
To bound the range of $r_1$, note that 
\begin{align*}
\left|\sum_{i\in[m], j\in[n]} (\vv_i + \xi) |\AA_{ij}| (\uu_j + \xi) \log (\uu_j + \xi)\right|
&\le (1 + \xi) \max_i \sum_{j\in[n]} |\AA_{ij}| (\uu_j + \xi) |\log (\uu_j + \xi)| \\
&\le 4 B \linfbound \,,
\end{align*}
where we used that $\max_{\uu \in \sset} (1 + \xi) \max_{\uu \in \uset} \log (\uu_j + \xi) \le 2 \max_{\uu \in \uset} \log (\uu_j + \xi) \le  2 B$ and that $\max_{i \in [m]} \sum_{j\in[n]} |\AA_{ij}| (\uu_j + \xi) = \max_{i \in [m]} |\AA_i| (\uu + \xi \vecone^n) \le 2 \linfbound$. 
The lemma follows.
\end{proof}

\subsection{Putting It All Together}
\label{subsec:3_lastsubsec}

In this subsection, we prove \Cref{thm:box_simplex_solver}. First, we present the pseudocode of \Cref{alg:box_simplex}, which is used in \Cref{thm:box_simplex_solver}. At a high level, \Cref{alg:box_simplex} runs the extragradient method (\Cref{alg:main}) on the space $\{(\uu, \vv)\} \subseteq \uset \times \Delta^m$ for $\otilde(\eps^{-1} \log (\max(\frac{1}{\xi}, \usetBound + \xi)))$ iterations. We then present the result stating that area convexity is a stronger property than relaxed relative Lipschitzness, followed by the proof of \Cref{thm:box_simplex_solver}. 

\begin{algorithm}[htp!]
\caption{Box-simplex Algorithm}\label{alg:box_simplex}
\KwData{Convex function $\psi$, constraint set $\uset \subseteq \R_{\ge 0}^n$, matrix $\AA \in \R^{m \times n}$, bounds $\linfbound \ge \max_{\uu \in \uset} \||\AA| \uu\|_{\infty}, \usetBound \ge \usetSize$, a $\delta$-ABRO $\oracle_{\bestrsp}^{\uset, \psi, r}$ with respect to $(\uset, \psi, r)$, error parameters $\delta, \epsilon$}
\KwResult{A point $\algoutput{\uu}, \algoutput{\vv}$ with guarantees in \Cref{thm:box_simplex_solver}}
\SetKwProg{Fn}{Function}{:}{}
\SetKwFunction{compblopt}{CompositeBilinearOptimization}
\SetKwFunction{boxsimplex}{BoxSimplexOptimizer}
\Fn{\boxsimplex{$\psi, \uset, \AA, \linfbound, \usetBound, \oracle_{\bestrsp}^{\uset, \psi, r}, \epsilon$}}{
    $r(\uu, \vv) 
    := \sum_{i \in [m], j \in [n]} (\vv_i + \xi) |\AA_{ij}| (\uu_j + \xi) \log (\uu_j + \xi) + \fixedalpha \sum_{i \in [m]} \vv_i \log \vv_i$, where $\fixedalpha \gets 8 \linfbound (1+\log (\max(\frac{1}{\xi}, \usetBound + \xi)))$ \label{line:choose_r_alpha} \;
    \tcp*[h]{The doubly entropic regularizer; $\xi$ picked as in \Cref{def:regularizer_complete}}
    $s(\vv) := \fixedalpha \sum_{i \in [m]} \vv_i \log \vv_i, \psi(\uu) := \sum_{i \in [m], j \in [k]} \psi_{ij}(\uu_{ij})$ \label{line:choose_rest_of_reg}\;
    $\gg(\uu, \vv) := [\AA^\top \vv; -\AA\uu]$ \tcp*{Gradient operator}
    $\oracle_{\mathrm{grad}}^{\psi, r, s}, \oracle_{\mathrm{xgrad}}^{\psi / 2, r+s, s} \gets$ gradient and extragradient step oracles \; 
    \tcp{Implemented via \Cref{lem:xgrad_steps_implementability}} 
    $\vv^{(1)} = \frac{1}{m} \Vec{1}, \enspace \uu^{(1)} \gets$ feasible point in $\uset$ obtained by calling $\oracle_{\bestrsp}^{\uset, \psi, r}$ \label{line:set_starting_point} \;
    $T \gets \otilde(\frac{1}{\epsilon} \log (\max(\frac{1}{\xi}, \usetBound + \xi)))$\;
    $(\algoutput{\uu}, \algoutput{\vv}) \gets$ \compblopt{$\psi, r, s, \gg, T, \uu^{(0)}, \vv^{(0)}$} \label{line:calling_comp_blnopt}\;
    \Return $\algoutput{\uu}$
}    
\end{algorithm}

Note that the regularizers in \Cref{line:choose_r_alpha} and \Cref{line:choose_rest_of_reg} are picked in a way that is compatible with \Cref{assm:4.2space}. 
Along with the ABRO received as input, this choice allows us to follow the framework developed throughout this section.  
We next present the result that states area convexity is a stronger property than relaxed relative Lipschitzness.

\begin{lemma}[Lemma 17 from \cite{JT23}]\label{lem:area_to_relaxed}
    Let $\zset \subseteq \R^{N}$ be a convex, compact set and 
    $r:\zset \to \R$ an $\eta$-area convex regularizer with respect to gradient 
    operator $\gg:\zset \to \zset^*$. Then, $r$ is $\eta$-relaxed relative Lipschitz with respect to $\gg$. 
\end{lemma}

Now, we are ready to prove \Cref{thm:box_simplex_solver}. 
\begin{proof}[Proof of \Cref{thm:box_simplex_solver}]
    Working within the context of \Cref{subsec:implement_oracles}, the sets $\uset, \vset$ are $\uset = \uset, \vset = \Delta^m$. Moreover, $\psi$ is only defined on $\uset$, and hence the assumptions of \Cref{assm:4.2space} are satisfied. 
    Define 
    the family of regularizers $\{r^{\alpha}\}_{\alpha \ge 0}$ by $r^{\alpha}(\uu, \vv) = r_1(\uu, \vv) + \alpha r_2(\vv)$, where $r_1: \uset \times \Delta^m \to \R$ as $r_1(\uu, \vv) := \sum_{i \in [m], j \in [k]} (\vv_i + \xi) (\uu_{ij} + \xi) \log (\uu_{ij} + \xi)$ and $r_2:\Delta^m \to \R$ as $r_2(\vv) = \sum_{i \in [m]} \vv_i \log \vv_i$. We show that $\{r^{\alpha}\}_{\alpha \ge 0}$ is block-compatible. First, note that
    $r_1$ is convex in both variable $\uu$ and $\vv$ and $r_2$ is convex. Next, by
    \Cref{lem:regularizer}, for any $\alpha \ge 8 \linfbound (1+\log (\max(\frac{1}{\xi}, \usetBound + \xi)))$, $r^{\alpha}$ is both jointly convex on $\uset \times \Delta^m$ and $\frac{1}{3}$-area convex with respect to gradient operator $\gg(\uu, \vv) = [\AA^\top \vv; -\AA\uu]$. Thus, we obtain that the family of regularizers $\{r^{\alpha}\}_{\alpha \ge 0}$ is block compatible. 
    
    Note that in \Cref{line:choose_r_alpha}, we have selected the regularizers $r(\uu, \vv) = r^{\fixedalpha}(\uu, \vv), s(\uu, \vv) = \beta r_2(\vv)$, for $\fixedalpha = \beta = 8 \linfbound (1+ \log (\max(\frac{1}{\xi}, \usetBound + \xi)))$. 
    Next, note that we have access to exact minimization for the $\Delta^m$ variables and the $\delBR$-ABRO $\oracle_{\bestrsp}^{\uset, \psi, r}$ with respect to $(\uset, \psi, r)$. Consequently, $r = r^{\fixedalpha}$ is $(\fixedalpha, \fixedalpha)$-solvable (\Cref{def:opt_oracles}) and thus,
    by \Cref{lem:xgrad_steps_implementability}, this allows us to implement $\oracle_{\mathrm{grad}}^{\psi, r, s}$ and $\oracle_{\mathrm{xgrad}}^{\psi / 2, r+s, s}$, the $3 \delBR$-approximate composite (extra)gradient steps respectively. 
    Moreover, the points $\vv^{(0)}, \uu^{(0)}$ in \Cref{line:set_starting_point} are valid feasible points (i.e., $\vv^{(0)} \in \Delta^m, \uu^{(0)} \in \uset$), since $\uu^{(0)}$ is the output resulting from a call to the algorithm in \Cref{thm:build_BR_oracle} with an arbitrary $\VV \in \R^n$, which is guaranteed to be in $\uset$. 

    Let $\{\fullit^{(t)}, \fullit^{(1)}_{\mathrm{aux}}, \halfit^{(t)}\}_{t \in [T]}$ be the iterates generated by the call of $\mathsf{CompositeBilinearOptimization}(\cdot)$ in \Cref{line:calling_comp_blnopt}, where $\fullit^{(t)} := (\uu^{(t, 1)}, \vv^{(t, 1)}), \fullit^{(1)}_{\mathrm{aux}} := (\uu^{(t, 1)}_{\mathrm{aux}}, \vv^{(t, 1)}_{\mathrm{aux}})$ are the full iterates (with $\fullit^{(1)}_{\mathrm{aux}}$ being the auxiliary iterate) and $\halfit^{(t)} := (\uu^{(t, 2)}, \vv^{(t, 2)})$ are the half iterates. 
    Since $r$ is $\frac{1}{3}$-area convex with respect to $\gg$ (due to our choice of $\fixedalpha$), 
    by \Cref{lem:area_to_relaxed} applied on space $\zset = \uset \times \Delta^m$, $r$ is also $\frac{1}{3}$-relaxed relative Lipschitz with respect to $\gg$, and since we can implement $\oracle_{\mathrm{grad}}^{\psi, r, s}$ and $\oracle_{\mathrm{xgrad}}^{\psi / 2, r+s, s}$ and are starting from feasible points $\fullit^{(1)} = \fullit^{(1)}_{\mathrm{aux}} = (\vv^{(1)}, \uu^{(1)})$, by \Cref{lem:convergence}, we obtain 
    \begin{align*}
    \frac{1}{T} \sum_{t \in [T]} \mathrm{regret}^{\zset}_{\gg(\uu^{(t, 2)}, \vv^{(t, 2)}), \psi}((\uu^{(t, 2)}, \vv^{(t, 2)}); (\uu, \vv)) \le \frac{6}{T} (V_{(\uu^{(1)}, \vv^{(1)})}^r(\uu, \vv) + V_{\vv^{(1)}}^{s}(\vv)) + 6 \delBR, \\
    \text{ for all }
    (\uu, \vv) \in \uset \times \Delta^m, \vv \in \Delta^m
    \,,
    \end{align*}
    where we used that $s$ only depends on the variable living in $\Delta^m$. 
    Bounding $\max_{(\uu , \vv) \in \uset \times \Delta^m} |(r+s)(\uu, \vv)| \le \otilde(\linfbound \log (\min(\frac{1}{\xi}, \usetBound + \xi)))$ by \Cref{lem:regularizer} (as $r+s = r^{\alpha + \beta}$), as well as $V_{\vv^{(1)}}^{s}(\vv)) \le O(\fixedalpha \log m)$ (since $s = \fixedalpha r_2$), we obtain 
    \[V_{(\uu^{(1)}, \vv^{(1)})}^r(\uu, \vv) + V_{\vv^{(1)}}^{s}(\vv) \le \otilde(\linfbound)\,.\]
    Additionally, note that 
    \[\mathrm{regret}^{\zset}_{\gg(\uu^{(t, 2)}, \vv^{(t, 2)}), \psi}((\uu^{(t, 2)}, \vv^{(t, 2)}); (\uu, \vv)) = \vv \AA \uu^{(t, 2)} - \vv^{(t, 2)} \AA \uu + \psi(\uu^{(t, 2)}) - \psi(\uu)\,,
    \]
    where we used that $\langle \gg(\uu^{(t, 2)}, \vv^{(t, 2)}), (\uu^{(t, 2)}, \vv^{(t, 2)})\rangle = \langle \vv^{(t, 2)}, \AA \uu^{(t, 2)} \rangle - \langle \vv^{(t, 2)}, \AA \uu^{(t, 2)} \rangle = 0$. 
    Hence, we obtain 
    \begin{equation}\label{eq:regret_to_error}
        \frac{1}{T} \sum_{t \in [T]} (\langle \vv, \AA \uu^{(t, 2)} \rangle - \langle \vv^{(t, 2)}, \AA \uu \rangle + \psi(\uu^{(t, 2)}) - \psi(\uu)) \le \frac{6}{T} \cdot \otilde(\linfbound) + 6 \delBR\,.
    \end{equation}
    Now, since $\algoutput{\uu} = \frac{1}{T} \sum_{t \in [T]} \uu^{(t, 2)}$, by convexity, $\psi(\algoutput{\uu}) \le \frac{1}{T} \sum_{t \in [T]} \psi(\uu^{(t, 2)})$. Using that $\frac{1}{T} \sum_{t \in [T]} \vv^{(t, 2)} = \algoutput{\vv}$, 
    we obtain 
    \[\langle \vv, \AA \algoutput{\uu} \rangle + \psi(\algoutput{\uu}) \le \langle \algoutput{\vv}, \AA \uu \rangle + \psi(\uu) + \epsilon \linfbound + 6 \delBR, \forall \uu \in \uset, \vv \in \Delta^m \,,\]
    where we used that for $T = \otilde(\epsilon^{-1} \log (\max(\frac{1}{\xi}, \usetBound + \xi)))$ large enough, $\frac{6}{T} \cdot \otilde(\linfbound) \le \epsilon \linfbound$. 

    Additionally, by \Cref{lem:convergence}, the call to $\mathsf{CompositeBilinearOptimization}(\cdot)$ requires $T$ steps of $\oracle_{\mathrm{grad}}^{\psi, r, s}$ and $T$ steps of $\oracle_{\mathrm{xgrad}}^{\psi / 2, r+s, s}$. Each step of $\oracle_{\mathrm{grad}}^{\psi, r, s}$ and $\oracle_{\mathrm{xgrad}}^{\psi / 2, r+s, s}$ takes, by \Cref{lem:xgrad_steps_implementability}, $O(1)$ calls to $\oracle_{\bestrsp}^{\uset, \psi, r}$ and $O(\nnz(\AA) + m + n)$ additional time (from computing $O(1)$ gradients of $r$ or $r+s$ respectively). 
    Hence, this yields the stated runtime of $\otilde((\nnz(\AA) + m + n) B \epsilon^{-1})$ and $\otilde(B \epsilon^{-1})$ calls to $\oracle_{\bestrsp}^{\uset, \psi, r}$. 

    Lastly, we prove the bound regarding the inputs given to  when implementing either $\oracle_{\mathrm{grad}}^{\psi, r, s}$ or $\oracle_{\mathrm{xgrad}}^{\psi / 2, r+s, s}$ (inside the call to $\mathsf{CompositeBilinearOptimization}(\cdot)$). 
    Fix iteration step $t$. 
    To 
    implement either $\oracle_{\mathrm{grad}}^{\psi, r, s}(\fullit^{(t)}, \eta \gg(\fullit^{(t)}))$ (inside the call to $\mathsf{CompositeBilinearOptimization}(\cdot)$) or $\oracle_{\mathrm{xgrad}}^{\psi / 2, r+s, s}(\fullit^{(t)}, \allowbreak \auxit^{(t)}, \frac{\eta}{2} \gg(\halfit^{(t)}))$, the calls to $\oracle_{\bestrsp}^{\uset, \psi, r}$ are made on $(\hh^{\uset}, \vv)$ with $\vv \in \Delta^m$ and either $\hh^{\uset} = \eta \AA^\top \vv^{(t, 1)} - \nabla_{\uset} r^{\fixedalpha}(\fullit^{(t)})$ or $\hh^{\uset} = \frac{\eta}{2} \AA^\top \vv^{(t, 2)} - \nabla_{\uset} r^{\fixedalpha}(\fullit^{(t)})$. Since $\eta = 1/3$ and $\vv^{(t, 1)}, \vv^{(t, 2)} \in \Delta^m$, we have $\|\AA^\top \vv^{(t, 1)}\|_{\infty}, \allowbreak \|\AA^\top \vv^{(t, 2)}\|_{\infty} \le n \|\AA\|_{\infty}$. Next, note that $\nabla_{\uset} r^{\fixedalpha}(\fullit^{(t)}) = (\vv^{(t, 1)} + \xi \vecone^m) \AA (\vecone^n + \log (\xi \vecone^n + \uu^{(t, 1)})$ and hence $\|\nabla_{\uset} r^{\fixedalpha}(\fullit^{(t)})\|_{\infty} \le n (1+\log (\max(\frac{1}{\xi}, \usetBound + \xi))) \|\AA\|_{\infty} \le 2 n B \|\AA\|_{\infty}$. Hence, the vectors $\hh^{\uset}$ given as input to $\oracle_{\bestrsp}^{\uset, \psi, r}$ satisfy $\|\hh^{\uset}\|_{\infty} \le 3 n B \|\AA\|_{\infty}$, as needed. 
\end{proof}

\section{Composite \texorpdfstring{$\ell_{q, p}$}{Lqp}-Regression in \texorpdfstring{$\O_{q, p}(k)$}{OHat(k)} Query Complexity}
\label{sec:comp_ellqp_regr}

In this section, we present an $\O_{q, p}(k)$-query complexity high-accuracy algorithm for solving composite $\ell_{q, p}$ regression (\Cref{prob:comp_ellqp_regr}, restated below) and prove \Cref{thm:comp_ellqp_regr} given below:
\compLqpRegrProb*

\begin{restatable}{theorem}{compLqpRegrAlg}
\label{thm:comp_ellqp_regr}
Consider the setting of \Cref{prob:comp_ellqp_regr} where $1 < q \le 2 \le p$, $p \in \Z$ is odd with $p = O(\sqrt{\log m})$ and $\coefnorm = \exp(\O(1))$. 
Given a $(\Time_j, \delSCO)$ SCO $\oracle_j$, on $\psi_j$ over $\xset_j$ for every $j \in [k]$,
there is an algorithm that outputs $\algoutput{\XX} \in \xset$ such that
\begin{align*}
    \cE_{q, p}^{\coefnorm}(\algoutput{\XX}) \le \min_{\XX \in \xset} \cE_{q, p}^{\coefnorm}(\XX) + 3 k \delSCO \,.
\end{align*}
The algorithm runs in time
\begin{align*}
    O\left(\sum_{j \in [k]} \left(\Time_j + m\right) \cdot \left(\frac{p}{q - 1}\right)^{O\left(\frac{1}{q-1}\right)} \cdot \log\left(\frac{k|\min_{\XX \in \xset} \cE_{q, p}^{\coefnorm}(\XX)|}{\delSCO}\right)\right) \,.
\end{align*}
\end{restatable}

We design the algorithm by combining the iterative refinement framework for the $\|\cdot\|_{q, p}$ objective~\cite{chen2023high} with a sequential block minimization procedure. 

\subsection{\texorpdfstring{$\ell_{q, p}$}{Lqp}-Iterative Refinement}

The algorithm for \Cref{thm:comp_ellqp_regr} computes a solution in an iterative fashion. It starts with an initial feasible solution $\XX^{(0)}$ and, at each iteration $t$, approximately solves $\min_{\bDelta: \XX^{(t)} + \bDelta \in \xset} \cE_{q, p}^{\coefnorm}(\XX^{(t)} + \bDelta)$. 
A natural approach in iterative refinement 
is to approximate $\cE_{q, p}^{\coefnorm}(\XX^{(t)} + \bDelta)$ with a simpler objective and approximately minimize it. 
Specifically, we will prove that
\begin{align}
\cE_{q, p}^{\coefnorm}(\XX + \bDelta)
\approx \cE_{q, p}^{\coefnorm}(\XX) + \psi(\XX + \bDelta) - \psi(\XX) + \sum_{i \in [m], j \in [k]} c_{ij}(\bDelta_{ij}) + \coefnorm \sum_{i \in [m]} \left(\sum_{j \in [k]} \phi_{j}(\bDelta_{j})\right)^p
\end{align}
for some convex functions $\{c_{ij}\}$ and $\{\phi_{ij}\}$ depending on $\XX.$
To approximately minimize the RHS (for $\XX = \XX^{(t)}$), we use a sequential block minimization procedure, which we discuss in \Cref{sec:resP}. The convergence rate of the method depends on the quality of the approximation achieved by the approximate minimization performed in each iteration. 

Our approximation of $\|\XX+\bDelta\|_{q, p}^{pq} - \|\XX\|_{q, p}^{pq}$ arises from the $\gamma_q(\cdot)$ function, an approximation of the Bregman divergence of $|x|^q$ for $1 < q \le 2.$
\begin{definition}[$\gamma_q(\cdot)$ functions, \cite{bubeck2018homotopy}]
\label{def:gamma}
\begin{align*}
    \gamma_q(x; f) \defeq \begin{cases}
        \frac{q}{2} f^{q-2}x^2 \text{, if $|x| \le f$} \\
        |x|^q - (1 - \frac{q}{2})f^q \text{, otherwise}
    \end{cases}
\end{align*}
\end{definition}

Using $\gamma_q(\cdot)$ functions, we prove the following lemmas regarding the Bregman divergence of $\|\XX+\bDelta\|_{q, p}^{pq}.$

\begin{restatable}[Bregman Divergence of $\|\xx\|_q^{pq}$: Lower Bound]{lemma}{qpIRLB}
\label{lem:qpIRLB}
For $1 < q \le 2 \le p$ and $\xx, \yy \in \R^d$,
\begin{align*}
V^{\|\cdot\|_q^{pq}}_{\xx}(\xx+\yy)
&= \norm{\xx + \yy}_q^{pq} - \norm{\xx}_q^{pq} - pq \norm{\xx}_q^{(p-1)q}\l\xx^{q-1}, \yy\r \\
&\ge \frac{p(q-1)}{20}\norm{\xx}_q^{(p-1)q}\sum_{j \in [d]} \gamma_q(\yy_j; \xx_j) + \left(\frac{q-1}{100}\right)^{p}\left(\sum_{j \in [d]} \gamma_q(\yy_j; \xx_j)\right)^p
\end{align*}
\end{restatable}

\begin{restatable}[Bregman Divergence of $\|\xx\|_q^{pq}$: Upper Bound]{lemma}{qpIRUB}
\label{lem:qpIRUB}
For $1 < q \le 2 \le p$ and $\xx, \yy \in \R^d$, 
\begin{align*}
V^{\|\cdot\|_q^{pq}}_{\xx}(\xx+\yy)
&= \norm{\xx + \yy}_q^{pq} - \norm{\xx}_q^{pq} - pq \norm{\xx}_q^{(p-1)q}\l\xx^{q-1}, \yy\r \\
&\le 200 \norm{\xx}_q^{(p-1)q}\sum_{j \in [d]} \gamma_q(p\yy_j; \xx_j) + 64^p\left(\sum_{j \in [d]} \gamma_q(p\yy_j; \xx_j)\right)^p
\end{align*}
\end{restatable}

We defer the proofs of \Cref{lem:qpIRUB} and \Cref{lem:qpIRLB} to \Cref{sec:proofQPIR}.
Next, we define the \emph{residual objective} $\cR(\bDelta; \XX)$ that approximates the objective difference $\cE_{q, p}^{\coefnorm}(\XX+\bDelta) - \cE_{q, p}^{\coefnorm}(\XX)$, along with the \emph{residual problem}. 

\begin{problem}[Residual Problem]
\label{def:resP}
Consider the setting of \Cref{prob:comp_ellqp_regr} when $1 < q \le 2 \le p.$
Given any feasible solution $\XX \in \xset$, the \emph{residual problem} w.r.t.\ the current solution $\XX$ asks to solve
\begin{align*}
    \min_{\bDelta: \XX + \bDelta \in \xset} \cR(\bDelta; \XX)
\end{align*}
where
\begin{align*}
\cR(\bDelta; \XX) \defeq \sum_{j \in [k]} \psi_j(\XX_j + \bDelta_j) - \psi_j(\XX_j) + \sum_{i \in [m], j \in [k]} c_{ij}(\bDelta_{ij}) + 64^p \coefnorm \sum_{i \in [m]} \left(\sum_{i \in [m]} \gamma_q(p\bDelta_{ij}; \XX_{ij})\right)^p
\end{align*}
and $c_{ij}(x) \defeq \gg_{ij}x + \hh_{ij}\gamma_q(px; \XX_{ij})$, $\gg_{ij} \defeq \coefnorm pq \norm{\XX_i}_q^{(p-1)q}\XX_{ij}^{q-1}$, and $\hh_{ij} \defeq 200 \coefnorm \norm{\XX_i}_q^{(p-1)q}.$
\end{problem}
Note that $\bDelta = 0$ is always a feasible residual solution and, therefore, the optimal residual value $\min \cR(\bDelta; \XX)$ is always non-positive.
Also, by convexity, if $\XX + \bDelta$ is feasible, so is $\XX + t \bDelta$ for any $t \in [0, 1].$
The following lemma states how well $\cR(\bDelta; \XX)$ approximates $\cE_{q, p}^{\coefnorm}(\XX+\bDelta) - \cE_{q, p}^{\coefnorm}(\XX).$

\begin{lemma}
\label{lem:resApprox}
For any feasible solution $\XX$ and feasible update $\bDelta$ (i.e., $\XX + \bDelta \in \xset$), we have
\begin{align*}
\cE_{q, p}^{\coefnorm}(\XX + \bDelta) - \cE_{q, p}^{\coefnorm}(\XX) &\le \cR(\bDelta; \ff)\text{, and} \\
\cE_{q, p}^{\coefnorm}(\XX + \bDelta) - \cE_{q, p}^{\coefnorm}(\XX) &\ge \apxfactor \cR\left(\frac{\bDelta}{\apxfactor}; \XX\right)\text{, where } \apxfactor = \left(\frac{6400 p^q}{q-1}\right)^{1/(q-1)} \\
\end{align*}
\end{lemma}
To prove \Cref{lem:resApprox}, we need the following fact regarding the range of $\gamma_q(tx; f) / \gamma_q(x; f)$ as a function of $t \ge 1$:
\begin{fact}[Lemma 3.3, ~\cite{akps19}]
\label{fact:gammaScale}
For any $x, f \in \R$ and $t \ge 1$, we have
$t^{q}\gamma_q(x; f) \le \gamma_q(tx; f) \le t^2 \gamma_q(x; f)$
\end{fact}
We are now ready to prove \Cref{lem:resApprox}.
\begin{proof}[Proof of \Cref{lem:resApprox}]

The upper bound follows from the definition of $\cR(\bDelta; \XX)$ and \Cref{lem:qpIRUB}.
For the lower bound, we observe that
\begin{align*}
\cE_{q, p}^{\coefnorm}(\XX + \bDelta) - \cE_{q, p}^{\coefnorm}(\XX) 
= \psi(\XX + \bDelta) - \psi(\XX) + \coefnorm (\norm{\XX + \bDelta}_{q, pq}^{pq} - \norm{\XX}_{q, pq}^{pq})
\end{align*}
\Cref{lem:qpIRLB} gives that
\begin{align*}
\norm{\XX + \bDelta}_{q, pq}^{pq} - \norm{\XX}_{q, pq}^{pq}
&= \sum_{i \in [m]} \left( \norm{\XX_i + \bDelta_i}_q^{pq} - \norm{\XX_i}_q^{pq}\right) \\
&\ge \sum_{i \in [m], j \in [k]} \norm{\XX_i}_q^{(p-1)q}\left(pq\XX_{ij}^{q-1}\bDelta_{ij} + \frac{p(q-1)}{4000}\gamma_q(\bDelta_{ij}; \XX_{ij})\right) \\
&+ \left(\frac{q-1}{100}\right)^{p}\sum_{i \in [m]} \left(\sum_{i \in [m]} \gamma_q(\bDelta_{ij}; \XX_{ij})\right)^p
\end{align*}

On the other hand, we have
\begin{align*}
\apxfactor \cR\left(\frac{1}{\apxfactor}\bDelta; \XX\right)
&= \apxfactor \left(\psi\left(\XX + \frac{1}{\apxfactor}\bDelta\right) - \psi(\XX)\right) + \apxfactor\l\gg, \frac{1}{\apxfactor}\bDelta\r \\
&+ \sum_{i \in [m], j \in [k]} \apxfactor \hh_{ij} \gamma_q\left(\frac{p}{\apxfactor}\bDelta_{ij}; \XX_{ij}\right) + 64^p \sum_{i \in [m]} \left(\sum_{j \in [k]} \gamma_q\left(\frac{p}{\apxfactor}\bDelta_{ij}; \XX_{ij}\right)\right)^p
\end{align*}

Because $\apxfactor \ge 1$ and $\psi_{j}$'s are convex, we have $\frac{1}{\apxfactor}(\psi(\XX + \bDelta) - \psi(\XX)) \ge \psi(\XX + \frac{1}{\apxfactor}\bDelta) - \psi(\XX).$
Because $\apxfactor \ge p$, we have, by \Cref{fact:gammaScale}, that
\begin{align*}
\apxfactor \gamma_q\left(\frac{p}{\apxfactor}\bDelta_{ij}; \XX_{ij}\right) \le \apxfactor \frac{p^q}{\apxfactor^q} \gamma_q\left(\bDelta_{ij}; \XX_{ij}\right) = \apxfactor^{1-q}p^q\gamma_q\left(\bDelta_{ij}; \XX_{ij}\right) \le \frac{q-1}{6400} \gamma_q(\bDelta_{ij}; \XX_{ij})
\end{align*}
The lemma follows from direct calculations.
\end{proof}

The following lemma relates the multiplicative-additive error of a solution $\algoutput{\bDelta}$ to the residual problem (\Cref{def:resP}) with respect to a current $\XX \in \xset$ to the progress in function error from $\XX$ to $\XX + \algoutput{\bDelta}$ in the context of \Cref{prob:comp_ellqp_regr}. 

\begin{lemma}
\label{lem:convRate}
Let $\XX^{\star}$ be the optimal solution to \eqref{eq:comp_ellqp_regr} and $\XX$ be the current feasible solution.
For some $\alpha > 1$ and $\delta > 0$, let $\algoutput{\bDelta}$ be a feasible solution (i.e., $\XX + \algoutput{\bDelta} \in \xset$) such that $\cR(\algoutput{\bDelta}; \XX) \le \delta + \frac{1}{\alpha}\cR(\bDelta^{\star}; \XX)$, where $\bDelta^{\star} = \argmin_{\bDelta: \XX + \bDelta \in \xset} \cR(\bDelta; \XX)$. Then, we have
\begin{align*}
    \cE_{q, p}^{\coefnorm}(\XX + \algoutput{\bDelta}) - \cE_{q, p}^{\coefnorm}(\XX^{\star}) \le \left(1 - \frac{1}{\alpha \apxfactor}\right)\left(\cE_{q, p}^{\coefnorm}(\XX) - \cE_{q, p}^{\coefnorm}(\XX^{\star})\right) + \delta
\end{align*}
\end{lemma}
\begin{proof}
\Cref{lem:resApprox} and the approximation ratio of $\algoutput{\bDelta}$ yield
\begin{align*}
\cE_{q, p}^{\coefnorm}(\XX + \algoutput{\bDelta}) - \cE_{q, p}^{\coefnorm}(\XX)
\le \cR(\algoutput{\bDelta}; \XX) 
\le \frac{1}{\alpha}\cR\left(\frac{\XX^{\star} - \XX}{\apxfactor}; \XX\right) + \delta
\le \frac{\cE_{q, p}^{\coefnorm}(\XX^*) - \cE_{q, p}^{\coefnorm}(\XX)}{\alpha\apxfactor} + \delta
\end{align*}
Therefore, we have
\begin{align*}
\cE_{q, p}^{\coefnorm}(\XX + \algoutput{\bDelta}) - \cE_{q, p}^{\coefnorm}(\XX^{\star})
&= \cE_{q, p}^{\coefnorm}(\XX) - \cE_{q, p}^{\coefnorm}(\XX^{\star}) + \cE_{q, p}^{\coefnorm}(\XX + \algoutput{\bDelta}) - \cE_{q, p}^{\coefnorm}(\XX) \\
&\le \cE_{q, p}^{\coefnorm}(\XX) - \cE_{q, p}^{\coefnorm}(\XX^{\star}) + \frac{\cE_{q, p}^{\coefnorm}(\XX^*) - \cE_{q, p}^{\coefnorm}(\XX)}{\alpha\apxfactor} + \delta \\
&= \left(1 - \frac{1}{\alpha \apxfactor}\right)\left(\cE_{q, p}^{\coefnorm}(\XX) - \cE_{q, p}^{\coefnorm}(\XX^{\star})\right) + \delta
\end{align*}
\end{proof}

As a corollary of \Cref{lem:convRate}, we obtain the following result regarding the number of iterations consisting of approximately solving residual problems (\Cref{def:resP}) that ensure a solution to \Cref{prob:comp_ellqp_regr} of suitable additive error. 
\begin{corollary}
\label{coro:convRate}
Suppose we have access to an algorithm that solves the residual problem (\Cref{def:resP}) up to $(\alpha, \delta)$ multiplicative-additive approximation in $\cT_{\alpha}(m, k)$ time. There is an algorithm that solves \Cref{prob:comp_ellqp_regr} to additive error $3\delta$ in time 
\begin{align*}
    \cT_{\alpha}(m, k) \cdot \alpha \cdot \left(\frac{p}{q - 1}\right)^{O\left(\frac{1}{q-1}\right)} \cdot \log\left(\frac{k|\cE_{q, p}^{\coefnorm}(\XX^{\star})|}{\delta}\right) \,.
\end{align*}

\end{corollary}
\begin{proof}
We need an initial solution so that we can start the iterative refinement process.
The solution is initialized to minimize the following separable objective derived from \eqref{eq:comp_ellqp_regr}:
\begin{align*}
    \XX^{(0)} \defeq \argmin_{\XX \in \xset}  \sum_{j \in [k]} \psi_j(\XX_j) + \sum_{i \in [m], j \in [k]} |\XX_{ij}|^{pq}
\end{align*}
This can be computed by querying each oracle $\oracle_i$ once with error parameter $\delta$ set to by a constant.
Note that, for any $\XX$, we have
\begin{align*}
\sum_{i \in [m], j \in [k]} |\XX_{ij}|^{pq} \le \sum_{i \in [m]} \left(\sum_{j \in [k]} |\XX_{ij}|^{q}\right)^p \le k^p \sum_{i \in [m], j \in [k]} | \XX_{ij}|^{pq}
\end{align*}
This implies that $\XX^{(0)}$ is a $k^p$-approximate solution to \eqref{eq:comp_ellqp_regr}, i.e.,
\begin{align*}
    \cE_{q, p}^{\coefnorm}(\XX^{(0)}) - \cE_{q, p}^{\coefnorm}(\XX^{\star}) \le k^{p} |\cE_{q, p}^{\coefnorm}(\XX^{\star})|
\end{align*}

Starting at $\XX^{(0)}$, we run the iterative refinement process for $T$ iterations, each of which consisting of a batch of cyclic updates.
In the $t$-th iteration, we use the given approximate residual problem solver to compute $\bDelta^{(t)}$ so that $\XX^{(t-1)} + \bDelta^{(t)} \in \xset$ and
\begin{align*}
    \cR\left(\bDelta^{(t)}; \XX^{(t-1)}\right) \le \frac{1}{\alpha} \min_{\bDelta: \XX^{(t-1)} + \bDelta \in \xset} \cR\left(\bDelta; \XX^{(t-1)}\right) + \delta
\end{align*}
Then, we set $\XX^{(t)} \gets \XX^{(t-1)} + \bDelta^{(t)}.$

\Cref{lem:convRate} gives that
\begin{align*}
    \cE_{q, p}^{\coefnorm}(\XX^{(T)}) - \cE_{q, p}^{\coefnorm}(\XX^{\star}) \le \left(1 - \frac{1}{\alpha \apxfactor}\right)^T\left(\cE_{q, p}^{\coefnorm}(\XX^{(0)}) - \cE_{q, p}^{\coefnorm}(\XX^{\star})\right) + \frac{\delta}{1 - \frac{1}{\alpha \apxfactor}} \le \left(1 - \frac{1}{\alpha \apxfactor}\right)^T k^p |\cE_{q, p}^{\coefnorm}(\XX^{\star})| + 2\delta
\end{align*}
The lemma follows by setting $T = O(\alpha \apxfactor \log \frac{k^p |\cE_{q, p}^{\coefnorm}(\XX^{\star})|}{\delta}).$
\end{proof}

The next section presents a sequential block minimization procedure that solves the residual problems (\Cref{def:resP}) to $(m^{o(1)}, k \delSCO / 2)$ multiplicative-additive approximation by querying each SCO $\oracle_j$ once.

\subsection{Residual Problem via Sequential Block Minimization}
\label{sec:resP}
This subsection presents an algorithm that solves the residual problem (\Cref{def:resP}) approximately by calling each SCO $\oracle_j$ once. 
This is formalized as the follow lemma:
\begin{lemma}
\label{lem:approxResSolver}
Consider the setting of \Cref{def:resP}.
Given a $(\Time_j, \delSCO)$ SCO $\oracle_j$ on $\psi_j$ over $\xset_j$ for every $j \in [k]$, there is an algorithm that, in $O(\sum_{j \in [k]} (\Time_j+m) )$ time, outputs $\algoutput{\bDelta}$ with $\XX + \algoutput{\bDelta} \in \xset$ such that
\begin{align*}
\cR\left(\algoutput{\bDelta}; \XX\right) \le \frac{k \delSCO}{2} + \frac{1}{200p}  \min_{\bDelta: \XX + \bDelta \in \xset}\cR(\bDelta; \XX) \,
\end{align*}
\end{lemma}

We design an algorithm that can, in fact, approximately optimize a more general class of objectives categorized as follows:

\begin{problem}[Centered $\ell_p$-Convex Regression]
\label{prob:ellp_cvx_regr}
Consider the setting of \Cref{def:resP}, where we are given functions $\phi_{ij}$, with $c_{ij}(0) = \phi_{ij}(0) = 0, \forall i \in [m], j \in [k]$ and $\phi_{ij}(x) \ge 0, \forall i \in [m], j \in [k]$. 
Given some $\XX \in \xset$, the \emph{centered $\ell_{p}$-convex regression problem} asks to solve
\begin{align}
\label{eq:ellp_cvx_regr}
\min_{\bDelta: \XX + \bDelta \in \xset} \cF(\bDelta) \defeq \psi(\XX+\bDelta) - \psi(\XX) + \sum_{i \in [m], j \in [k]} c_{ij}(\bDelta_{ij}) + \sum_{i \in [m]} \left(\sum_{j \in [k]} \phi_{ij}(\bDelta_{ij})\right)^p
\end{align}
\end{problem}

Without the non-separable part involving $\phi_{ij}(\cdot)$'s, one can solve each block independently and compute a $\delSCO$-optimal solution using the given SCOs $\{\oracle_j\}_{j \in [k]}$.
To deal with the non-separable part, we utilize a sequential block minimization procedure and update the blocks one at a time. 
\begin{lemma}[Approximate Centered $\ell_p$-Convex Regression Solver]
\label{lem:approx_ellp_cvx_regr}
Consider the setting of \Cref{prob:ellp_cvx_regr}. Suppose that for any $i \in [m], j \in [k]$, $c_{ij}(x), c_1 \phi_{ij}(x), c_1 (\phi_{ij}(x) + c_2)^p$ are $m$-decomposable for any $c_1 \ge 0, |c_1, c_2| \le \exp(\O(1))$.
Given a $(\Time_j, \delSCO)$ SCO $\oracle_j$ on $\psi_j$ over $\xset_j$ for every $j \in [k]$, 
\Cref{algo:approx_ellp_cvx_regr} outputs, in $O(\sum_{j \in [k]} (\Time_j + m) )$ time, a feasible solution $\algoutput{\bDelta}$ (i.e., $\XX + \algoutput{\bDelta} \in \xset$) such that
\begin{align*}
    \cF(\algoutput{\bDelta}) \le \frac{k \delSCO}{2} + \frac{1}{200p} \cdot \min_{\bDelta: \XX + \bDelta \in \xset} \cF(\bDelta)\,.
\end{align*}
\end{lemma}

The pseudocode of \Cref{algo:approx_ellp_cvx_regr} is provided below. 

\begin{algorithm}[H]
  \caption{Approximate Centered $\ell_p$-Convex Regression Solver \label{algo:approx_ellp_cvx_regr}}
  \SetKwProg{Globals}{global variables}{}{}
  \SetKwProg{Proc}{procedure}{}{}
  \Proc{$\textrm{BlockMinimization}(\xset, \XX, \{\oracle_j\}_{j \in [k]}, \{\psi_j(\cdot)\}_{j \in [k]}, \{c_{ij}(\cdot)\}_{i \in [m], j \in [k]}, \{\phi_{ij}(\cdot)\}_{i \in [m], j \in [k]}, \delSCO)$}{
    Initialize $\ww^{(0)} \in \R^m$ to be an all-zero vector. \\
    \For{block $j = 1, 2, \ldots, k$}{
    	Query  $\oracle_j$ to compute $\bDelta_j \in \R^m$ (the update for $\xset_j$) that minimizes
    \begin{align}\label{eq:perBlockSubP}
    \min_{\XX_j + \bDelta_j \in \xset_j} \psi_j(\XX_j + \bDelta_j) - \psi_j(\XX_j) + \sum_{i \in [m]} c_{ij}(\bDelta_{ij}) + (\ww^{(j-1)}_i + \phi_{ij}(\bDelta_{ij}))^p\,.
    \end{align}\\
    Update $\ww^{(j)}_i \gets \ww^{(j - 1)}_i + \phi_{ij}(\bDelta_{ij}), \forall i \in [m].$ \\
    }
    Collect updates $\bDelta = (\bDelta_1; \bDelta_2; \ldots; \bDelta_k) \in \R^{m \times k}.$ \label{line:collect_updates} \\
    \Return $\frac{1}{2} \bDelta$
  }
\end{algorithm}

We analyze the algorithm via the following potential, which is defined for the iterates $\bDelta_j$ generated by \Cref{algo:approx_ellp_cvx_regr}:
\begin{align}\label{eq:potential_pqnorm}
\Phi^{(j)} \defeq \sum_{\ell \le j} \underbrace{g_{\ell}(\bDelta_{\ell})}_{\defeq \psi_{\ell}(\XX_{\ell} + \bDelta_{\ell}) - \psi_{\ell}(\XX_{\ell}) + \sum_{i \in [m]} c_{i\ell}(\bDelta_{i\ell})}  + \norm{\ww^{(j)}}_p^p \forall j \in [k]
\end{align}
First, we show that $\Phi^{(k)}$ does not grow too quickly, as captured by the lemma below. We will then show how to use this upper bound on $\Phi^{(k)}$ to prove \Cref{lem:approx_ellp_cvx_regr}.

\begin{lemma}[Upper bounding $\Phi^{(k)}$]
\label{lem:boundPhik}
Consider the setting of \Cref{lem:approx_ellp_cvx_regr}. 
Let $\bDelta^{\star}$ be the optimal solution to \eqref{eq:ellp_cvx_regr} and potentials $\Phi^{(j)}$ defined as in \Cref{eq:potential_pqnorm} for the iterates $\bDelta_j$ generated by a call of \Cref{algo:approx_ellp_cvx_regr}.
We have 
\begin{align*}
\Phi^{(k)}
\le \frac{1}{100p} \cF(\bDelta^{\star}) + \frac{1}{2} \norm{\ww^{(k)}}_p^p + k \delSCO
\end{align*}
\end{lemma}

\begin{proof}
Define $\eta = \frac{1}{100p}.$
For any $j \in [k]$, we have, be definition, that
\begin{align*}
\Phi^{(j)} 
&= \sum_{\ell \le j} g_{\ell}(\bDelta_{\ell}) + \sum_{i \in [m]} (\ww^{(j)}_i)^p \\
&= \sum_{\ell \le j} g_{\ell}(\bDelta_{\ell}) + \sum_{i \in [m]} \left(\ww^{(j - 1)}_i + \phi_{ij}(\bDelta_{ij})\right)^p \\
&\le \sum_{\ell < j} g_{\ell}(\bDelta_{\ell}) + g_j(\eta \bDelta_j^{\star}) + \sum_{i \in [m]} \left(\ww^{(j - 1)}_i + \phi_{ij}(\eta\bDelta_{ij}^{\star})\right)^p + \delSCO \\
&= \Phi^{(j-1)} - \|\ww^{(j-1)}\|_p^p + g_j(\eta \bDelta_j^{\star}) + \sum_{i \in [m]} \left(\ww^{(j - 1)}_i + \phi_{ij}(\eta\bDelta_{ij}^{\star})\right)^p + \delSCO
\end{align*}
where the inequality comes from that $\bDelta_j$ minimizes \eqref{eq:perBlockSubP}.

To bound the summation involving $\ww^{(j-1)}$, we use the fact that $(1+x)^p \le 1+3px+3(px)^p$ for $x \ge 0$ and $p \ge 2$ and have
\begin{align*}
\sum_{i \in [m]} \left(\ww^{(j - 1)}_i + \phi_{ij}(\eta\bDelta_{ij}^{\star})\right)^p
&\le \sum_{i \in [m]} \left[\left(\ww^{(j - 1)}_i\right)^p + 3p(\ww^{(j-1)}_e)^{p-1} \phi_{ij}(\eta\bDelta_{ij}^{\star}) + 3p^p\phi_{ij}(\eta\bDelta_{ij}^{\star})^p\right] \\
&\le \|\ww^{(j-1)}\|_p^p + \sum_{i \in [m]} \left[3p(\ww^{(k)}_e)^{p-1} \phi_{ij}(\eta\bDelta_{ij}^{\star}) + 3p^p\phi_{ij}(\eta\bDelta_{ij}^{\star})^p\right]\,.
\end{align*}
where we use the fact that $\ww^{(j)}$ is non-decreasing in $j$ entry-wise.
Combining with the previous inequality for $\Phi^{(j)}$, we have, by induction on index $j \in [k]$, that
\begin{align}
\label{eq:PhiKmid}
\Phi^{(k)}
&\le \sum_{j \in [k]} g_j(\eta \bDelta_j^{\star}) +  \sum_{i \in [m], j \in [k]} \left[3p(\ww^{(k)}_i)^{p-1} \phi_{ij}(\eta \bDelta^{\star}_{ij}) + 3p^p\phi_{ij}(\eta \bDelta^{\star}_{ij})^p\right] + k \delSCO\,.
\end{align}

Here, we use Young's inequality for products to further bound $\sum_{i \in [m], j \in [k]}(\ww^{(k)}_i)^{p-1}  \phi_{ij}(\eta \bDelta^{\star}_{ij})$ as follows:
\begin{align*}
\sum_{i \in [m], j \in [k]}(\ww^{(k)}_i)^{p-1}  \phi_{ij}(\eta \bDelta^{\star}_{ij})
&= \sum_{i \in [m]}(\ww^{(k)}_i)^{p-1} \left(\sum_{i \in [m]} \phi_{ij}(\eta \bDelta^{\star}_{ij})\right) \\
&\le \sum_{i \in [m]}\frac{p-1}{10p^2} (\ww^{(k)}_e)^{p} + 10^{p-1}p^{p-2}\left(\sum_{j \in [k]} \phi_{ij}(\eta \bDelta^{\star}_{ij})\right)^p \\
&\le \frac{1}{10p} \norm{\ww^{(k)}}_p^p + 10^{p-1}p^{p-2} \sum_{i \in [m]} \left(\sum_{j \in [k]} \phi_{ij}(\eta \bDelta^{\star}_{ij})\right)^p
\end{align*}
where we use the following fact derived from Young's inequality for products:
\begin{align*}
    A^{\frac{p-1}{p}}B^{\frac{1}{p}} = \left(\frac{A}{10p}\right)^{\frac{p-1}{p}}\left(B(10p)^{p-1}\right)^{\frac{1}{p}} \le \frac{p-1}{10p^2}A + 10^{p-1}p^{p-2}B, \forall A, B > 0.
\end{align*}

This allows us to further bound the part of \eqref{eq:PhiKmid} involving $\phi_{ij}(\eta \bDelta^{\star})$ as follows:
\begin{align*}
&\sum_{i \in [m], j \in [k]} \left[3p(\ww^{(k)}_i)^{p-1} \phi_{ij}(\eta \bDelta^{\star}_{ij}) + 3p^p\phi_{ij}(\eta \bDelta^{\star}_{ij})^p\right] \\
\le &\sum_{i \in [m], j \in [k]} 3p^p\phi_{ij}(\eta \bDelta^{\star}_{ij})^p + \frac{3}{10} \norm{\ww^{(k)}}_p^p + (10p)^p \sum_{i \in [m]} \left(\sum_{j \in [k]} \phi_{ij}(\eta \bDelta^{\star}_{ij})\right)^p \\
\le &(20p)^p \sum_{i \in [m]} \left(\sum_{j \in [k]} \phi_{ij}(\eta \bDelta^{\star}_{ij})\right)^p + \frac{1}{2} \norm{\ww^{(k)}}_p^p \\
\le &(20p)^p \eta^p\sum_{i \in [m]} \left(\sum_{j \in [k]} \phi_{ij}(\bDelta^{\star}_{ij})\right)^p + \frac{1}{2} \norm{\ww^{(k)}}_p^p \\
\end{align*}
where the 2nd inequality comes from that $\sum_{i \in [m], j \in [k]} \phi_{ij}(\eta \bDelta^{\star})^p \le \sum_{i \in [m]} (\sum_{j \in [k]} \phi_{ij}(\eta \bDelta^{\star}))^p$ and $3p^p + (10p)^p \le (20p)^p$, and the 3rd inequality comes from that $\phi_{ij}(\eta x) \le \eta \phi_{ij}(x)$ by convexity and positivity of $\phi_{ij}(\cdot)$, and by the fact that $\phi_{ij}(0) = 0$.

Plugging this back to \eqref{eq:PhiKmid} yields
\begin{align*}
\Phi^{(k)}
&\le \sum_{j \in [k]} g_j(\eta \bDelta_j^{\star}) +  \sum_{i \in [m], j \in [k]} \left[3p(\ww^{(k)}_i)^{p-1} \phi_{ij}(\eta \bDelta^{\star}_{ij}) + 3p^p\phi_{ij}(\eta \bDelta^{\star}_{ij})^p\right] + k\delSCO \\
&\le \eta \sum_{j \in [k]} g_j(\bDelta_j^{\star})  + (20p)^p \eta^p\sum_{i \in [m]} \left(\sum_{j \in [k]} \phi_{ij}(\bDelta^{\star}_{ij})\right)^p + \frac{1}{2} \norm{\ww^{(k)}}_p^p + k\delSCO \\
&\le \eta \left(\sum_{j \in [k]} g_j(\bDelta_j^{\star})  + \sum_{i \in [m]} \left(\sum_{j \in [k]} \phi_{ij}(\bDelta^{\star}_{ij})\right)^p\right) + \frac{1}{2} \norm{\ww^{(k)}}_p^p + k\delSCO \\
&= \frac{1}{100p} \cF(\bDelta^{\star}) + \frac{1}{2} \norm{\ww^{(k)}}_p^p + k\delSCO
\end{align*}
where the 3rd inequality comes from that $(20\eta p)^p \le \eta.$
This concludes the proof.
\end{proof}

Using \Cref{lem:boundPhik}, we show that the solution returned by \Cref{algo:approx_ellp_cvx_regr} is a suitable approximate solution for solving \Cref{prob:ellp_cvx_regr}, which is captured by the lemma below. 
\begin{lemma}
\label{lem:mwuApproxRatio}
Consider the setting of \Cref{lem:approx_ellp_cvx_regr}. 
Let $\bDelta$ be the vector generated by \Cref{line:collect_updates} during a call to \Cref{algo:approx_ellp_cvx_regr} and $\bDelta^*$ be the optimal solution to \eqref{eq:ellp_cvx_regr}. 
We have
\begin{align*}
\cF\left(\frac{1}{2}\bDelta\right)
\le \frac{1}{200p} \cF(\bDelta^{\star}) + \frac{k \delSCO}{2}
\end{align*}
\end{lemma}
\begin{proof}
First, observe that
\begin{align*}
\norm{\ww^{(k)}}_p^p
= \sum_{i \in [m]} \left(\ww^{(k)}_i\right)^p 
= \sum_{i \in [m]} \left(\sum_{j \in [k]} \phi_{ij}(\bDelta_{ij})\right)^p
\end{align*}
Via rearrangement, \Cref{lem:boundPhik} implies that
\begin{align}
\label{eq:rearrange}
\sum_{j \in [k]} g_{j}\left(\bDelta_{j}\right) + \frac{1}{2}\sum_{i \in [m]} \left(\sum_{j \in [k]} \phi_{ij}(\bDelta_{ij})\right)^p \le \frac{1}{100p} \cF(\bDelta^{\star}) + k\delSCO
\end{align}
By convexity, we know that $g_{j}(\frac{\xx}{2}) \le \frac{1}{2}g_{j}(\xx)$ for any $\xx$ and $j$ and we have
\begin{align*}
\cF\left(\frac{1}{2}\bDelta\right)
&= \sum_{j \in [k]} g_{j}\left(\frac{1}{2}\bDelta_{j}\right) + \sum_{i \in [m]} \left(\sum_{j \in [k]} \phi_{ij}\left(\frac{1}{2} \bDelta_{ij}\right)\right)^p \\
&\le \frac{1}{2}\sum_{j \in [k]} g_{j}\left(\bDelta_{j}\right) + \frac{1}{2^p}\sum_{i \in [m]} \left(\sum_{j \in [k]} \phi_{ij}\left(\bDelta_{ij}\right)\right)^p \\
&\le \frac{1}{2}\sum_{j \in [k]} g_{j}\left(\bDelta_{j}\right) + \frac{1}{4}\sum_{i \in [m]} \left(\sum_{j \in [k]} \phi_{ij}\left(\bDelta_{ij}\right)\right)^p \\
&\le \frac{1}{200p} \cF(\bDelta^{\star}) + \frac{k\delSCO}{2}
\end{align*}
where the 1st inequality comes from that $\phi_{ij}(\frac{1}{2} x) \le \frac{1}{2} \phi_{ij}(x)$ by convexity, the 2nd inequality comes from that $2^{-p} \le 1/4$, and the last inequality comes from \eqref{eq:rearrange}.
\end{proof}

\Cref{lem:mwuApproxRatio} establishes that \Cref{algo:approx_ellp_cvx_regr} outputs a solution with $(\almostTime(p), k \delSCO / 2)$ multiplicative-additive approximation error.  
\begin{proof}[Proof of \Cref{lem:approx_ellp_cvx_regr}]
\Cref{lem:mwuApproxRatio} proves the approximation ratio of $\frac{1}{2}\bDelta$ and the lemma follows.
\end{proof}

\subsection{Putting it All Together}

Using \Cref{lem:approx_ellp_cvx_regr}, we are ready to prove \Cref{lem:approxResSolver}, along with the main theorem in this section, \Cref{thm:comp_ellqp_regr}. 
\begin{proof}[Proof of \Cref{lem:approxResSolver}]
First, we formulate the residual problem (\Cref{def:resP}) as a centered $\ell_p$-convex regression problem (\Cref{prob:ellp_cvx_regr}) by defining
\begin{align*}
c_{ij}(x) \defeq \gg_{ij}x + \hh_{ij}\gamma_q(px; \XX_{ij})
\text{ and }
\phi_{ij}(x) \defeq 64\gamma_q(px; \XX_{ij})\,.
\end{align*}
for all $i \in [k]$ and $e \in [m].$
First, note that $c_{ij}(0) = \phi_{ij}(0) = 0$ and that $\phi_{ij}(x) \ge 0, \forall x \in \R$, as $\gamma_q$ is non-negative. Second, 
by \Cref{lem:gamma_big_p}, we have that $c_1 \gamma_q(x; \XX_{ij})$, as well as $c_1 (\gamma_q(x; \XX_{ij}) + c_2)^p$ are $m$-decomposable, which also implies that $c_1 \gamma_q(px; \XX_{ij})$ and $c_1 (\gamma_q(px; \XX_{ij}) + c_2)^p$ are $m$-decomposable for any $c_1 \ge 0, |c_1, c_2| \le \exp(\O(1))$. 
Additionally, note that $\gg_{ij}, \hh_{ij} = O(\exp(\otilde(1)))$, as $\coefnorm = O(\coefnorm)$ and $\|\XX\|_{\infty}^{pq} = O(\exp(\otilde(1))), \forall \XX \in \xset$ from our choice of $q, p$. 
Hence, we have that for every $i, j$ and $c_1, c_2 \le O(\exp(\otilde(1)))$, $c_1 c_{ij}(x), c_1 \phi_{ij}(x), c_1 (\phi_{ij}(x) + c_2)^p$ are all $m$-decomposable. 
Thus, 
\Cref{lem:approx_ellp_cvx_regr} gives that the solution $\algoutput{\bDelta}$ outputted by running \Cref{algo:approx_ellp_cvx_regr} satisfies 
\begin{align*}
\cR\left(\algoutput{\bDelta}; \XX\right) \le \frac{1}{200p} \cR(\bDelta^{\star}; \XX) + \frac{k \delSCO}{2}\,.
\end{align*}
Moreover, by \Cref{lem:approx_ellp_cvx_regr}, running \Cref{algo:approx_ellp_cvx_regr} takes time $O(\sum_{j \in [k]} (\Time_j + m) )$, which completes the proof. 
\end{proof}

We now use \Cref{lem:approxResSolver} and \Cref{coro:convRate} to prove \Cref{thm:comp_ellqp_regr}.

\begin{proof}[Proof of \Cref{thm:comp_ellqp_regr}]
Combining \Cref{lem:approxResSolver} with \Cref{coro:convRate} yields an algorithm that solves \Cref{prob:comp_ellqp_regr} to additive error $3 k \delSCO$ in time 
\begin{align*}
\sum_{j \in [k]} \left(\Time_j + m\right) \cdot \left(\frac{p}{q-1}\right)^{O\left(\frac{1}{q-1}\right)} \cdot \log\left(\frac{k|\cE(\XX^{\star})|}{\delSCO}\right) = m^{o(1)} \sum_{j \in [k]} \left(\Time_j + m\right) \cdot \log\left(\frac{k|\cE(\XX^{\star})|}{\delSCO}\right)
\,.
\end{align*}
\end{proof}

\subsection{Bregman Divergence Approximation of \texorpdfstring{$\|\xx\|_q^{pq}$}{Lqp}}
\label{sec:proofQPIR}

In this section, we prove our two main technical lemmas, \Cref{lem:qpIRUB} and \Cref{lem:qpIRLB}.
Our proof relies on the following two lemmas regarding approximating the Bregman divergence of $|x|^q$ and $|x|^p$ respectively.
\begin{lemma}[Iterative Refinement for $1 < q  \leq 2$ \cite{akps19, AKPS22}]
\label{lem:LqIR}
For all $q \in (1, 2]$,
\[
\frac{q-1}{q 2^q} \cdot \gamma_q(x, |f|) \le V^{|\cdot|^q}_f(f+x) \le 2^q \cdot \gamma_q (x, |f|)
\text{ for all }
f, x \in \R\,.
\]
\end{lemma}

\begin{lemma}[Iterative Refinement for $p \geq 2$ \cite{akps19, AKPS22}]
\label{lem:LpIR}
For all $p \ge 2$,
\[
\frac{p}{8}|f|^{p - 2} x^2 + \frac{1}{2^{p+1}} |x|^p\le V^{|\cdot|^p}_f(f+x) \le 2p^2 |f|^{p - 2} x^2 + p^p |x|^p
\text{ for all }
f, x \in \R\,.
\]
\end{lemma}

First, we prove the lower bound on the Bregman divergence of $\|\xx\|_q^{pq}$, which we restate below.
\qpIRLB*
To prove the lemma, we need the following claim that generalizes the Cauchy-Schwarz inequality in terms of $\gamma_q(\cdot).$
\begin{claim}
\label{claim:qpGradSize}
For $1 < q \le 2$ and any vectors $\xx, \yy \in \R^d$, we have
\begin{align*}
    |q\l\xx^{q-1}, \yy\r| \le \sqrt{2q} \norm{\xx}_q^{q/2} \sqrt{\sum_{j \in [d]} \gamma_q(\yy_j; \xx_j)} + q \norm{\xx}_q^{q-1} \left(\sum_{j \in [d]} \gamma_q(\yy_j; \xx_j)\right)^{\frac{1}{q}}
\end{align*}
\end{claim}
\begin{proof}[Proof of \Cref{claim:qpGradSize}]
Without loss of generality, we assume that both $\xx$ and $\yy$ are nonnegative.
First, we observe that, for any $a, b \in \R$, $\gamma_q(a; b) \le a^q$ and $\gamma_q(a;b) \ge a^q/2$ when $|a| > |b|.$

Using the observation and Hölder's inequality, we have that
\begin{align*}
q\l\xx^{q-1}, \yy\r
&= q\sum_{j \in [d]} \xx_j^{q-1} \yy_j \\
&\le q \sum_{\yy_j \le \xx_j} \xx_j^{q-1} \yy_j + q \sum_{\yy_j > \xx_j} \xx_j^{q-1} \yy_j \\
&\le \sqrt{2q} \sum_{\yy_j \le \xx_j} \xx^{\frac{q}{2}}\sqrt{\gamma_q(\yy_j; \xx_j)} + q \sum_{\yy_j > \xx_j} \xx_j^{q-1} \gamma_q(\yy_j; \xx_j)^{\frac{1}{q}} \\
&\le \sqrt{2q} \norm{\xx}_q^{q/2} \sqrt{\sum_{j \in [d]} \gamma_q(\yy_j; \xx_j)} + q \norm{\xx}_q^{q-1} \left(\sum_{j \in [d]} \gamma_q(\yy_j; \xx_j)\right)^{\frac{1}{q}}
\end{align*}
where the 2nd inequality uses the definition of $\gamma_q(y; x)$ (\Cref{def:gamma}).
\end{proof}

Now, we prove \Cref{lem:qpIRLB}.
\begin{proof}[Proof of \Cref{lem:qpIRLB}]
We first define the following quantities:
\begin{align*}
N &\defeq \|\xx\|_q \\
L &\defeq q \l\xx^{q-1}, \yy\r\\
B &\defeq \|\xx+\yy\|_q^q - \|\xx\|_q^q - L\\
G &\defeq \sum_{j \in [d]} \gamma_q(\yy_j; \xx_j) \\
t &\defeq \left(\frac{100}{q-1}\right)^{\frac{q}{q - 1}}
\end{align*}
\Cref{lem:LqIR} implies that $B \ge \frac{q-1}{8}G.$
\Cref{lem:LpIR} yields that
\begin{align*}
\norm{\xx + \yy}_q^{pq}
&= \left(N^q + L + B\right)^p \\
&\ge N^{pq} + p N^{(p-1)q}(L + B) + \frac{p}{8} N^{(p-2)q}(L + B)^2 + \frac{1}{2^{p+1}}(L + B)^p
\end{align*}

If $G \le t N^q$, we have
\begin{align*}
\norm{\xx + \yy}_q^{pq} - N^{pq} - p N^{(p-1)q}L
&\ge p N^{(p-1)q}B \\
&\ge \frac{p(q-1)}{8}N^{(p-1)q}G \\
&\ge \frac{p(q-1)}{16}N^{(p-1)q}G + \frac{p(q-1)}{16t^{p-1}} G^p
\end{align*}
where the last inequality uses the fact that $N^{(p-1)q}G \ge t^{1-p} G^p.$

For the other case where $G > tN^q$, we know that, by \Cref{claim:qpGradSize}, 
\begin{align*}
|L| \le \sqrt{2q} N^{\frac{q}{2}} G^{\frac{1}{2}} + q N^{q-1}G^{\frac{1}{q}} \le \frac{q - 1}{100}G
\end{align*}
Combining with the fact that $\frac{q-1}{8} G \le B \le G$, we know that $|L + B| \ge \frac{q-1}{10}G.$
This implies that
\begin{align*}
\norm{\xx + \yy}_q^{pq} - N^{pq} - p N^{(p-1)q}L
&\ge p N^{(p-1)q}B + \frac{1}{2^{p+1}}(L + B)^p \\
&\ge \frac{p(q-1)}{8}N^{(p-1)q}G + \left(\frac{q-1}{100}\right)^{p} G^p
\end{align*}
The lemma concludes by combining both cases of whether $G > t N^q.$
\end{proof}

Finally, we prove \Cref{lem:qpIRUB}, the upper bound on the Bregman divergence of $\|\xx\|_q^{pq}$ which we restate below.

\qpIRUB*

\begin{proof}[Proof of \Cref{lem:qpIRUB}]
We first define the following quantities:
\begin{align*}
N &\defeq \|\xx\|_q \\
L &\defeq q \l\xx^{q-1}, \yy\r\\
B &\defeq \|\xx+\yy\|_q^q - \|\xx\|_q^q - L\\
G &\defeq \sum_{j \in [d]} \gamma_q(p\yy_j; \xx_j) 
\end{align*}
\Cref{lem:LqIR} implies that $B \le \frac{4}{p} G$ and, therefore, $\|\xx+\yy\|_q^q \le N^q + L + \frac{4}{p}G.$
\Cref{lem:LpIR} yields that
\begin{align*}
\norm{\xx + \yy}_q^{pq}
&\le \left(N^q + L + \frac{4}{p}G\right)^p \\
&\le N^{pq} + p N^{(p-1)q}\left(L + \frac{4}{p}G\right) + 2p^2N^{(p-2)q}\left(L + \frac{4}{p}G\right)^2 + p^p\left(L + \frac{4}{p}G\right)^p
\end{align*}

\Cref{claim:qpGradSize} implies that
\begin{align*}
    |L| = \frac{q}{p}\l\xx^{q-1}, p\yy\r \le \frac{\sqrt{2q}}{p} N^{\frac{q}{2}} G^{\frac{1}{2}} + \frac{q}{p} N^{q-1}G^{\frac{1}{q}}
\end{align*}
If $G \ge N^q$, we have $|L| \le \frac{4}{p}G$ and $|L + \frac{4}{p}G| \le \frac{8}{p}G.$
This implies that
\begin{align*}
V^{\|\cdot\|_q^{pq}}_{\xx}(\xx+\yy)
&= \norm{\xx + \yy}_q^{pq} - N^{pq} - p N^{(p-1)q}L \\
&\le 4 N^{(p-1)q}G + 128N^{(p-2)q}G^2 + 8^p G^p \\
&\le 200 N^{(p-1)q}G + 64^p G^p
\end{align*}
where the last inequality follows by $128 N^{(p-2)q}G^2 \le 16 N^{(p-1)q}G + 16^p G^p$ via Young's inequality for products.

Otherwise, we have $G < N^q$ and $|L| \le \frac{4}{p} N^{\frac{q}{2}}G^{\frac{1}{2}}.$
This implies that
\begin{align*}
    \left|L + \frac{4}{p}G\right| \le \frac{4}{p}N^{\frac{q}{2}}G^{\frac{1}{2}} + \frac{4}{p}G \le \frac{8}{p}N^{\frac{q}{2}}G^{\frac{1}{2}}
\end{align*}
Therefore, we have
\begin{align*}
V^{\|\cdot\|_q^{pq}}_{\xx}(\xx+\yy)
&= \norm{\xx + \yy}_q^{pq} - N^{pq} - p N^{(p-1)q}L \\
&\le 4 N^{(p-1)q}G + 2p^2N^{(p-2)q}\left(L+\frac{4}{p}G\right)^2 + p^p\left(L + \frac{4}{p}G\right)^p \\
&\le 4 N^{(p-1)q}G + 128 N^{(p-1)q}G + 8^p N^{\frac{pq}{2}}G^{\frac{p}{2}} \\
&\le 200 N^{(p-1)q}G + 64^p G^p
\end{align*}
where the last inequality follows from the fact that $8^p N^{\frac{pq}{2}}G^{\frac{p}{2}} \le 8^2 N^{(p-1)q}G + 8^{2p} G^p$ by Young's inequality for products.

The lemma follows from combining both cases of whether $G \ge N^q.$
\end{proof}

\section{Composite \texorpdfstring{$\ell_{1, \infty}$}{L1inf}-Regression and Directed Multi-Commodity Flows}
\label{sec:L1infMin}

In this section, we present the proofs of our main theorems stated in \Cref{sec:intro}.
In \Cref{subsec:1.4solve} we present the algorithm (\Cref{alg:composite_ell1inf_regression})
for approximately solving the composite $\ell_{1, \infty}$ regression problem (\Cref{prob:comp_ell1inf_regr}), along with its error and runtime guarantee (\Cref{thm:approxL1InfMin}).
In \Cref{sec:kcommflow} we present the applications of \Cref{thm:approxL1InfMin}, namely our improved multi-commodity flow results (\Cref{coro:approxConcurrentFlow,coro:approxMaxKCommFlow}). 

\subsection{Algorithm for Composite \texorpdfstring{$\ell_{1, \infty}$}{L1inf}-Regression}
\label{subsec:1.4solve}

First, we formally state the theorem regarding our algorithm for approximately solving the composite $\ell_{1, \infty}$-regression problem (\Cref{prob:comp_ell1inf_regr}). This theorem was previously stated as \Cref{informal:approxL1InfMin}. 

\begin{restatable}{theorem}{ApproxLOneInfMinThm}
\label{thm:approxL1InfMin}
Consider the setting of \Cref{prob:comp_ell1inf_regr} where $\xset_j \subseteq [0, R]^m$ for every $j \in [k]$, $|\psi(\XX)| \in [0, \psibound], \forall \XX \in \xset$, 
and $\XX^{\star} = \argmin_{\XX \in \xset} \cE_{1, \infty}(\XX)$.
Given $0 < R, \psibound \le \poly(m,k)$, $\epsilon, \delta > 1/\poly(m,k)$, and $p \in \Z_{\ge 2}$, 
\Cref{alg:composite_ell1inf_regression} outputs $\algoutput{\XX} \in \xset$ such that 
\begin{align*}
    \psi(\algoutput{\XX}) + \norm{\algoutput{\XX}}_{1, \infty} \le \psi(\XX^{\star}) + (1 + \eps)\norm{\XX^{\star}}_{1, \infty} + \delta \,.
\end{align*}
When $p = 2 \lceil \sqrt{\log m} \rceil + 1,$ given access to a $(\Time_j, \delSCO)$-SCO $\oracle_j$ on $\psi_j$ over $\xset_j$ for every $j \in [k]$, with $\delSCO > \frac{1}{\poly(m, k)}$ small enough, the algorithm can be implemented in $\almostTime(\sum_{j \in [k]} (\Time_j + m) \epsilon^{-1})$ time. 
\end{restatable}

$\mathsf{Composite\ell_{1, \infty}Regression}$ from \Cref{alg:composite_ell1inf_regression}, which is the algorithm that we use to \Cref{thm:approxL1InfMin}, relies 
on the subroutine $\mathsf{Scaled\ell_{1, \infty}Regression}$. 
$\mathsf{Scaled\ell_{1, \infty}Regression}$ takes as input an instance of \Cref{prob:comp_ell1inf_regr}, bound $\psibound$ on the absolute values of $\psi$, error parameters $\epsilon, \delta >0$, $p \in \Z_{\ge 2}$, and $\linfbound$, which can be thought of as a guess on $\norm{\XX^{\star}}_{1, \infty}$. 
It outputs $\algoutput{\XX} \in \xset$ of additive error $\epsilon \linfbound + \delta$ with respect to the composite $\ell_{1, \infty}$-regression objective over the constraint set $\sset = \{\XX \in \xset \cap [0, \linfbound]^{m \times k}: \|\XX\|_{q, p} \le \linfbound \cdot m^{\frac{1}{pq}}\}$, where $q = 1+1/p$, as in \Cref{line:choose_constraintset}. 
To compute such a solution, $\mathsf{Scaled\ell_{1, \infty}Regression}$ calls \Cref{alg:box_simplex} over the space $\sset$ of $\XX \in \xset$ of $\ell_{q,p}$ norm at most $\linfbound \cdot m^{\frac{1}{pq}}$. To implement the steps of the constrained box-simplex solver, it uses an ABRO with respect to $\sset$, whose steps are implemented by applying 
our composite $\ell_{q, p}$-regression solver designed in \Cref{sec:comp_ellqp_regr}, along with a binary search procedure. Hence, parameter $p$ determines what composite $\ell_{q,p}$-regression problems the given instance of composite $\ell_{1,\infty}$-regression is reduced to. An instantiation of $\mathsf{Composite\ell_{1, \infty}Regression}$ that we can efficiently implement occurs when $p = 2 \lceil \sqrt{\log m} \rceil + 1$. 

To approximately solve \Cref{prob:comp_ell1inf_regr}, $\mathsf{Composite\ell_{1, \infty}Regression}$ tries out a sequence of values $\Bar{\linfbound}$ in $[\delta, \boxbound]$ with the property that any 2 consecutive values are a factor of $2$ away. The goal is to find a value $\linfbound^{\star}$, for which the call to $\mathsf{Scaled\ell_{1, \infty}Regression}(\xset, \linfbound^{\star}, \psibound, \{\psi_{ij}\}_{i \in [m], j \in [k]}, \epsilon/3, \delta/2, p)$ is guaranteed to output a good enough $\algoutput{\XX} \in \xset$. 
The pseudocode of \Cref{alg:composite_ell1inf_regression} is presented below. 

\begin{algorithm}
\caption{Composite $\ell_{1, \infty}$ Regression}\label{alg:composite_ell1inf_regression}
\KwData{$\xset = \prod_j \xset_j$, $R, \psibound>0$, functions $\{\psi_{ij}\}_{i \in [m], j \in [k]}$, $\epsilon, \delta > 0$, and $p \in \Z_{\ge 2}$}
\KwResult{$\algoutput{\XX}$ with guarantees in \Cref{thm:approxL1InfMin}}
\SetKwProg{Fn}{Function}{:}{}
\SetKwFunction{compregr}{Composite$\ell_{1, \infty}$Regression}
\SetKwFunction{scaledregr}{Scaled$\ell_{1, \infty}$Regression}
\SetKwFunction{Xgradoracle}{ApxXGradStep}
\SetKwFunction{gradoracle}{ApxGradStep}
\SetKwFunction{boxsimplex}{BoxSimplexOptimizer}
\Fn{\compregr{$\xset, R, \psibound, \{\psi_{ij}\}_{i \in [m], j \in [k]}, \epsilon, \delta, p$}}{
    \For{$\ell \in \lceil \log_2 \frac{6kR}{\delta} \rceil$}{
        $\Bar{\linfbound} \gets 2^{-\ell + 1} k R$ \label{line:choose_scale} \tcp*{$\Bar{\linfbound}$ is our guess for the value of $\|\XX^{\star}\|_{1, \infty}$} 
        $\XX^{(\ell)} \gets $ \scaledregr{$\xset, \Bar{\linfbound}, \{\psi_{ij}\}_{i \in [m], j \in [k]}, \epsilon / 3, \delta / 2, p$} \label{line:call_scaled_regr}
    }
    \Return $\argmin_{\{\XX^{(\ell)} \mid \ell \in \lceil \log_2 \frac{6kR}{\delta} \rceil\}} \psi(\XX^{(\ell)}) + \|\XX^{(\ell)}\|_{1, \infty}$ \label{line:return_best_XX}
}
\Fn{\scaledregr{$\xset, \linfbound, \psibound, \{\psi_{ij}\}_{i \in [m], j \in [k]}, \epsilon, \delta, p$}}{
    $q = 1 + \frac{1}{p}$, $\sset \defeq \{\XX \in \xset \cap [0, \linfbound]^{m \times k}: \|\XX\|_{q, p} \le \linfbound \cdot m^{\frac{1}{pq}}\}$ \label{line:choose_constraintset} \;
    Pick $\AA_{1, \infty} \in \R_{\ge 0}^{m \times mk}$ satisfying $\|\AA_{1, \infty} \vvec(\XX)\|_{\infty} = \|\XX\|_{1, \infty}, \forall \XX \in \xset$  \label{line:pick_A_and_reg} \;
    $\linfreg^{\linfbound, p}: \xset \times \Delta^m \to \R$ is the doubly entropic regularizer $r^{\alpha, \xi}$ in \Cref{def:regularizer} with $\xi = \min(\frac{1}{m}, \frac{\linfbound}{m k \|\AA_{1, \infty}\|_{\infty}})$, $\alpha = 16 \linfbound m^{\frac{2}{pq}} \log (\max(\frac{1}{\xi}, \boxbound + \xi))$ \label{line:pick_just_reg} \; 
    \tcp{We may choose $\AA_{1, \infty}$ and $\linfreg^{\linfbound, p}$ as in \Cref{def:specialized_reg}}
    $\oracle_{\bestrsp}^{\sset, \psi, \linfreg^{\linfbound, p}} \gets $ $\delBR$-ABRO, $\delBR = \frac{\delta}{10}$ \label{line:build_BR_oracle} \;
    \tcp*[h]{this $\delBR$-ABRO is given by \Cref{thm:build_BR_oracle} }\; 
    $\algoutput{\XX} \gets$ \boxsimplex{$\psi, \sset, \AA_{1, \infty}, \linfbound \cdot m^{\frac{2}{pq}}, \boxbound, \oracle_{\bestrsp}^{\sset, \psi, \linfreg^{\linfbound, p}}, \epsilon \cdot m^{- \frac{2}{pq}}$} \label{line:call_boxsimplex} \;
    \Return $\algoutput{\XX}$
}\end{algorithm}

As explained above, 
the underlying optimization method in $\mathsf{Scaled\ell_{1, \infty}Regression}$ is the constrained box-simplex solver given by \Cref{thm:box_simplex_solver}. This constrained box-simplex solver solves several subproblems involving minimizations of expressions that involve a term related to the regularizer defined as in \Cref{line:pick_just_reg}. 
To efficiently implement the $\mathsf{Scaled\ell_{1, \infty}Regression}$ subroutine, we use the following regularizer (and associated matrix $\AA_{1, \infty}$):

\begin{definition}\label{def:specialized_reg} 
    Consider the setting of \Cref{prob:comp_ell1inf_regr}. 
    Let $\AA_{1, \infty} \in \{0, 1\}^{m \times mk}$ be the matrix satisfying that $(\AA_{1, \infty})_{i,\ell} = 1$ if and only if $\ell \mod m = i$. 
    Given $p \in \Z_{\ge 2}$ and 
    $\linfbound > 0$, let $q = 1 + 1/p$, $\xi = \min(\frac{1}{m}, \frac{\linfbound}{m k \|\AA_{1, \infty}\|_{\infty}})$, $\alpha = 16 \linfbound m^{\frac{2}{pq}} \log (\max(\frac{1}{\xi}, \boxbound + \xi))$, 
    and define 
    $\linfreg^{\linfbound, p}: \xset \times \Delta^m \to \R$ as  
    \[\linfreg^{\linfbound, p}(\XX, \yy) = \sum_{i \in [m], j \in [km]} (\yy_i + \xi) |(\AA_{1, \infty})_{ij}|  (\vvec(\XX)_{ij} + \xi) \log (\vvec(\XX)_{ij} + \xi) + \alpha \sum_{i \in [m]} \yy_i \log \yy_i \,.\]
\end{definition}

It is easy to see that the matrix $\AA_{1, \infty}$ satisfying the conditions in \Cref{def:specialized_reg} satisfies that $\|\AA_{1, \infty} \vvec(\XX)\|_{\infty} = \|\XX\|_{1, \infty}, \forall \XX \in \xset$ and is thus a valid choice for \Cref{line:pick_A_and_reg}. 
Additionally, note that we can alternatively write $\linfreg^{\linfbound, p}$ as \[\linfreg^{\linfbound, p}(\XX, \yy) = \sum_{i \in [m], j \in [k]} (\yy_i + \xi) (\XX_{ji} + \xi) \log (\XX_{ji} + \xi) + \alpha \sum_{i \in [m]} \yy_i \log \yy_i \,.\]
The $\linfreg^{\linfbound, p}$ regularizer is the same as the regularizer defined in \Cref{line:choose_r_alpha} of \Cref{alg:box_simplex} corresponding to the parameters given to the call in \Cref{line:call_boxsimplex} of \Cref{alg:composite_ell1inf_regression}. Specifically, in the context of \Cref{alg:box_simplex}, and its corresponding performance guarantee \Cref{thm:box_simplex_solver}, we work with $\uset = \sset = \{\XX \in \xset \cap [0, 1]^{m \times k}: \|\XX\|_{q, p} \le \linfbound \cdot m^{\frac{1}{pq}}\}$, where $q = 1 + \frac{1}{p}$, $\AA = \AA_{1, \infty}$, the upper bound $\linfbound \cdot m^{\frac{2}{pq}}$ on $\max_{\XX \in \sset} \||\AA_{1, \infty}| \vvec(\XX)\|_{\infty} = \max_{\XX \in \sset} \|\XX\|_{1, \infty}$ and the upper bound $\boxbound_{\xset} = \boxbound$ on $\mathrm{diam}(\xset)$. 

This subroutine $\mathsf{Scaled\ell_{1, \infty}Regression}$, if given a parameter $\linfbound$ with the property that there exists some $\anyX \in \xset$ with $\frac{3}{2} \|\anyX\|_{1, \infty} \le \linfbound$, will return a solution $\algoutput{\XX} \in \sset$ with an additive error $\epsilon \linfbound + \delta$ for the problem $\min_{\XX \in \sset} \psi(\XX) + \|\XX\|_{1, \infty}$. When $p = 2 \lceil \sqrt{\log m} \rceil + 1$, we can implement the subroutine efficiently, due to our result on composite $\ell_{q.p}$-regression (\Cref{thm:comp_ellqp_regr}). 
This is summarized in the theorem below. 

\begin{theorem}
\label{thm:generalized_approxL1InfMin}
Consider the setting of \Cref{prob:comp_ell1inf_regr} where $\xset_j \subseteq [0, R]^m$ for every $j \in [k]$, $|\psi(\XX)| \in [0, \psibound], \forall \XX \in \xset$. 
Let $\anyX \in \xset$ and $\linfbound > 0$ with $\linfbound \ge \frac{3}{2} \|\anyX\|_{1, \infty}$. 
Given $\linfbound$, $\psibound$, $\epsilon, \delta > 1/\poly(m,k)$ and $p \in \Z_{\ge 2}$, 
the subroutine $\mathsf{Scaled\ell_{1,\infty}Regression}$   
outputs $\algoutput{\XX} \in \xset$ such that  
\begin{align*}
\psi(\algoutput{\XX}) + \norm{\algoutput{\XX}}_{1, \infty} \le \psi(\anyX) + \norm{\anyX}_{1, \infty} + \eps \linfbound + \delta.  
\end{align*}
When $p = 2 \lceil \sqrt{\log m} \rceil + 1$, given access to a $(\Time_j, \delSCO)$-SCO $\oracle_j$ on $\psi_j$ over $\xset_j$ for every $j \in [k]$, with $\delSCO > \frac{1}{\poly(m, k)}$ small enough, the subroutine can be implemented in $\almostTime(\sum_{j \in [k]} (\Time_j + m) \epsilon^{-1})$ time.
\end{theorem}

We now show how to prove \Cref{thm:approxL1InfMin} given the above correctness and efficiency guarantee on the subroutine $\mathsf{Scaled\ell_{1, \infty}Regression}$. 

\begin{proof}[Proof of \Cref{thm:approxL1InfMin}]
    First, we show 
    the call $\mathsf{Composite\ell_{1, \infty}Regression}(\xset, R, \psibound, \{\psi_{ij}\}_{i \in [m], j \in [k]}, \epsilon, \allowbreak \delta,p)$  yields the desired output. 
    In particular, note that it suffices to show that for at least one of the values $\Bar{\linfbound}$ in \Cref{line:choose_scale}, the output $\algoutput{\XX}$ of $\mathsf{Scaled\ell_{1, \infty}Regression}(\xset, \Bar{\linfbound}, \psibound, \allowbreak \{\psi_{ij}\}_{i \in [m], j \in [k]}, \epsilon / 3, \delta/2,p)$ satisfies 
    \[\psi(\algoutput{\XX}) + \norm{\algoutput{\XX}}_{1, \infty} \le \psi(\XX^{\star}) + (1 + \epsilon) \norm{\XX^{\star}}_{1, \infty} + \delta\,.\]
    
    We consider two cases. The first case is $\|\XX^{\star}\|_{1, \infty} \le \delta$. For this case, by \Cref{thm:generalized_approxL1InfMin}, the call 
    $\mathsf{Scaled\ell_{1, \infty}Regression}(\xset, 2^{- \lceil \log_2 \frac{6kR}{\delta} \rceil + 1} k \boxbound, \psibound, \{\psi_{ij}\}_{i \in [m], j \in [k]}, \epsilon / 3, \delta/2, p)$ outputs $\algoutput{\XX}$ with 
    \[\psi(\algoutput{\XX}) + \norm{\algoutput{\XX}}_{1, \infty} \le \psi(\XX^{\star}) + \norm{\XX^{\star}}_{1, \infty} + \epsilon \cdot 2^{- \lceil \log_2 \frac{6kR}{\delta} \rceil + 1} k \boxbound / 3 + \delta / 2.\]
    Since $2^{- \lceil \log_2 \frac{6kR}{\delta} \rceil + 1} k \boxbound \le \delta$, we have $\psi(\algoutput{\XX}) + \norm{\algoutput{\XX}}_{1, \infty} \le \psi(\XX^{\star}) + \norm{\XX^{\star}}_{1, \infty} + \delta$, which implies that $\Bar{\linfbound} = 2^{- \lceil \log_2 \frac{6kR}{\delta} \rceil + 1} k \boxbound$ yields a suitable $\algoutput{\XX}$. 
    
    The second case is $\|\XX^{\star}\|_{1, \infty} \ge \delta$. 
    Let $\linfbound^{\star}$ be a value picked in \Cref{line:choose_scale} with the property that $\frac{3}{2} \|\XX^{\star}\|_{1, \infty} \le \linfbound^{\star} < 3 \|\XX^{\star}\|_{1, \infty}$. It is clear that such a value exists since the values of $\Bar{\linfbound}$ in \Cref{line:choose_scale} are picked to be the form $2^{-\ell + 1} kR$ for $\ell \in \lceil \log_2 \frac{6kR}{\delta} \rceil$, thus consecutive values are apart by a factor of at most $2$, and the largest $\|\XX^{\star}\|_{1, \infty}$ can be is $kR$. 
    By \Cref{thm:generalized_approxL1InfMin}, the call $\mathsf{Scaled\ell_{1, \infty}Regression}(\xset, \linfbound^{\star}, \psibound, \{\psi_{ij}\}_{i \in [m], j \in [k]}, \epsilon / 3, \delta / 2, p)$ outputs $\algoutput{\XX}$ so that $\psi(\algoutput{\XX}) + \|\algoutput{\XX}\|_{1, \infty} \le \psi(\XX^{\star}) + \|\XX^{\star}\|_{1, \infty} + \epsilon \linfbound^{\star} / 3$. Since $\linfbound^{\star} / 3 < \|\XX^{\star}\|_{1, \infty}$, we have the desired bound of
    \[\psi(\algoutput{\XX}) + \norm{\algoutput{\XX}}_{1, \infty} \le \psi(\XX^{\star}) + (1 + \epsilon) \norm{\XX^{\star}}_{1, \infty} + \delta\,.\]

    Now, we prove the runtime when $p = 2 \lceil \sqrt{\log m} \rceil + 1$ and we are given access to the SCOs. Note that by \Cref{thm:generalized_approxL1InfMin}, every call of $\mathsf{Scaled\ell_{1, \infty}Regression}$ in \Cref{line:call_scaled_regr} calls each of the $(\Time_j, \delSCO)$-SCOs $\otilde(\eps^{-1})$ times. Since $\frac{\boxbound}{\delta} \le \poly(m,k)$, as $\delta > 1/\poly(m,k)$, each of the $(\Time_j, \delSCO)$-SCOs is called $\almostTime(\eps^{-1})$ times during the execution of $\mathsf{Composite\ell_{1, \infty}Regression}$, which proves the runtime bound. 
\end{proof}

The rest of the subsection is dedicated to proving \Cref{thm:generalized_approxL1InfMin}. 
First, we start with a lemma which informally says that the $\|\cdot\|_{1, \infty}$ norm is well-approximated by $\|\cdot\|_{q, p}$ norm, in the sense that the approximation factor is $m^{o(1)}$ when $p = O(\sqrt{\log m}), q = 1 + \frac{1}{p}$. 

\begin{lemma}\label{obs:sset}
\label{obs:qp1inf}
Consider the setting of \Cref{prob:comp_ell1inf_regr}. For any $\XX \in \xset$, and $1 < q \le p$, we have 
\[m^{-\frac{1}{pq}} \cdot \norm{\XX}_{q,p} \le \norm{\XX}_{1, \infty} \le m^{1 - \frac{1}{q}} \cdot \norm{\XX}_{q,p}\,.\]
\end{lemma}
\begin{proof}
First, for any $\XX \in \xset$, we have
\begin{align*}
\norm{\XX}_{1, \infty} \le \norm{\XX}_{1, pq} \le m^{1-\frac{1}{q}}\norm{\XX}_{q, p} = m^{1-\frac{1}{q}} \norm{\XX}_{q, p}\,,
\end{align*}
where for the first inequality we used $\max_{i \in [m]} \sum_{j \in [k]} \XX_{ij} \le (\sum_{i \in [m]} (\sum_{j \in [k]} \XX_{ij})^{pq})^{\frac{1}{pq}}$, and for the second inequality we used $(\sum_{j \in [k]} \XX_{ij})^q \le k^{q-1} \sum_{j \in [k]} \XX_{ij}^q$, which is given by Holder's inequality. 
Second, we have 
\[\norm{\XX}_{q, p} \le \left(\sum_{i \in [m]} (\sum_{j \in [k]} \XX_{ij})^{pq} \right)^{\frac{1}{pq}} \le \left(m \norm{\XX}_{1, \infty}^{pq}\right)^{\frac{1}{pq}} = m^{\frac{1}{pq}} \cdot \norm{\XX}_{1,\infty}\,,\]
where for the first inequality we used the positivity of $\XX_{ij}$ and that $q > 1$ to obtain $\sum_{j \in [k]} |\XX_{ij}|^{q} \le (\sum_{j \in [k]} \XX_{ij})^q$, and for the second inequality we used that, by definition $\sum_{j \in [k]} \XX_{ij} \le \norm{\XX}_{1, \infty}$. 
\end{proof}

\Cref{obs:sset} is useful for bounding the size of $\sset$ favorably in terms of $\linfbound$, where $\sset$ is picked as a function of $\linfbound$ as in \Cref{line:choose_constraintset}. Specifically, it 
helps us bound the quantity $\max_{\XX \in \sset} \|\AA_{1, \infty} \vvec(\XX)\|_{\infty} = \max_{\XX \in \sset} \|\XX\|_{1, \infty}$ (here $\AA_{1, \infty}$ is as defined in \Cref{def:specialized_reg}) by $\linfbound \cdot m^{\frac{2}{pq}}$ (see \Cref{line:call_boxsimplex}). This favorable bound on $\max_{\XX \in \sset} \|\XX\|_{1, \infty}$ helps us obtain additive error $\epsilon \linfbound + \delta$ in $\almostTime(\eps^{-1})$ calls to our ABRO defined in \Cref{line:build_BR_oracle} of $\mathsf{Scaled\ell_{1, \infty}Regression}$, when $p = 2 \lceil \sqrt{\log m} \rceil + 1$, by invoking \Cref{thm:box_simplex_solver}.
Next, we show how to implement the ABRO (formally defined in \Cref{defn:apx_best_resp}) in \Cref{line:build_BR_oracle}. This is captured by the theorem below. 

\begin{theorem}\label{thm:build_BR_oracle}
    Consider the setting of \Cref{prob:comp_ell1inf_regr} where $\xset_j \subseteq [0, R]^m$ for every $j \in [k]$, $|\psi(\XX)| \in [0, \psibound], \forall \XX \in \xset$. 
    Let $ 0 < \linfbound \le \poly(m,k)$ with $\frac{3}{2} \|\anyX\|_{1, \infty} \le \linfbound$ for some $\anyX \in \xset$ and $\sset = \{\XX \in \xset: \|\XX\|_{q, p} \le \linfbound \cdot m^{\frac{1}{pq}}\}$, where $p = 2 \lceil \sqrt{\log m} \rceil + 1, q = 1 + \frac{1}{p}$. 
    Then, given $\psibound \le \poly(m,k)$, $\delBR > 1/\poly(m,k)$ and a $(\Time_j, \frac{\delBR}{\poly(m, k)})$-SCO $\oracle_j$ on $\psi_j$ over $\xset_j$ for every $j \in [k]$,
    we can implement a call to a $\delBR$-ABRO w.r.t. $(\sset, \psi, \linfreg^{\linfbound, p})$ on $(\VV, \yy)$ with $\|\VV\|_{\infty} \le \poly(m, k)$, in time $\almostTime(\sum_{j \in [k]} (\Time_j + m))$.
\end{theorem}

Before proving \Cref{thm:build_BR_oracle}, we present the proof of \Cref{thm:generalized_approxL1InfMin}, as we have gathered all the ingredients for it. 

\begin{proof}[Proof of \Cref{thm:generalized_approxL1InfMin}]
We first prove the correctness. 
In the context of \Cref{thm:box_simplex_solver}, we set $\uset = \{\vvec(\XX), \XX \in \sset\}, \vset = \Delta^m$ and the matrix $\AA = \AA_{1, \infty} \in \R^{m \times mk}$, where $\AA_{1, \infty}$ is defined as in \Cref{line:pick_A_and_reg}.  
Note that $\boxbound \ge \mathrm{diam}(\uset)$ and that by \Cref{obs:sset} $\linfbound \cdot m^{\frac{2}{pq}} \ge \max_{\XX \in \sset} \|\XX\|_{1, \infty} = \max_{\XX \in \sset} \|\AA_{1, \infty} \vvec(\XX)\|_{\infty}$, where $\sset$ is the set defined as in \Cref{line:choose_constraintset}. 
By \Cref{thm:box_simplex_solver}, the call in \Cref{line:call_boxsimplex} outputs $\algoutput{\XX}, \algoutput{\yy}$ with 
\[\langle \yy, \AA_{1, \infty} \vvec(\algoutput{\XX}) \rangle + \psi(\algoutput{\XX}) \le \langle \algoutput{\yy}, \AA_{1, \infty} \vvec(\XX) \rangle + \psi(\XX) + \epsilon \linfbound + 6 \delBR, \enspace \forall \XX \in \sset, \yy \in \vset \,.\]
Since $\max_{\yy \in \Delta^m} \langle \yy, \AA_{1, \infty} \vvec(\algoutput{\XX}) \rangle = \|\algoutput{\XX}\|_{1, \infty}$ and $\langle \algoutput{\yy}, \AA_{1, \infty} \vvec(\XX) \rangle \le \|\XX\|_{1, \infty}, \enspace \forall \XX \in \sset$, we obtain, since $\delBR = \frac{\delta}{10}$, that 
\[\|\algoutput{\XX}\|_{1, \infty} + \psi(\algoutput{\XX}) \le \min_{\XX \in \sset} \|\XX\|_{1, \infty} + \psi(\XX) + \epsilon \linfbound + \delta\,,\]
as needed. This proves the correctness of the call to $\mathsf{Scaled\ell_{1,\infty}Regression}$. 

We now show how to implement the call in $\almostTime(\sum_{j \in [k]} (\Time_j + m) \epsilon^{-1})$ time when $p = 2 \lceil \sqrt{\log m} \rceil + 1$, and we are given access to a $(\Time_j, \delSCO)$-SCO $\oracle_j$ on $\psi_j$ over $\xset_j$ for every $j \in [k]$, with $\delSCO > \frac{1}{\poly(m, k)}$ small enough. 
We do so by choosing $\AA_{1, \infty} \in \R^{m \times mk}$ in \Cref{line:pick_A_and_reg} as in \Cref{def:specialized_reg}, meaning $(\AA_{1, \infty})_{i,\ell} = 1$ if and only if $\ell \mod m = i$ (which is consistent with the condition in \Cref{line:pick_A_and_reg}).  
For this, we first show that we can implement $\delBR$-ABRO $\oracle_{\bestrsp}^{\uset, \psi, r}$ in $\almostTime(\sum_{j \in [k]} (\Time_j +m))$. Note that from the definition of $\xi$ in \Cref{def:specialized_reg}, we have that $\frac{1}{\xi} \ge 1$ and $\xi \le \boxbound$. Thus, $B := \log(\max(\frac{1}{\xi}, \boxbound + \xi)) = O(\log mk)$, as $\boxbound \le \poly(m,k)$. 
By \Cref{thm:box_simplex_solver}, to implement \Cref{alg:box_simplex}, it suffices to have access to a $\delBR$-ABRO $\oracle_{\bestrsp}^{\uset, \psi, r}$ that is called on inputs $(\VV, \yy)$ with $\|\VV\|_{\infty} \le m^2 k B \|\AA_{1, \infty}\|_{\infty}$. Since $\|\AA_{1, \infty}\|_{\infty} \le k$ and $m^2 k^2 \log (mk) \le \poly(m, k)$, \Cref{thm:build_BR_oracle} gives us that we can indeed implement the calls to such an oracle $\oracle_{\bestrsp}^{\uset, \psi, r}$ (defined in \Cref{line:build_BR_oracle}) in time $\almostTime(\sum_{j \in [k]} (\Time_j +m)$. 
Next, by \Cref{thm:box_simplex_solver}, \Cref{line:call_boxsimplex} takes $\almostTime(m k \epsilon^{-1})$ time and $\almostTime(\epsilon^{-1})$ calls to $\oracle_{\bestrsp}^{\uset, \psi, r}$. 
By \Cref{thm:build_BR_oracle}, 
each call to $\oracle_{\bestrsp}^{\uset, \psi, r}$ takes $\almostTime(\sum_{j \in [k]} (\Time_j +m) \epsilon^{-1})$ time. Hence, 
the call to our subroutine $\mathsf{Scaled\ell_{1,\infty}Regression}(\xset, \linfbound, \{\psi_{ij}\}_{i \in [m], j \in [k]}, \epsilon, \delta, p)$ takes $\almostTime(\sum_{j \in [k]} (\Time_j +m) \epsilon^{-1})$ time and outputs $\algoutput{\XX}$ with the desired guarantees when $p = 2 \lceil \sqrt{\log m} \rceil + 1$. 
\end{proof}

The proof of \Cref{thm:build_BR_oracle} relies on showing that we can reduce approximately solving the constrained problem 
\begin{align}\label{eq:approxL1InfMinSubP}
        \min_{\XX \in \sset} \langle \VV, \XX \rangle + \linfreg^{\linfbound, p}(\XX, \Bar{\yy}) + \psi(\XX) .
    \end{align}
to approximately solving an ``unconstrained'' version, namely 
\begin{align}\label{prob:unconstrained_step}
    \min_{\XX \in \xset} \langle \VV, \XX \rangle + \linfreg^{\linfbound, p}(\XX, \Bar{\yy}) + \psi(\XX) + \scaledparam \|\XX\|_{q, p}^{pq},
\end{align}
for a proper choice of parameter $\scaledparam > 0$. To show this reduction, and that searching for a proper value of parameter $\scaledparam$ can be done efficiently, we use the following result. 

\begin{lemma}\label{lem:unconstrained_red}
    Let $\xset \subseteq \R^n$ and $R, M, \delta > 0$. 
    Let $\Gamma, \zeta:\xset \to \R$ be convex functions so that $\zeta(\XX) \in [0, R], \enspace \forall \XX \in \xset$. Assume that $\Gamma$ is $\alpha$-strongly convex with respect to a norm $\|\cdot\|$ and $\zeta$ is $L$-Lipschitz with respect to $\|\cdot\|$. 
    Let $\sset = \{\XX \in \xset \mid \zeta(\XX) \le 1\}$. Suppose we have access to a $\delta$-approximate parametric oracle $\oracle:[0, M] \to \xset$, which takes as input parameter $C \in [0, M]$ and outputs $\XX(C)$ so that $\Gamma(\XX(C)) + C \zeta(\XX(C)) \le \min_{\XX \in \xset} \Gamma(\XX) + C \zeta(\XX) + \delta$. Assuming that 
    $\XX(M) \in \sset$, there exists an algorithm that, in $O(\log \frac{MR}{\delta})$ queries to $\oracle$, outputs $\algoutput{\XX} \in \sset$ so that \[\Gamma(\algoutput{\XX}) \le \min_{\XX \in \sset} \Gamma(\XX) + M L \sqrt{\frac{6 \delta}{\alpha}} + \delta \,.\] 
\end{lemma}

In the context of applying \Cref{lem:unconstrained_red} to prove \Cref{thm:build_BR_oracle}, the parametric oracle used is implemented by the algorithm for solving \Cref{prob:comp_ellqp_regr} to high accuracy. 
In particular, since all quantities regarding range, Lipschitzness and strong convexity (with respect to norm $\ell_{q, p}$) are quasi-polynomially-bounded, \Cref{lem:unconstrained_red} implies that reducing \eqref{eq:approxL1InfMinSubP} to \eqref{prob:unconstrained_step} will take $\Tilde{O}(1)$ calls to our high-accuracy composite $\ell_{q, p}$-regression algorithm. We defer the proof of and pseudocode corresponding to \Cref{lem:unconstrained_red} to the appendix. 

In addition to the lemma above, to prove \Cref{thm:build_BR_oracle}, we also need the following result, which states that solving \eqref{prob:unconstrained_step} up to small enough additive error suffices for satisfying the properties needed to obtain a $\delBR$-ABRO (see \Cref{defn:apx_best_resp}) with respect to $(\sset, \psi, \linfreg^{\linfbound, p})$, where $\sset$ is defined as in \Cref{line:choose_constraintset}, and $\linfreg^{\linfbound, p}$ is defined in \Cref{def:specialized_reg}. 

\begin{lemma}
\label{lem:approxL1InfMinSubPBR}
Consider the setting of \Cref{def:specialized_reg} and 
let $\delta > 0$. Let $\algo$ be an algorithm that, for some $\VV \in \R^{m \times k}$ with $\|\VV\|_{\infty} \le O(\poly(m, k))$ and $\Bar{\yy} \in \Delta^m$, outputs $\algoutput{\XX}$ so that 
\[\langle \VV, \algoutput{\XX} \rangle + \linfreg^{\linfbound, p}(\algoutput{\XX}, \Bar{\yy}) + \psi(\algoutput{\XX}) \le \min_{\XX \in \sset} \langle \VV, \XX \rangle + \linfreg^{\linfbound, p}(\XX, \Bar{\yy}) + \psi(\XX) + \delta \,.\]
Then, $\algo$ implements the call to a
$\O(\sqrt{\delta} m^2 k^2 \xi^{-2} \boxbound)$-ABRO on $(\VV, \Bar{\yy})$ with respect to $(\sset, \psi, \linfreg^{\linfbound, p})$. 
\end{lemma}

By \Cref{lem:approxL1InfMinSubPBR}, any algorithm that solves \eqref{eq:approxL1InfMinSubP} for any $\|\VV\|_{\infty} \le \poly(m, k)$ and $\Bar{\yy} \in \Delta^m$ to additive error $\delta$ implements a $\O(\sqrt{\delta} m^2 k^2 \xi^{-2} \boxbound)$-ABRO for $\|\VV\|_{\infty} \le \poly(m, k)$ and $\Bar{\yy} \in \Delta^m$.
The algorithm in \Cref{thm:comp_ellqp_regr} solves \eqref{eq:approxL1InfMinSubP} to high accuracy with $O_{q,p}(1)$ queries to each one of the $(\Time_j, \delSCO)$-SCOs $\oracle_1, \ldots, \oracle_k$. This is due to the fact that the expression $\langle \VV, \XX \rangle + \linfreg^{\linfbound, p}(\XX, \Bar{\yy})$ is separable over the variables $\XX_{ij}$, and each function of $\XX_{ij}$ is a $m$-decomposable function. 
We are now ready to prove \Cref{thm:build_BR_oracle}. 

\begin{proof}[Proof of \Cref{thm:build_BR_oracle}]
    Let $1 > \altdelta > 1/\poly(m)$ be a parameter which we fix later. 
    Define $\Gamma, \zeta:\xset \to \R$ by $\Gamma(\XX) := \langle \VV, \XX \rangle + \linfreg^{\linfbound, p}(\XX, \Bar{\yy}) + \psi(\XX)$ and $\zeta(\XX) := \frac{\|\XX\|_{q, p}^{pq}}{m \linfbound^{pq}}$. 
    Define $\algoutput{\Gamma}(\XX) = \Gamma(\XX) + \frac{\altdelta}{2 R^{pq}} \|\XX\|_{q, p}^{pq}$ and $\gradbound = \|\VV\|_{\infty}$. We build a $\altdelta$-approximate parametric oracle $\oracle$ (as defined in \Cref{lem:unconstrained_red}) for objective function $\algoutput{\Gamma}(\cdot)$ and constraint function $\zeta$. 

    For this, note that setting $\beta_1 = m (\psibound + 1 + \linfbound + km \gradbound \boxbound) \le \poly(m)$ yields that $\algoutput{\Gamma}(\XX) \le \frac{\beta_1}{100 m^{\frac{1}{pq}}}, \enspace \forall \XX \in \xset$, as $\max_{\XX \in \xset} |\algoutput{\Gamma}(\XX)| \le \max_{\XX \in \xset} |\Gamma(\XX)| + \frac{\altdelta}{2}$ and $|\Gamma(\XX)|$ is upper bounded by $\max_{\XX \in \xset} |\VV, \XX \rangle| + \max_{\XX \in \xset} |\linfreg^{\linfbound, p}(\XX, \Bar{\yy})| + \max_{\XX \in \xset} |\psi(\XX)| \le km \gradbound \boxbound + \almostTime(\linfbound) + \psibound$, where for the last inequality we used \Cref{lem:regularizer}. 
    Additionally, note that the expression $\linfreg^{\linfbound, p}(\XX, \Bar{\yy})$ as a function of $\XX$, for a fixed $\Bar{\yy}$, is separable over $\XX_{ij}$'s, since it can be written as $\sum_{i \in [m], j \in [k]} (\Bar{\yy}_i + \xi) (\XX_{ij} + \xi) \log (\XX_{ij} + \xi)$. Hence, we define functions $c_{ij}:\R \to \R$ for every $i \in [m], j \in [k]$ by $c_{ij}(x) := (\Bar{\yy}_i + \xi) (\XX_{ij} + \xi) \log (\XX_{ij} + \xi)$. By \Cref{lem:xlogx_comp}, each $c_{ij}$ is a computable single-variable convex function, as $\xi \ge \frac{1}{\poly(m)}$. 
    Additionally, since $\|\VV\|_{\infty} \le \poly(m, k)$, each function $\VV_{ij} \XX_{ij}$ is a computable function. Hence, we are in the conditions of \Cref{thm:comp_ellqp_regr}. 
    Let $\oracle$ be the algorithm given by \Cref{thm:comp_ellqp_regr} with respect to convex functions $\psi'_{j}(\XX_{:j}) := \psi_{j}(\XX_{:j}) + \sum_{i \in [m]} c_{ij}(\XX_{ij})$ and factor $C$ in front of the $\|\XX\|_{q, p}^{pq}$ term, assuming access to $(\Time_j, \delSCO)$ oracles $\{\oracle_j\}_{j \in [k]}$. 
    Set $M = 10 \beta_1$ and, for any $C \in [0, M]$, by \Cref{thm:comp_ellqp_regr}, $\oracle$ computes, in $\almostTime(\sum_{j \in [k]} (\Time_j + m))$ time, $\XX(C)$ with \[\algoutput{\Gamma}(\XX(C)) + C \zeta(\XX(C)) \le \min_{\XX \in \xset} \algoutput{\Gamma}(\XX) + C \zeta(\XX) + 3 k \delSCO\,.\] 
    Hence $\oracle$ is a $3 k \delSCO$-approximate parametric oracle as defined in \Cref{lem:unconstrained_red}, to which each call takes $\almostTime(\sum_{j \in [k]} (\Time_j + m))$ time. 

    Now that we have built a $3 k \delSCO$-approximate parametric oracle $\oracle$ for objective function $\algoutput{\Gamma}(\cdot)$ and constraint function $\zeta$ for parameter space $[0, M]$, we build an algorithm $\algo$ for solving $\min_{\XX \in \sset} \algoutput{\Gamma}(\XX)$
    up to additive error $\altdelta$. 
    First, note that $\algoutput{\Gamma}(\cdot)$ is $\beta_2$-strongly convex with respect to $\ell_{q, p}$ for $\beta_2 = \frac{\altdelta}{2 R^{pq}}$. Moreover, $\zeta(\cdot)$ is $\beta_3$-Lipschitz with respect to $\ell_{q, p}$ for $\beta_3 = \boxbound^{pq}$. 
    Additionally, the output of $\oracle(M)$ will be in $\sset$. 
    To see why this is the case, note that since there exists $\anyX$ with $\|\anyX\|_{1, \infty} \le \frac{2}{3} \linfbound$, by \Cref{obs:qp1inf} we have that $\|\anyX\|_{q, p} \le m^{\frac{1}{pq}} \frac{2}{3} \linfbound$, which implies $\zeta(\anyX) \le \frac{2}{3}$. Hence, for any $\XX$ with $\XX \notin \sset$, we have $\zeta(\XX) > 1$, and thus $\algoutput{\Gamma}(\XX) + M \zeta(\XX) > \algoutput{\Gamma}(\anyX) + M \zeta(\anyX) + 3 k \delSCO$, since $|\algoutput{\Gamma}(\XX) - \algoutput{\Gamma}(\anyX)| \le M / 5$ and $3 k \delSCO < M / 5$ (based on how picking $\delSCO$ to be small enough). 
    Thus, 
    we are in the conditions of \Cref{lem:unconstrained_red}, so
    we obtain an algorithm $\algo$ that solves \Cref{eq:approxL1InfMinSubP} up to additive error  $\altdelta' = \beta_1 \beta_3 \sqrt{\frac{3 k \delSCO}{\beta_2}}$. Further, each call to $\algo$ makes $O(\log \beta_1 \frac{\boxbound}{3 k \delSCO}) = \almostTime(1)$ queries to $\oracle$, where we used that $\beta_1, \beta_3, \boxbound, \frac{1}{\beta_2}, \frac{1}{\delSCO} \le \poly(m, k)$. 

    Hence, $\algo$ outputs $\algoutput{\XX}$ with the properties that $\zeta(\algoutput{\XX}) \le 1$, which is equivalent to $\algoutput{\XX} \in \sset$, and 
    \[\algoutput{\Gamma}(\algoutput{\XX}) \le \min_{\XX \in \sset} \algoutput{\Gamma}(\XX) + \altdelta'.\]
    Note that from the way we defined $\algoutput{\Gamma}$, we have $|\algoutput{\Gamma}(\XX) - \Gamma(\XX)| \le \frac{\altdelta}{2}, \forall \XX \in \xset$. Thus, for any $\VV, \Bar{\yy}$, $\algo$ outputs $\algoutput{\XX} \in \sset$ so that 
    \[\langle \VV, \algoutput{\XX} \rangle + \linfreg^{\linfbound, p}(\algoutput{\XX}, \Bar{\yy}) + \psi(\algoutput{\XX}) \le \min_{\XX \in \sset} \langle \VV, \XX \rangle + \linfreg^{\linfbound, p}(\XX, \Bar{\yy}) + \psi(\XX) + \altdelta' + \altdelta\,.\]
    Thus, by \Cref{lem:approxL1InfMinSubPBR}, $\algo$ implements a $O(\sqrt{(\altdelta' + \altdelta)} m^2 k^2 \xi^{-2} \boxbound)$-ABRO on $(\VV, \Bar{\yy})$ with respect to $(\sset, \psi, r)$. 
    Since $\xi^{-2}, \boxbound \le m^{C_1}$ for some constant $C_1$ and $\delBR > m^{-C_2}$ for some constant $C_2$, it suffices to pick $\altdelta = m^{-50(C_1 + C_2)}$ and to have $\altdelta' = O(m^{-50(C_1 + C_2)})$. This is equivalent to 
    $\beta_1 \beta_3 \sqrt{\frac{3 k \delSCO}{\beta_2}} = O(m^{-50(C_1 + C_2)})$, which is satisfied for $\delSCO = (mk)^{-C_3}$ for a small enough constant $C_3$, as 
    $\beta_1, \beta_3, \frac{1}{\beta_2}, \le \poly(m, k)$. 

    Finally, note that to build our algorithm $\algo$ we call our $\delBR$-ABRO $\oracle$ $\almostTime(1)$ times.
    By \Cref{thm:comp_ellqp_regr}, each such call requires $\almostTime(1)$ calls to each of the oracles $\oracle_1, \ldots \oracle_k$,  because of our choice of $p = 2 \lceil \sqrt{\log m} \rceil + 1$ and $q = 1 + \frac{1}{p}$. 
\end{proof}

To prove \Cref{lem:approxL1InfMinSubPBR}, we use a technical lemma that relates the additive error of solving \eqref{eq:approxL1InfMinSubP} to the ABRO requirements (\Cref{defn:apx_best_resp}).
The proof is deferred to \Cref{sec:BRFromAdditiveError}. 

\begin{lemma}
\label{lem:BRFromAdditiveError}
Let $\uset \subseteq \R^n$ be convex, $\vset = \Delta^m$, $\psi: \uset \to \R$ convex, a convex regularizer $r: \uset \times \vset \to \R$, $\gg \in \R^n$, $\vv \in \vset$, 
$\uu^{\star} \in \argmin_{\uu \in \uset} \l\gg, \uu\r+r(\uu, \vv)+\psi(\uu)$ and $\algoutput{\uu}$ such that
\begin{align*}
    \l\gg, \algoutput{\uu}\r+r(\algoutput{\uu}, \vv)+\psi(\algoutput{\uu}) \le \l\gg, \uu^{\star}\r+r(\uu^{\star}, \vv)+\psi(\uu^{\star}) + \delta\,.
\end{align*}
If for some $\lambda_{\max} \ge \lambda_{\min} > 0$ and $\bbeta \ge 0$ we have 
$\lambda_{\min} \II_n \preceq \g^2_{\uu, \uu} r(\uu, \vv) \preceq \lambda_{\max} \II_n$ and $\norm{\g^2_{\vv, \uu}r(\uu, \vv)}_{\infty} \allowbreak \le \bbeta$ for every $\uu \in \uset$, we have
\begin{enumerate}
    \item $\l\gg + \g_{\uu} r(\algoutput{\uu}, \vv), \algoutput{\uu} - \uu\r + \psi(\algoutput{\uu}) - \psi(\uu) \le \delta + 2 \lambda_{\max} \sqrt{\frac{2 n \delta}{\lambda_{\min}}} \usetSize $ for any $\uu \in \uset$, and 
    \item $\|\g_{\vv}r(\uu^{\star}, \vv) - \g_{\vv}r(\algoutput{\uu}, \vv)\|_{\infty} \le \frac{\bbeta \sqrt{2 \delta}}{\sqrt{\lambda_{\min}}}$
\end{enumerate}
\end{lemma}

We are now ready to prove \Cref{lem:approxL1InfMinSubPBR} using \Cref{lem:BRFromAdditiveError}.

\begin{proof}[Proof of \Cref{lem:approxL1InfMinSubPBR}]
Let $\mathcal{A}$ be our $\delta$-approximate algorithm for solving \eqref{eq:approxL1InfMinSubP}.
Fix $\VV \in \R^{m \times k}$ and some $\yy \in \Delta^m$ and let 
$\hat{\XX}$ be the output of $\mathcal{A}$ on $\VV, \yy$. Since $\hat{\XX}$ is a 
a solution to \eqref{eq:approxL1InfMinSubP} of additive error $\delta$, we have 
\begin{align*}
    \psi(\hat{\XX}) + \l\VV, \hat{\XX}\r + \linfreg^{\linfbound, p}(\hat{\XX}, \yy) \le \delta + \psi(\XX^{\textrm{BR}}) + \l\VV, \XX^{\textrm{BR}}\r + \linfreg^{\linfbound, p}(\XX^{\textrm{BR}}, \yy)
\end{align*}
where $\XX^{\textrm{BR}} \in \sset$ minimizes the right hand side. We show that $\hat{\XX}$ satisfies both criteria in \Cref{defn:apx_best_resp} via \Cref{lem:BRFromAdditiveError}. 
To do so, we 
analyze the spectrum of the Hessian of $\linfreg^{\linfbound, p}(\XX, \yy)$ for $\XX \in \xset.$

Fix $\XX$ be any element in $\xset$ and note that 
\begin{align*}
\g^2_{\XX, \XX} \linfreg^{\linfbound, p}(\XX, \yy) = \ddiag(|\AA_{1, \infty}|^\top (\yy + \xi)) \ddiag(\XX + \xi)^{-1}
\end{align*}
We now show that since $\AA_{1, \infty}$ satisfies the conditions in \Cref{def:specialized_reg}, we have 
that $\lambda_{\min} \II_{mk} \preceq \g^2_{\XX, \XX} \linfreg^{\linfbound, p}(\XX, \yy) \preceq \lambda_{\max} \II_{mk}$ for $\lambda_{\min} = \frac{k \xi}{1 + \xi}$ and $\lambda_{\max} = \frac{k (1 + \xi)}{\xi}$. 

Now, to bound $\norm{\g^2_{\XX, \yy} \linfreg^{\linfbound, p}(\XX, \yy)}_{\infty}$, note that
\begin{align*}
\g^2_{\XX, \yy} \linfreg^{\linfbound, p}(\XX, \yy) = \ddiag(1 + \log(\XX + \xi)) |\AA_{1, \infty}|^\top
\end{align*}
Next, note that $|\log(\XX + \xi)| \le O(\log km), \enspace \forall \XX \in \sset$. This is due to the fact that $\XX \in [0, \boxbound]^{m \times k}, \enspace \forall \XX \in \sset$ and $\frac{\delta}{\|\AA_{1, \infty}\|_{\infty}} \le \frac{1/\poly(m)}{k} \le \xi \le 1$. 
This implies \[\norm{\g^2_{\XX, \yy} \linfreg^{\linfbound, p}(\XX, \yy)}_{\infty} = O(k \log km)\]
because $\||\AA_{1, \infty}|\|_{\infty} \le k$. 

Thus, by \Cref{lem:BRFromAdditiveError}, $\hat{\XX}$ satisfies both criteria in \Cref{defn:apx_best_resp} with the error parameter 
\begin{align*}
    \delta' = \delta +  \max\left\{\frac{2 k (1 + \xi)}{\xi} \sqrt{\frac{2 km \delta (1+\xi)}{k\xi}} \boxbound, O(k \log km)\sqrt{\frac{2\delta (1+\xi)}{k\xi}}\right\} = \delta + \O(\sqrt{\delta} m^2 k^2 \xi^{-2} \boxbound)
\end{align*}
Hence, $\mathcal{A}$ implements a $\delta'$-ABRO call on $(\VV, \Bar{\yy})$ with respect to $(\sset, \psi, \linfreg^{\linfbound, p})$. 
\end{proof}

\subsection{Application: Approximate Directed Multi-Commodity Flows}
\label{sec:kcommflow}

In this section, we present our results regarding approximately solving the \mcf{} problems introduced in \Cref{sec:intro}. We start by proving our result regarding the directed composite \mcf{} problem (\Cref{thm:mainKCommFlow}). 
We then show how to use this result to obtain the improved $\almostTime(mk \eps^{-1})$ runtimes for concurrent \mcf{} (\Cref{coro:approxConcurrentFlow}) and maximum weighted concurrent \mcf{} (\Cref{coro:approxMaxKCommFlow}). Before presenting \Cref{thm:mainKCommFlow}, we recall the formulation of the composite \mcf{} problem and the statement of \Cref{thm:mainKCommFlow}. 

\cvxKCommFlow*

\mainKCommFlowThm*

\begin{proof}[Proof of \Cref{thm:mainKCommFlow}]
We prove the theorem using our $\ell_{1, \infty}$-regression algorithm~\Cref{thm:approxL1InfMin} and the almost linear time convex flow solver~\Cref{coro:convexFlow}. 

First, we reduce the problem to an instance of the composite $\ell_{1, \infty}$-regression problem (\Cref{prob:comp_ell1inf_regr}).
Define $\boxbound = \max(\|\uu\|_{\infty}, \eta \RdUpper)$.
For any commodity $i \in [k]$, we define a convex set $\xset_i \subseteq [0, \boxbound]^{m+1}$ as follows:
\begin{align*}
    \xset_i \defeq \left\{\xx = \begin{pmatrix}
        \uu^{-1} \ff_i \\
        \eta \bbeta_i
    \end{pmatrix} \in [0, \boxbound]^{m+1} \middle\vert \imbal(\ff_i) = \bbeta_i \dd_i\right\}\text{, where }\eta \defeq \frac{\min_i \sum_{u\in V} |\dd_{iu}|}{2 m (\max_{e \in E} \uu_e)} > 0
\end{align*}
We also define $\psi_i: \xset_i \to \R$ by
\begin{align*}
\psi_i\left(\begin{pmatrix}
        \uu^{-1} \ff_i \\
        \eta \bbeta_i
    \end{pmatrix}\right) \defeq v_i(\bbeta_i) + \sum_{e \in E} c_{ei}(\ff_{ie})
\end{align*}
Let $\xset \defeq \bigoplus_{i \in [k]} \xset_i.$
There is a natural one-to-one correspondence between $\xset$ and the set of feasible solutions $(\FF, \bbeta)$ of the composite \mcf{} instance:
\begin{align*}
(\FF, \bbeta) \mapsto \begin{pmatrix}
    \uu^{-1} \ff_1 & \ldots & \uu^{-1} \ff_k \\
    \eta \bbeta_1 & \ldots & \eta \bbeta_k
\end{pmatrix} \in \xset
\end{align*}
For any $\XX \in \xset$, we have $\psi(\XX) = \cc(\FF, \bbeta)$ where $(\FF, \bbeta)$ is the feasible solution corresponding to $\XX$ by definition.

To establish the reduction, we show that $\|\XX\|_{1, \infty} = \congest(\FF).$
Let $(\FF, \bbeta)$ be the feasible solution corresponding to $\XX.$
We have
\begin{align*}
\|\XX\|_{1, \infty} = \max\left\{\max_e \sum_{i \in [k]} \frac{|\ff_{ie}|}{\uu_e}, \eta \sum_{i \in [k]} \bbeta_i\right\}
\end{align*}
The first term, $\max_{e \in E} \sum_{i \in [k]} \frac{|\ff_{ie}|}{\uu_e}$, is exactly $\congest(\FF).$
The second term, $\eta \sum_{i \in [k]} \bbeta_i$, is at most $\congest(\FF)$ because the total amount of the flows, $\sum_{i, e} \ff_{ie}$, is at least half the total amount of net flow of each vertices, $\sum_{i \in [k], u \in V} \bbeta_i |\dd_{iu}|.$
We have
\begin{align*}
\sum_{i \in [k]} \bbeta_i \left(\min_{i \in [k]} \sum_{u \in V} |\dd_{iu}|\right) \le \sum_{i \in [k], u \in V} \bbeta_i |\dd_{iu}| \le 2\sum_{i \in [k], e \in E}\ff_{ie} \le 2m \cdot (\max_e \uu_e) \cdot \congest(\FF)
\end{align*}
The claim and the reduction follow from our definition of $\eta.$
That is, we have, for any $\XX \in \xset$ and its corresponding feasible solution $(\FF, \bbeta)$, that
$\psi(\XX) + \|\XX\|_{1, \infty} = c(\FF, \bbeta) + \congest(\FF).$

Next, we show, 
for any commodity $i \in [k]$, how to implement a $(\almostTime(m), \delSCO)$-SCO $\oracle_i$ on $\psi_i$ over $\xset_i$ (\Cref{def:sepCvxOracle}), with $\delSCO = \frac{\delta}{\poly(m, k)}$, using the almost linear time convex flow solver~\Cref{coro:convexFlow}.
To see this, $\oracle_i$ needs to solve, given a collection $\{g_j\}_{j \in [m+1]}$ functions that are the sum of $O(1)$ of computable convex functions, the following problem up to additive error $\delta$:
\begin{align*}
    \min_{\xx \in \xset_i} \psi_i(\xx) + \sum_{j \in [m+1]} g_j(\xx_j)\,.
\end{align*}
It is equivalent to solving the following problem up to additive error $\delSCO$:
\begin{align*}
    \min_{\ff_i \in [0, \boxbound_u]^E, \bbeta_i \in [\boxbound_{\ell}, \boxbound_{\mu}]: \imbal(\ff_i) = \bbeta_i \dd_i} v_i(\bbeta_i) + g_{m+1}(\eta \bbeta_i) + \sum_{e \in E} c_{ei}(\ff_{ie}) + g_e(\uu_e^{-1} \ff_{ie})
\end{align*}
where we use the one-to-one correspondence between $\xset_i$ and the set of feasible solutions on the $i$-th commodity $(\ff_i, \bbeta_i) \mapsto [\uu^{-1} \ff_i; \eta \bbeta_i]$ (and $g_e$'s are sums of $O(1)$ of computable convex functions and).
Therefore, each invocation to any $\oracle_i, i \in [k]$ takes $\almostTime(m)$ time.

Now, we solve the composite $\ell_{1, \infty}$-regression problem instance via \Cref{thm:approxL1InfMin} for $p = 2 \lceil \sqrt{\log m} \rceil + 1$. 
Having access to $(\almostTime(m), \delSCO)$-SCOs $\oracle_j$, and noting that $\boxbound \le \poly(m)$, due to $\RdUpper, \|\uu\|_{\infty} \le \poly(m)$,
by \Cref{thm:approxL1InfMin}, it takes $\almostTime(k \eps^{-1})$ calls to oracles to find $(\algoutput{\FF}, \algoutput{\bbeta})$ such that

\begin{align*}
c(\algoutput{\FF}, \algoutput{\bbeta}) + \congest(\algoutput{\FF}) \le c(\FF^{\star}, \bbeta^{\star}) + (1+\eps)\congest(\FF^{\star}) + \delta
\end{align*}
where $(\FF^{\star}, \bbeta^{\star})$ minimizes the objective $c(\FF, \bbeta) + \congest(\FF).$
The total runtime is $\almostTime(mk\eps^{-1})$ and this concludes the proof.
\end{proof}

Using \Cref{thm:mainKCommFlow}, we obtain the $\almostTime(mk \eps^{-1})$-time $(1+\eps)$-approximate algorithms for concurrent and maximum weighted directed multi-commodity problems (\Cref{coro:approxConcurrentFlow} and \Cref{coro:approxMaxKCommFlow} respectively). 
In particular, up to an $\O(1)$ overhead from binary search, one can reduce all the classical \mcf{} problems to directed composite MCF. 
We now present the proof of \Cref{coro:approxConcurrentFlow}. 

\begin{corollary}[Approximate Concurrent \mcf{}]
\label{coro:approxConcurrentFlow} 
Given a poly-regular \mcf{} instance $\instance$ and
$\eps > 1 / \poly(m)$, there is an algorithm that, in $\almostTime(mk \eps^{-1})$ time, outputs $(\algoutput{\FF}, \algoutput{\beta}) \in S_{\instance}$ so that 
\[
\algoutput{\beta} \ge (1 - \epsilon) \max_{(\FF, \beta \vecone^k) \in S_{\instance}} \beta \,.
\]
\end{corollary}
\begin{proof}
We first show that
to solve the approximate \mcf{} problem, it suffices to find
a flow $\FF$ that routes $\DD$ with minimum congestion $\congest(\FF)$ (by scaling down by the quantity $\max_{(\FF, \beta \vecone^k) \in S_{\instance}} \beta$). In other words, it suffices to find $(\algoutput{\FF}, \algoutput{\beta} \vecone^k) \in S_{\instance}$ so that 
\begin{equation}\label{eq:conc_mcflow_reform}
    \congest(\algoutput{\FF}) \le (1+\epsilon) \min_{(\FF, \beta \vecone^k) \in S_{\instance}} \congest(\FF)\,.
\end{equation}
To see why this is the case, note that $(\FF', \beta' \vecone^k) = \frac{1}{\congest(\algoutput{\FF})} (\algoutput{\FF}, \algoutput{\beta} \vecone^k)$ satisfies $(\FF', \beta' \vecone^k) \in S_{\instance}$. Moreover, it is easy to see that $\max(\FF, \beta \vecone^k) \in \frac{1}{\congest(\FF)} S_{\instance} = \max_{(\FF, \beta \vecone^k) \in S_{\instance}} \beta$. Hence, $\beta' \ge \frac{1}{1+\epsilon} \max_{(\FF, \beta \vecone^k) \in S_{\instance}} \beta \ge (1-\epsilon) \max_{(\FF, \beta \vecone^k) \in S_{\instance}} \beta$. 

Now, we solve \Cref{eq:conc_mcflow_reform}. First, we may assume $\epsilon < 1/4$.
Denote by $\FF^{\star}$ the optimal flow routing demand $\DD$ and let $C_1$ be a constant so that $\|\DD\|_\infty, \|\uu\|_\infty \le m^{C_1}$, which clearly exists since $\DD$ is polynomially bounded. 
Then, we have that $\congest(\FF^{\star}) > m^{-C_1 - 1}$, since there must be at least one unit of flow on some edge $e$. We also have $\congest(\FF^{\star}) \le \|\DD\|_\infty \le m^{C_1 + 1}$, since on each edge there is at most $m \|\DD\|_{\infty}$ flow and $\uu$ has integer entries. 

Finding minimum congested flow follows directly from \Cref{thm:mainKCommFlow} by considering the case where $c_{ei}(x) = 0, v_i(x) = 0$ for any edge $e$, commodity $i$, where we impose the constraint $\im(\FF_j) = \bbeta_j \dd_j$ by making range bounds $\RdLower, \RdUpper$ equal to $1$, and setting error parameters $\epsilon' = \epsilon/2$ and $\delta = \epsilon m^{-(C_1 + 2)}$. Calling the algorithm from \Cref{thm:mainKCommFlow} yields a flow $\algoutput{\FF}$ with $\congest(\algoutput{\FF}) \le (1 + \epsilon/2) \congest(\FF^{\star}) + \delta$, and since $\delta < \epsilon \frac{\congest(\FF^{\star})}{m}$, we have $\congest(\algoutput{\FF}) \le (1 + \epsilon) \congest(\FF^{\star})$, as needed. 
\end{proof}

Finally, we prove that there is an algorithm that solves the
\emph{maximum weighted \mcf{}} problem up to $(1 + \eps)$-approximation in $\almostTime(mk \eps^{-1})$ time (\Cref{coro:approxMaxKCommFlow}).
\begin{corollary}[Approximate Maximum Weighted \mcf{}]
\label{coro:approxMaxKCommFlow}
Given a poly-regular \mcf{} instance $\instance$, weights $\ww \in \Z^k_{\ge 0}$  with $\|\ww\|_{\infty} \le \poly(m)$ and
$\eps > 1 / \poly(m)$, there is an algorithm that, for objective $\mwmcObj(\FF, \bbeta) := \langle \ww, \bbeta \rangle$, in time $\almostTime(mk \eps^{-1})$, outputs $(\algoutput{\FF}, \algoutput{\bbeta}) \in S_{\instance}$ with $\congest(\algoutput{\FF}) \le 1$ so that 
\[
\mwmcObj(\algoutput{\FF}, \algoutput{\bbeta}) \le \max_{(\FF, \bbeta) \in S_{\instance}, \congest(\FF) \le 1} \mwmcObj(\FF, \bbeta)\,.
\]
\end{corollary}

\begin{proof}
First, we may assume $\epsilon < 1/4$ and that there exists $j \in k$ so that $\ww_j > 0$. 
Denote by $(\FF^{\star}, \bbeta^{\star})$ the optimal solution and $\mathrm{OPT} = \langle \ww, \bbeta^{\star} \rangle$. 
We may assume without loss of generality that $\FF^{\star}$ has no directed cycles of positive weight for any commodity. Specifically, for every $j \in [k]$, there does 
not exist $\lambda > 0$ and $C$ a directed cycle so that $\FF_{je} \ge \lambda, \forall e \in C$. This is due to the fact we can eliminate such a cycle and have $\FF$ route the same demand.

Let $C_1$ be a constant so that $\|\uu\|_\infty \|\DD\|_\infty \le m^{C_1}$, which exists since our \mcf{} instance is poly-regular. 
We show that $m^{-C_1 - 1} \le \|\bbeta^{\star}\|_{\infty} \le m^{C_1+1}$. To see why this is the case, pick $\Bar{j} \in [k]$ arbitrarily and let $\Bar{v} \in V$ such that $\DD_{\Bar{v}\Bar{j}} \ne 0$. Since $\DD$ has integer entries, we have $|\DD_{\Bar{v} \Bar{j}}| \ge 1$. Thus, there exists $e \in [m]$ adjacent to $v$ so that $\FF^{\star}_{e\Bar{j}} \ge \frac{\bbeta^{\star}}{m}$. 
Note that $1 \ge \congest(\FF^{\star}) \ge \frac{\FF^{\star}_{e\Bar{j}}}{\uu_e} \ge \bbeta_j^{\star} / m$. We thus obtain $\bbeta_{\Bar{j}}^{\star} \le m \|\uu\|_{\infty} \le m^{C_1+1}$. Using this argument for every $\Bar{j}$ implies $\|\bbeta^{\star}\|_{\infty} \le m^{C_1+1}$. 
Next, we show that $m^{-C_1 - 1} \le \|\bbeta^{\star}\|_{\infty}$. 
First, note that since $w \ge 0$, we must have that $\congest(\FF^{\star}) = 1$. 
Let $\Bar{e} \in E$ be an edge that is saturated, meaning $\sum_{j \in [k]} \FF^{\star}_{\Bar{e}j} = \uu_{\Bar{e}}$. We show that $\uu_{\Bar{e}} \ge \min_{v \in V, j \in [k]} \bbeta^{\star}_j \DD_{vj} / m$. For this, we perform the following procedure. While there exists $v \in V, j \in [k]$ with , where $\im(\FF^{\star}) = \DD \ddiag(\bbeta^{\star})$, pick an edge $e$ adjacent to $v$ so that $\ff_{je}$ is the largest, and delete $\ff_{je}$ along a path. This process terminates in a finite number of steps and arrives at a new flow $\FF^{\star}_{\mathrm{clean}}$ with $\im(\FF^{\star}_{\mathrm{clean}}) = \veczero^{V \times k}$. Since we assumed $\FF^{\star}$ has no directed cycles of positive weight, we have that $\FF^{\star}_{\mathrm{clean}}$ is the null flow ($\FF^{\star}_{\mathrm{clean}} = \veczero^{m \times k}$). Moreover, each operation removes at least $\frac{\min_{v \in V, j \in [k]} \bbeta^{\star}_j \DD_{vj}}{m}$ amount of flow on each edge, which proves $\uu_{\Bar{e}} \ge \min_{v \in V, j \in [k]} \bbeta^{\star}_j \DD_{vj} / m$. 
Consequently, we obtain that there exists $j \in [k]$ so that $\bbeta^{\star}_j \ge m^{-C_1 - 1}$. 

Hence, we can set $\boxbound = m^{C_1+1}$ to be our range parameter, in the sense that we work over the space $\FF \in [0, \boxbound]^{E \times k}, \bbeta \in [0, \boxbound]^{k}$. 
Additionally, let $C_2$ be a constant so that $\|\ww\|_{\infty} \le m^{C_2}$. Since $\|\bbeta^{\star}\|_\infty \ge m^{-C_1-1}$, 
we can bound $m^{-(C_1 + C_2 + 1)} \le \OPT \le m^{C_1 + C_2 + 1}$. 

The algorithm for proving \Cref{coro:approxMaxKCommFlow} is as follows. For every $\lambda \in S$, where $S = \{2^{-\ell + 1} \cdot m^{C_1 + C_2 + 1} \mid \ell \in [\lceil \log_{3/2} m^{2 (C_1 + C_2 + 1)} \rceil]\}$ we call the algorithm given by \Cref{thm:mainKCommFlow} on graph instance $(G, \uu, \DD)$, range bounds $\RdUpper = \boxbound$ and $\RdLower = 0$, error parameters $\epsilon / 10, \delta = \epsilon / 100$, and cost function $c(\FF, \bbeta) = - \lambda \sum_{i \in [k]} w_i \bbeta_i$. For every such value of $\lambda$, denote $(\FF_{\lambda}, \bbeta_{\lambda})$ as the output received. 
For each $\lambda \in S$,
we let $(\FF_{\lambda}', \bbeta_{\lambda}') = t (\FF_{\lambda}, \bbeta_{\lambda})$ be the flow  rescaled by the proper amount $t > 0$ to ensure $\congest(\FF_{\lambda}') = 1$, unless $\congest(\FF_{\lambda}') = 0$, in which case we do not rescale. We then return $\argmax_{\FF_{\lambda}', \lambda \in S} \langle \ww, \bbeta_{\lambda}' \rangle$. 

We now prove that this algorithm achieves the stated guarantees. First, note that the number of values of $\lambda$ we are considering is $\otilde(1)$, and each call to the algorithm in \Cref{thm:mainKCommFlow} takes $\almostTime(mk \eps^{-1})$-time. Hence, the runtime of our algorithm is $\almostTime(mk \eps^{-1})$-time, so it suffices to show correctness. For this, it suffices to show that for some $\lambda \in S$, $\langle \ww, \bbeta_{\lambda}' \rangle \ge \frac{\OPT}{1+\epsilon}$. Let $\lambda^{\star}$ be an element in $S \cap [\frac{3}{2 \OPT}, \frac{2}{\OPT}]$. By the way we defined $S$, we know this intersection is non-empty, so $\lambda^{\star}$ is well-defined. 

We show that $\langle w, \bbeta_{\lambda^{\star}}' \rangle \ge \frac{\OPT}{1+\epsilon}$. 
For this, let $t^{\star} > 0$ be the value for which $(\FF_{\lambda^{\star}}', \bbeta_{\lambda^{\star}}') = t^{\star} (\FF_{\lambda^{\star}}, \bbeta_{\lambda^{\star}})$. First, we show $t^{\star} \le 2$, which is equivalent to showing $\congest(\FF_{\lambda^{\star}}) \ge \frac{1}{2}$. For this, first note that by definition of $\FF^{\star}$, we have that for any $t > 0$, 
\[\max_{\im(\FF) = \DD \ddiag(\bbeta) \mid \congest(\FF) = t} \langle \ww, \bbeta \rangle = t \langle \ww, \bbeta^{\star} \rangle.\]
Next, 
note that \Cref{thm:mainKCommFlow} implies 
\[-\lambda^{\star} \langle \ww, \bbeta_{\lambda^{\star}} \rangle + \congest(\FF_{\lambda^{\star}}) \le -\lambda^{\star} \langle w, \bbeta^{\star} \rangle + \congest(\FF^{\star}) + \epsilon / 8.\]
Hence, \[-\congest(\FF_{\lambda^{\star}}) \lambda^{\star} \langle w, \bbeta^{\star} \rangle + \congest(\FF_{\lambda^{\star}}) \le -\lambda^{\star} \langle w, \bbeta^{\star} \rangle + 1 + \epsilon / 8,\]
where we used $\epsilon/10 + \epsilon/100 \le \epsilon / 8$. 
Since $\lambda \ge \frac{3}{2 \OPT} = \frac{3}{2 \langle w, \bbeta^{\star} \rangle}$, 
this implies $1-\congest(\FF_{\lambda^{\star}}) \le \epsilon / 8$, so $\congest(\FF_{\lambda^{\star}}) \ge 1 - \epsilon/8 \ge \frac{1}{2}$. 

Next, by \Cref{thm:mainKCommFlow}, we also have that 
\[-\lambda^{\star} \langle \ww, \bbeta_{\lambda^{\star}} \rangle + \congest(\FF_{\lambda^{\star}}) \le -\lambda^{\star} \cdot \congest(\FF_{\lambda^{\star}}) \langle w, \bbeta^{\star} \rangle + \congest(\FF_{\lambda^{\star}}) + \epsilon / 8,\]
by comparing with feasible solution $(\congest(\FF_{\lambda^{\star}}) \cdot \FF^{\star}, \congest(\FF_{\lambda^{\star}}) \cdot \bbeta^{\star})$. By cancellation of term $\congest(\FF_{\lambda^{\star}})$ on both sides and 
multiplying both sides by $\frac{t^{\star}}{\lambda^{\star}}$, we obtain \[\langle \ww, \bbeta_{\lambda^{\star}} \rangle \ge \langle \ww, \bbeta^{\star} \rangle - \frac{\epsilon}{3} \cdot \frac{t^{\star}}{\lambda^{\star}}.\]
As $t^{\star} \le 2$, and $1/\lambda^{\star} \ge 2/\OPT$, we get \[\langle \ww, \bbeta_{\lambda^{\star}} \rangle \ge \OPT (1 - \epsilon / 2) \ge \frac{\OPT}{1+\epsilon}.\]
\end{proof}

\section*{Acknowledgements}

We thank Yujia Jin for helpful comments and discussions. 
Part of the work was done while Li Chen was at the Carnegie Mellon University and was supported by NSF Grant CCF-2330255.
Andrei Graur was supported in part by the Nakagawa departmental fellowship award from the Management Science and Engineering Department at Stanford University, NSF CAREER Grant CCF1844855, and NSF Grant CCF-1955039. Aaron Sidford was supported in part by a Microsoft Research Faculty Fellowship, NSF CAREER Grant CCF1844855, NSF Grant CCF-1955039, and a PayPal research award, and a Sloan Research Fellowship. Part of this work was conducted while authors were visiting the Simons Institute for the Theory of Computing.

\bibliographystyle{alpha}
\bibliography{bib.bib}

\appendix

\section{Implementing Gradient and Extragradient Steps}
\label{apx:implementing_steps}

As mentioned in \Cref{subsec:implement_oracles}, in this section, we prove \Cref{lem:xgrad_steps_implementability}.
In particular, we provide algorithms (\Cref{alg:grad_step} and \Cref{alg:exgrad_step}) that implement our gradient and extragradient oracles (\Cref{defn:gen_apx_grad}, \Cref{defn:gen_apx_exgrad})
in the specialized setting introduced in \Cref{subsec:implement_oracles}, where we work under \Cref{assm:4.2space} and \Cref{def:opt_oracles}.
The algorithms which implement these oracles (\Cref{alg:grad_step} and \Cref{alg:exgrad_step}) are adaptations of algorithms 2 and 3, respectively, in \cite{JT23}. The main difference is replacing the exact minimization on space $\uset \subseteq \R^n$ (used in lines 2-3 of algorithms 2 and 3 in \cite{JT23}) with an oracle that outputs an approximate minimizer to a composite minimization problem (\Cref{defn:apx_best_resp}). The analysis is very similar to the other paper. 

For the rest of the section, for simplicity, we work with $\psi$ defined only on $\uset$. 
We start with the pseudocode and proof of correctness for the gradient oracle. 

\begin{algorithm}[htp!]
\caption{Composite gradient step oracle}\label{alg:grad_step}
\KwData{$\zz = (\uu, \vv) \in \zset$, $\hh = (\hh^{\uset}, \hh^{\vset}) \in \zset^*$, 
$\AA \in \R^{m \times n}, \alpha, \beta \ge 0$, convex function $\psi$, access to gradient operator $\gg:\zset \to \zset^*$, access to a $\delta$-approximate composite best response oracle $\oracle^{\psi, r^{\alpha}}_{\mathrm{BR}}: \uset \times \uset^{\star} \to \uset$ with respect to $\psi$ and $r^{\alpha}$}
\SetKwProg{Fn}{Function}{:}{}
\SetKwFunction{gradoracle}{ApxGradStep}
\Fn{\gradoracle{$\zz, \hh, \gg, \alpha, \beta, \psi$}}{
    $\uu_1 = \oracle^{\psi, r^{\alpha}}_{\mathrm{BR}}(\zz^{\vset}, \hh^{\uset} - \nabla_{\uset} r^{\alpha}(\zz))$\;
    $\vv_1 = \argmin_{\yy \in \vset} \langle \hh^{\vset} - \nabla_{\vset} r^{\alpha}((\uu_1, \zz^{\vset})) - \nabla_{\vset} r^{\alpha}(\zz), \yy \rangle + V^{\beta r_2}_{\zz^{\vset}}(\yy)$\;
    $\uu_2 = \oracle^{\psi, r^{\alpha}}_{\mathrm{BR}}(\vv_1, \hh^{\uset} - \nabla_{\uset} r^{\alpha}(\zz))$\;
    \Return $\zz' = (\uu_2, \vv_1)$
}
    
\end{algorithm}

\Cref{alg:grad_step} differs from Algorithm 2 in \cite{JT23} in the lines which define $\uu_1, \uu_2$, as these best response steps on $\uset$ are now approximate instead of exact. Having defined block compatible regularizers in \Cref{subsec:implement_oracles}, we also define several quantities that pertain to the Bregman divergence with respect to $r^{(\alpha)}$. 
First, we let \begin{equation}\label{not:param_breg}
    V_{\zz}^{(\alpha)}(\zz') := V_{\zz}^{(r^{(\alpha)})}(\zz').
\end{equation} 
In spirit of \eqref{not:param_breg}, we also define \[\Delta^{(\alpha)}_{\zz}(\zz',\ww) = V_{\zz}^{(\alpha)}(\ww) - V_{\zz'}^{(\alpha)}(\ww) - V_{\zz}^{(\alpha)}(\zz').\]
We now provide the lemma about the correctness of \Cref{alg:grad_step}.

\begin{lemma}[Composite gradient step]\label{lem:apx_grad_step}
    Let $\alpha, \beta, \delta \ge 0$ be parameters with $\beta \ge \alpha$ so that $r^{\alpha}$ is jointly convex over $\uset \times \vset$. Let $\psi$ be a convex function on $\uset$. 
    Assuming access to a $\delta$-approximate composite best response oracle $\oracle_{\mathrm{BR}}^{\psi, r^{\alpha}}$ with respect to $\psi$ and $r^{\alpha}$, \Cref{alg:grad_step} is a $3 \delta$-approximate gradient step oracle with respect to convex function $\psi$ and regularizers $(r^{\alpha}, \beta r_2)$. 
\end{lemma}

Before presenting the proof, we introduce a helper lemma that helps us reason about quotient functions. In particular, a function that shows up in the analysis of our correctness lemmas, and that is useful for the intuition regarding the optimization method that our algorithms are similar to, is $\Phi_{\hh}:\vset \to \vset$ defined by \[\Phi_{\hh}(\yy) := \min_{\xx \in \uset} \langle \hh, (\xx, \yy) \rangle + r^{\alpha}(\xx, \yy) + \psi(\xx),\]
for $\hh \in \zset^*$ some vector. In the context of \Cref{alg:grad_step} and \Cref{alg:exgrad_step}, vector $\hh \in \uset^{\star}$ is depends on the vector $\hh \in \zset^*$, and pivot $\zz \in \zset$ (also pivot $\vv\in \vset$ in the case of \Cref{alg:exgrad_step}) received as input.

\begin{lemma}[Lemmas 4 and 5 in \cite{JT23} restated]\label{lem:joint_convex4&5}
    Let $\uset \subseteq \R^n$ and $\vset \subseteq \R^m$ be convex compact subsets. Suppose $G:\uset \times \vset \to \R$ is jointly convex over its argument $(\xx, \yy) \in \uset \times \vset$. For $\yy \in \vset$, define $\xx_{\mathrm{br}}(\yy) = \argmin_{\xx \in \uset} G(\xx, \yy)$ and $\rho(\yy) = G(\xx_{\mathrm{br}}(\yy), \yy)$. Then, for all $\yy \in \vset$, $\partial_{\yy} G(\xx_{\mathrm{br}}(\yy), \yy) \subset \partial \rho(\yy)$.
    Furthermore, suppose that for any $\xx \in \uset$, $G(\xx, \cdot)$ always is convex $\Lambda:\vset \to \R$ plus a linear function (which may depend on $\xx$). Then, $\Lambda - \rho$ is convex, and $\rho-\Lambda'$ is convex for any $\Lambda':\yset \to \R$ such that $G(\xx,\yy) - \Lambda'(\yy)$ is jointly convex in $(\xx,\yy)$. 
\end{lemma}

A direct corollary of this lemma is \Cref{cor:relative_sm_str}, which proves that for any $\hh \in \zset^*$, the function $\Phi_{\hh}$ is $\alpha$-relatively smooth with respect to $r_2$. This corollary is used in proving both \Cref{lem:apx_grad_step} and \Cref{lem:apx_exgrad_step}.

\begin{corollary}\label{cor:relative_sm_str}
    Let $\{r^{\alpha}\}$ be a family of block compatible regularizers, with blocks $r_1, r_2$ and $\alpha_0 > 0$ the smallest coefficient so that $r^{\alpha'}$ is jointly convex over $\uset \times \vset$ for every $\alpha' \ge \alpha_0$. Let $\alpha \ge \alpha_0$
    For $\hh \in \zset^*$,
    define $\Phi_{\hh}:\vset \to \vset$ by \[\Phi_{\hh}(\yy) := \min_{\xx \in \uset} \langle \hh, (\xx, \yy) \rangle + r^{\alpha}(\xx, \yy) + \psi(\xx),\]
    and \[\xx_{BR}(\yy) := \argmin_{\xx \in \uset} \langle \hh, (\xx, \yy) \rangle + r^{\alpha}(\xx, \yy) + \psi(\xx) .\]
    Then, $\Phi$ is $\alpha$-relative smooth and $\alpha - \alpha_0$ strongly convex with respect to $r_2$. 
    Moreover, for any $\aa, \bb, \cc \in \vset$, we have 
    \[\langle \nabla_{\vset} r^{\alpha}(\xx_{\mathrm{BR}}(\aa), \aa) - \nabla_{\vset} r^{\alpha}(\xx_{\mathrm{BR}}(\bb), \bb), \aa - \cc \rangle \le V_{\aa}^{\alpha r_2}(\cc) + V_{\bb}^{\alpha r_2}(\aa) . \]
\end{corollary}

\begin{proof}
    To show relative smoothness, it suffices to prove that for any $\yy_1, \yy_2 \in \vset$, \[\Phi_h(\yy_2) \le \Phi_h(\yy_1) + \langle \nabla \Phi_h(\yy_1), \yy_2 - \yy_1 \rangle + V_{\yy_1}^{\alpha r_2}(\yy_2).\]
    Note that \Cref{lem:joint_convex4&5} implies that the function $\PhiHat(\yy) := \alpha r_2(\yy) - \Phi_h(\yy)$ is convex. Therefore,
    \[\PhiHat(\yy_2) \ge \PhiHat(\yy_1) + \langle \nabla \PhiHat(\yy_1), \yy_2 - \yy_1 \rangle,\]
    which implies that \[\alpha r_2(\yy_2) - \Phi_h(\yy_2) \ge \alpha r_2(\yy_1) - \Phi(\yy_1) + \langle \nabla \alpha r_2(\yy_1) - \nabla \Phi(\yy_1), \yy_2 - \yy_1 \rangle.\]
    Using that $\alpha r_2(\yy_2) - \alpha r_2(\yy_1) - \langle \alpha \nabla r_2(\yy_1), \yy_2 - \yy_1 \rangle = V_{\yy_1}^{\alpha r_2}(\yy_2)$, we obtain the desired inequality. 

    To show relative strong convexity, it suffices to prove that for any $\yy_1, \yy_2 \in \vset$, \[\Phi_h(\yy_2) \ge \Phi_h(\yy_1) + \langle \nabla \Phi_h(\yy_1), \yy_2 - \yy_1 \rangle + V_{\yy_1}^{(\alpha - \alpha_0) r_2}(\yy_2).\]
    Define $\PhiBar_h(\yy) := \Phi_h(\yy) - (\alpha - \alpha_0) r_2(\yy)$, which, by \Cref{lem:joint_convex4&5}, is convex, as $\langle \hh, (\xx, \yy) \rangle + r^{\alpha_0}(\xx, \yy) + \psi(\xx)$ is jointly convex on $\uset \times \vset$. Thus, we have that \[\PhiBar_h(\yy_2) \ge \PhiBar_h(\yy_1) + \langle \nabla \PhiBar_h(\yy_1), \yy_2 - \yy_1 \rangle,\]
    which implies that \[\Phi_h(\yy_2) - (\alpha - \alpha_0) r_2(\yy_2) \ge \Phi_h(\yy_1) - (\alpha - \alpha_0) r_2(\yy_1) + \langle \nabla \Phi(\yy_1) - (\alpha - \alpha_0) \nabla r_2(\yy_1), \yy_2 - \yy_1 \rangle.\]
    Using that $(\alpha - \alpha_0) (r_2(\yy_2) - r_2(\yy_1)) - \langle (\alpha - \alpha_0) \nabla r_2(\yy_1), \yy_2 - \yy_1 \rangle = V_{\yy_1}^{(\alpha - \alpha_0) r_2}(\yy_2)$, we obtain the desired inequality. 

    For the ``Moreover'' part of the corollary, note that 
    \[\langle \nabla_{\vset} r^{\alpha}(\xx_{\mathrm{BR}}(\aa), \aa) - \nabla_{\vset} r^{\alpha}(\xx_{\mathrm{BR}}(\bb), \bb), \aa - \cc \rangle = \langle \nabla \Phi_h(\aa) - \nabla \Phi_h(\bb), \yy - \cc \rangle,\]
    where we have used that $\nabla_{\vset} r^{\alpha}(\xx_{\mathrm{BR}}(\yy), \yy) = \nabla \Phi_h(y) - h, \enspace \forall \yy \in \zset$. 
    Now, by identity \eqref{def_3_breg_opt}, we have that \[\langle \nabla \Phi_h(\aa) - \nabla \Phi_h(\bb), \yy - \cc \rangle \le V_{\aa}^{\Phi_h}(\cc) + V_{\bb}^{\Phi_h}(\aa).\]
    By the first part of this corollary, $\Phi_h$ is $\alpha$-smooth relative to $r_2$, 
    which implies $V_{\aa}^{\Phi_h}(\cc) + V_{\bb}^{\Phi_h}(\aa) \le V_{\aa}^{\alpha r_2}(\cc) + V_{\bb}^{\alpha r_2}(\aa)$. 
\end{proof}

The relative strong convexity property of \Cref{cor:relative_sm_str} is not used in the proof of correctness of our oracles (unlike the relative smoothness), yet it perhaps offers some insight into why these methods work. In some sense, since we work with a family of regularizers and parameters $\alpha$ with the property that $\alpha = C \alpha_0$ for some constant $C$, we have that the subproblems corresponding to the gradient and extragradient steps correspond to approximately minimizing a function that is $\Theta(1)$ strongly convex and smooth relative to $r_2$. Standard gradient methods show that one gradient descent step for minimizing such functions make good enough progress. Inspecting the pseudocode of our implementations, we can notice that this is in fact close to what \Cref{alg:grad_step} and \Cref{alg:exgrad_step} are doing. 
We are now ready to prove our first lemma about correctness (\Cref{lem:apx_grad_step}). 

\begin{proof}[Proof of \Cref{lem:apx_grad_step}]
    Using notation in \Cref{sec:extragrad_framework}, it suffices to show \[\mathrm{regret}_{\hh, \psi}((\uu_2, \vv_1); \ww) \le \Delta^{(\alpha)}_{\zz}((\uu_2, \vv_1), \ww) + V_{\zz^{\vset}}^{\beta r_2}(\ww^{\vset}) + 3 \delta, \enspace \forall \ww \in \zset.\]
    Using identity \Cref{def_3_breg_opt}, this is equivalent to showing \[\mathrm{regret}_{\hh + \nabla r^{(\alpha)}(\uu_2, \vv_1) - \nabla r^{(\alpha)}(\zz), \psi}((\uu_2, \vv_1); \ww) \le V_{\zz^{\vset}}^{\beta r_2}(\ww^{\vset}) + 3 \delta, \enspace \forall \ww \in \zset.\]
    Hence, we fix $\ww = (\ww^{\uset}, \ww^{\vset}) \in \zset$, and show the inequality above. 
    By definition of $\uu_2$ and property 2 of \Cref{defn:apx_best_resp}, we have that \[\langle \hh^{\uset} + \nabla_{\uset} r^{\alpha}(\uu_2, \vv_1) - \nabla_{\uset} r^{\alpha}(\zz), \uu_2 - \xx \rangle \le \psi(\xx) - \psi(\uu_2) + \delta, \enspace \forall \xx \in \uset. \] 
    which, for $\xx \leftarrow \ww^{\uset}$ yields \begin{equation}\label{ineq:x-side_ready_grad}
        \langle \hh^{\uset} + \nabla_{\uset} r^{\alpha}(\uu_2, \vv_1) - \nabla_{\uset} r^{\alpha}(\zz), \uu_2 - \ww^{\uset} \rangle \le  \psi(\ww^{\uset}) - \psi(\uu_2) + \delta. 
    \end{equation} 
    Hence, it suffices to show \begin{equation}\label{ineq:y-side_ready_grad}
        \langle \hh^{\vset} + \nabla_{\vset} r^{\alpha}(\uu_2, \vv_1) - \nabla_{\vset} r^{\alpha}(\zz), \vv_1 - \ww^{\vset} \rangle \le V_{\zz^{\vset}}^{\beta r_2}(\ww^{\vset}) + 2 \delta.
    \end{equation}
    The rest of the proof is allotted to showing \eqref{ineq:y-side_ready_grad}.

    First, note that the optimality condition with respect to $\vv_1$ implies that \begin{equation}\label{ineq:opt_dirty_grad}
        \langle \hh^{\vset} + \nabla_{\vset} r^{\alpha}(\uu_1, \zz^{\vset}) - \nabla_{\vset} r^{\alpha}(\zz), \vv_1 - \yy \rangle \le V_{\zz^{\vset}}^{\beta r_2}(\yy) - V_{\vv_1}^{\beta r_2}(\yy) - V_{\zz^{\vset}}^{\beta r_2}(\vv_1), \enspace \forall \yy \in \vset. 
    \end{equation} 
    Hence, setting $\yy \leftarrow \ww^{\vset}$ implies that it suffices to show \begin{equation}\label{ineq:clean_grad_dif_grad}
        \langle \nabla_{\vset} r^{\alpha}(\uu_2, \vv_1) - \nabla_{\vset} r^{\alpha}(\uu_1, \zz^{\vset}), \vv_1 - \ww^{\vset} \rangle \le V_{\vv_1}^{\beta r_2}(\ww^{\vset}) + V_{\zz^{\vset}}^{\beta r_2}(\vv_1) + 2 \delta.
    \end{equation}

    Now, define $\Phi:\vset \to \vset$ by \[\Phi(\yy) = \min_{\xx \in \uset} \langle \hh, (\xx, \yy) \rangle + \psi(\xx) + r^{\alpha}(\xx, \yy) ,\] 
    and, also \[\xx_{\mathrm{BR}}(\yy) = \argmin_{\xx \in \uset} \langle \hh, (\xx, \yy) \rangle + \psi(\xx) + r^{\alpha}(\xx, \yy) .\] 
    Since $\uu_1 = \oracle_{\mathrm{BR}}(\psi, r^{\alpha}, \zz^{\vset}, \hh^{\uset} - \nabla_{\uset} r^{\alpha}(\zz))$, by property 2 of \Cref{defn:apx_best_resp}, we have that \[\langle \nabla_{\vset} r^{\alpha}(\xx_{\mathrm{BR}}(\zz^{\vset}), \zz^{\vset}) - \nabla_{\vset} r^{\alpha}(\uu_1, \zz^{\vset}), \vv_1 - \ww^{\vset} \rangle \le \|\nabla_{\vset} r^{\alpha}(\xx_{\mathrm{BR}}(\zz^{\vset}), \zz^{\vset}) - \nabla_{\vset} r^{\alpha}(\uu_1, \zz^{\vset})\|_* \|\vv_1 - \ww^{\vset}\| \le  \delta.\]
    Similarly, using $\uu_2 = \oracle_{\mathrm{BR}}(\psi, r^{\alpha}, \vv_1, \hh^{\uset} - \nabla_{\uset} r^{\alpha}(\zz))$, we have that \[\langle \nabla_{\vset} r^{\alpha}(\uu_2, \vv_1) - \nabla_{\vset} r^{\alpha}(\xx_{\mathrm{BR}}(\vv_1), \vv_1), \vv_1 - \ww^{\vset} \rangle \le \delta.\]
    Consequently, to show \eqref{ineq:clean_grad_dif_grad}, it suffices to show \[\langle \nabla_{\vset} r^{\alpha}(\xx_{\mathrm{BR}}(\vv_1), \vv_1) - \nabla_{\vset} r^{\alpha}(\xx_{\mathrm{BR}}(\zz^{\vset}), \zz^{\vset}), \vv_1 - \ww^{\vset} \rangle \le V_{\vv_1}^{\beta r_2}(\ww^{\vset}) + V_{\zz^{\vset}}^{\beta r_2}(\vv_1).\]
    However, this follows immediately by the ``Moreover'' part of \Cref{cor:relative_sm_str}, as $\alpha$ is picked so that $r^{\alpha}$ is jointly convex in $\uset \times \vset$. 
\end{proof}

Finally, we provide the pseudocode and the correctness guarantee for the extragradient step oracle. 

\begin{algorithm}[H]
\caption{Composite extragradient step oracle}\label{alg:exgrad_step}
\KwData{$\zz \in \zset, \auxit \in \zset$, $\hh = (\hh^{\uset}, \hh^{\vset}) \in \zset^*$, 
$\AA \in \R^{m \times n}, \alpha, \beta \ge 0$, convex function $\psi$, access to a $\delta$-approximate composite best response oracle $\oracle^{\psi / 2, r^{\alpha}}_{\mathrm{BR}}: \uset \times \uset^{\star} \to \uset$ with respect to $\psi$ and $r^{\alpha + \beta}$}
\KwResult{$\zz', \auxit'$ with guarantees stated in \Cref{lem:apx_exgrad_step}}
\SetKwProg{Fn}{Function}{:}{}
\SetKwFunction{Xgradoracle}{ApxXGradStep}
\Fn{\Xgradoracle{$\zz, \auxit, \hh, \gg, \alpha, \beta, \psi$}}{
    $\uu_1 = \oracle^{\psi/2, r^{\alpha + \beta}}_{\mathrm{BR}}(\auxit^{\vset}, \hh^{\uset} - \nabla_{\uset} r^{\alpha}(\zz))$\;
    $\vv_1 = \argmin_{\yy \in \vset} \langle \hh^{\vset} + \nabla_{\vset} r^{\alpha}(\uu_1, \auxit^{\vset}) - \nabla_{\vset} r^{\alpha}(\zz), \yy \rangle + V^{\beta r_2}_{\auxit^{\vset}}(\yy)$\;
    $\uu_2 = \oracle^{\psi/2, r^{\alpha + \beta}}_{\mathrm{BR}}(\vv_1, \hh^{\uset} - \nabla_{\uset} r^{\alpha}(\zz))$\;
    $\vv_2 = \argmin_{\yy \in \vset} \langle \hh^{\vset} + \nabla_{\vset} r^{\alpha}(\uu_2, \vv_1) - \nabla_{\vset} r^{\alpha}(\zz), \yy \rangle + V_{\auxit^{\vset}}^{\beta r_2}(\yy)$\;
    \Return $\zz' = (\uu_2, \vv_1), \auxit' = (\auxit^{\uset}, \vv_2)$\;
}
    
\end{algorithm}

\Cref{alg:exgrad_step} differs from Algorithm 3 in \cite{JT23} in the lines which define $\uu_1, \uu_2$, as these best response steps on $\uset$ are now approximate instead of exact. We now provide the lemma about the correctness of \Cref{alg:exgrad_step}. 

\begin{lemma}[Composite extragradient step]\label{lem:apx_exgrad_step}
    Let $\alpha, \beta, \delta \ge 0$ be parameters with $\beta \ge \alpha$ so that $r^{\alpha}$ is jointly convex over $\uset \times \vset$. Let $\psi$ be a convex function on $\uset$. 
    Assuming access to a $\delta$-approximate composite best response oracle $\oracle^{\psi/2, r^{\alpha}}_{\mathrm{BR}}$ with respect to $\psi / 2$ and $r^{\alpha}$, \Cref{alg:grad_step} is a $3 \delta$-approximate extragradient step oracle with respect to convex function $\psi / 2$ and regularizers $(r^{\alpha + \beta}, \beta r_2)$. 
\end{lemma}

\begin{proof}[Proof of \Cref{lem:apx_exgrad_step}]
    Using notation in \Cref{sec:extragrad_framework}, it suffices to show 
    \[\mathrm{regret}_{\hh, \psi/2}((\uu_2, \vv_1); \ww) \le \Delta^{(\alpha)}_{\zz}((\uu_2, \vv_1), \ww) + V_{(\auxit^{\uset}, \auxit^{\vset})}^{\beta r_2}(\ww) - V_{(\auxit^{\uset}, \vv_2)}^{\beta r_2}(\ww) + 3 \delta, \enspace \forall \ww \in \zset\,.\]
    Denote $\vv = \auxit^{\vset}$. Note that since $r_2(\ww)$ only depends on $\ww^{\vset}$, we have that for any $\uu', \uu'' \in \uset$, $\vv', \vv'' \in \vset$ and $\lambda \in \R$, 
    \begin{equation}\label{eq:dropping_u_part}
        V_{(\vv', \uu')}^{\lambda r_2}(\vv'', \uu'') = V_{\vv'}^{\lambda r_2}(\vv'') \,.
    \end{equation}
    Using \ref{eq:dropping_u_part}, we obtain that it suffices to show
    \[\mathrm{regret}_{\hh, \psi/2}((\uu_2, \vv_1); \ww) \le \Delta^{(\alpha)}_{\zz}((\uu_2, \vv_1), \ww) + V_{\vv}^{\beta r_2}(\ww^{\vset}) - V_{\vv_2}^{\beta r_2}(\ww^{\vset}) + 3 \delta, \enspace \forall \ww \in \zset\,.\]
    Using identity \ref{def_3_breg_opt}, along with the fact that $\nabla_{\uset} r^{\alpha}(\zz) = \nabla_{\uset} r^{\alpha+\beta}(\zz)$, this is equivalent to showing \[\mathrm{regret}_{\hh + \nabla r^{(\alpha)}(\uu_2, \vv_1) - \nabla r^{(\alpha)}(\zz), \psi/2}((\uu_2, \vv_1); \ww) \le V_{\vv}^{\beta r_2}(\ww^{\vset}) - V_{\vv_2}^{\beta r_2}(\ww^{\vset}) + 3 \delta, \enspace \forall \ww \in \zset.\]
    Hence, we fix $\ww = (\ww^{\uset}, \ww^{\vset}) \in \zset$, and show the inequality above. 
    By definition of $\uu_2$ and property 2 of \Cref{defn:apx_best_resp}, we have that \[\langle \hh^{\uset} + \nabla_{\uset} r^{\alpha}(\uu_2, \vv_1) - \nabla_{\uset} r^{\alpha}(\zz), \uu_2 - \xx \rangle \le \psi(\xx) - \psi(\uu_2) + \delta, \enspace \forall \xx \in \uset. \] 
    which, for $\xx \leftarrow \ww^{\uset}$ yields \begin{equation}\label{ineq:x-side_ready_exgrad}
        \langle \hh^{\uset} + \nabla_{\uset} r^{\alpha}(\uu_2, \vv_1) - \nabla_{\uset} r^{\alpha}(\zz), \uu_2 - \ww^{\uset} \rangle \le  \psi(\ww^{\uset}) - \psi(\uu_2) + \delta. 
    \end{equation} 
    Hence, it suffices to show \begin{equation}\label{ineq:y-side_ready_exgrad}
        \langle \hh^{\vset} + \nabla_{\vset} r^{\alpha}(\uu_2, \vv_1) - \nabla_{\vset} r^{\alpha}(\zz), \vv_1 - \ww^{\vset} \rangle \le V_{\vv}^{\beta r_2}(\ww^{\vset}) - V_{\vv_2}^{\beta r_2}(\ww^{\vset}) + 2 \delta.
    \end{equation}
    The rest of the proof is allotted to showing \eqref{ineq:y-side_ready_exgrad}. 

    First, note that the optimality condition on $\vv_1$ implies \begin{equation}\label{ineq:opt_v1_exgrad}
        \langle \hh^{\vset} + \nabla_{\vset} r^{\alpha}(\uu_1, \vv) - \nabla_{\vset} r^{\alpha}(\zz), \vv_1 - \yy \rangle \le \Delta^{\beta r_2}_{\vv}(\vv_1, \yy), \enspace \forall \yy \in \vset.
    \end{equation}
    Similarly, the optimality condition on $\vv_2$ implies
    \begin{equation}\label{ineq:opt_v2_exgrad}
        \langle \hh^{\vset} + \nabla_{\vset} r^{\alpha}(\uu_2, \vv_1) - \nabla_{\vset} r^{\alpha}(\zz), \vv_2 - \yy \rangle \le \Delta^{\beta r_2}_{\vv}(\vv_2, \yy), \enspace \forall \yy \in \vset.
    \end{equation}
    Plugging in $\yy \leftarrow \vv_2$ in \eqref{ineq:opt_v1_exgrad} and $\yy \leftarrow \ww^{\vset}$ in \eqref{ineq:opt_v2_exgrad}, we obtain, after cancellation of $\langle \hh^{\vset} + \nabla_{\vset} r^{\alpha}(\uu_1, \vv) - \nabla_{\vset} r^{\alpha}(\zz), \vv_2 \rangle$, 
    \begin{align}\label{ineq:v1v2mess}
        & \langle \hh^{\vset} + \nabla_{\vset} r^{\alpha}(\uu_2, \vv_1) - \nabla_{\vset} r^{\alpha}(\zz), \vv_1 - \ww^{\vset} \rangle + \langle \nabla_{\vset} r^{\alpha}(\uu_1, \vv) - \nabla_{\vset} r^{\alpha}(\uu_2, \vv_1), \vv_1 - \vv_2 \rangle \le \\
        & \Delta^{\beta r_2}_{\vv}(\vv_1, \vv_2) + \Delta^{\beta r_2}_{\vv}(\vv_2, \ww^{\vset})
    \end{align}
    Hence, to show \eqref{ineq:y-side_ready_exgrad}, it suffices to show \[\langle \nabla_{\vset} r^{\alpha}(\uu_1, \vv) - \nabla_{\vset} r^{\alpha}(\uu_2, \vv_1), \vv_1 - \vv_2 \rangle \ge - V^{\beta r_2}_{\vv}(\vv_1) - V^{\beta r_2}_{\vv_1}(\vv_2) - 2 \delta,\]
    which is equivalent to showing \begin{equation}\label{ineq:grad_dif_exgrad}
        \langle \nabla_{\vset} r^{\alpha}(\uu_2, \vv_1) - \nabla_{\vset} r^{\alpha}(\uu_1, \vv), \vv_1 - \vv_2 \rangle \le V^{\beta r_2}_{\vv}(\vv_1) + V^{\beta r_2}_{\vv_1}(\vv_2) + 2 \delta.
    \end{equation} 

    Now, as in the proof of \Cref{lem:apx_grad_step}, define $\Phi:\vset \to \vset$ by \[\Phi(\yy) = \min_{\xx \in \uset} \langle \hh, (\xx, \yy) \rangle + r^{\alpha}(\xx, \yy) + \psi(\xx) ,\] 
    and, also \[\xx_{\mathrm{BR}}(\yy) = \argmin_{\xx \in \uset} \langle \hh, (\xx, \yy) \rangle + r^{\alpha}(\xx, \yy) + \psi(\xx) .\] 
    To prove \eqref{ineq:grad_dif_exgrad}, we show that $\langle \nabla_{\vset} r^{\alpha}(\xx_{\mathrm{BR}}(\vv_1), \vv_1) - \nabla_{\vset} r^{\alpha}(\xx_{\mathrm{BR}}(\vv), \vv), \vv_1 - \vv_2 \rangle \le V^{\beta r_2}_{\vv}(\vv_1) + V^{\beta r_2}_{\vv_1}(\vv_2)$, which suffices, since both $\langle \nabla_{\vset} r^{\alpha}(\uu_2, \vv_1) - \nabla_{\vset} r^{\alpha}(\xx_{\mathrm{BR}}(\vv_1), \vv_1), \vv_1 - \vv_2 \rangle$ and $\langle \nabla_{\vset} r^{\alpha}(\xx_{\mathrm{BR}}(\vv), \vv) - \nabla_{\vset} r^{\alpha}(\uu_1, \vv), \vv_1 - \vv_2 \rangle$ can be bounded by $\delta$, by property $2$ of \Cref{defn:apx_best_resp}. This follows immediately by the ``Moreover'' part of \Cref{cor:relative_sm_str}, as $\alpha$ is picked so that $r^{\alpha}$ is jointly convex in $\uset \times \vset$. 
\end{proof}

We are now ready to prove \Cref{lem:xgrad_steps_implementability}. 

\begin{proof}[Proof of \Cref{lem:xgrad_steps_implementability}]
    The implementability of $\oracle_{\mathrm{grad}}^{\psi, r^\alpha, \beta r_2}$ and $\oracle_{\mathrm{xgrad}}^{\psi, r^{\alpha+\beta}, \beta r_2}$ follows from \Cref{lem:apx_grad_step} and \Cref{lem:apx_exgrad_step} respectively. Specifically, the implementations are \Cref{alg:grad_step} and \Cref{alg:exgrad_step} respectively. It is easy to see that both of these call the oracles mentioned in \Cref{def:opt_oracles} at most $2$ times. 
\end{proof}

\section{From Constrained to Unconstrained Optimization}
\label{subsec:unconstrained_red}

In this section, we prove our lemmas that pertain to the binary search procedures used in \Cref{sec:L1infMin}, such as \Cref{lem:unconstrained_red}, and \Cref{lem:BRFromAdditiveError}. 

\begin{proof}[Proof of \Cref{lem:unconstrained_red}]
    First, we show that for any $C_1 < C_2 \in [0, M]$, we have that \[|\zeta(\XX(C_1)) - \zeta(\XX(C_2))| \le L \sqrt{\frac{4 \delta + 2 (C_2 - C_1) R}{\alpha}}. \]
    To show this, note that by definition of $\XX(C_1), \XX(C_2)$, we have
    \begin{equation}\label{ineq:gen_C1_cond}
        \Gamma(\XX(C_1)) + C_1 \zeta(\XX(C_1)) \le \Gamma(\XX) + C_1 \zeta(\XX) + \delta, \enspace \forall \XX \in \xset ,
    \end{equation}
    and 
    \begin{equation}\label{ineq:gen_C2_cond}
        \Gamma(\XX(C_2)) + C_2 \zeta(\XX(C_2)) \le \Gamma(\XX) + C_2 \zeta(\XX) + \delta, \enspace \forall \XX \in \xset.
    \end{equation}
    In particular, \ref{ineq:gen_C2_cond} implies, since $\zeta(\XX(C_2)), \zeta(\XX) \in [\RdLower, \RdUpper]$, that \[\Gamma(\XX(C_2)) + C_1 \zeta(\XX(C_2)) \le \Gamma(\XX) + C_1 \zeta(\XX) + 2 (C_2 - C_1) R + \delta, \enspace \forall \XX \in \xset.\]
    Combining this with \ref{ineq:gen_C1_cond} for
    $\Bar{\XX} = \frac{\XX(C_1) + \XX(C_2)}{2}$, we obtain \[\Gamma(\XX(C_1)) + \Gamma(\XX(C_2)) + C_1 [\zeta(\XX(C_1)) + \zeta(\XX(C_2))] \le 2 \Gamma(\Bar{\XX}) + 2 C_1 \zeta(\Bar{\XX}) + 2 (C_2 - C_1) R + 2 \delta.\]
    By convexity of $\zeta(\cdot)$, $C_1 \zeta(\XX(C_1)) + C_1 \zeta(\XX(C_2)) \ge 2 C_1 \zeta(\Bar{\XX})$. Hence, this implies 
    \[\Gamma(\XX(C_1)) + \Gamma(\XX(C_2)) \le 2 \Gamma(\Bar{\XX}) + 2 (C_2 - C_1) R + 2 \delta.\]
    However, strong convexity implies \[\Gamma(\XX(C_1)) + \Gamma(\XX(C_2)) \ge 2 \Gamma(\Bar{\XX}) + \frac{\alpha}{2} \|\XX(C_1) - \XX(C_2)\|^2.\]
    Thus, \[\|\XX(C_1) - \XX(C_2)\|^2 \le \frac{4 \delta + 2 (C_2 - C_1) R}{\alpha}.\]
    Since $\zeta$ is $L$-Lipschitz with respect to $\|\cdot\|$, we have \[|\zeta(\XX(C_1)) - \zeta(\XX(C_2))| \le L \sqrt{\frac{4 \delta + 2 (C_2 - C_1) R}{\alpha}},\]
    as claimed. 

    Now, 
    let $\XX^* = \argmin_{\XX \in \sset} \Gamma(\XX)$ and $\mathrm{OPT} = \Gamma(\XX^*)$. 
    First, note that if $\XX(0) \in \sset$, then clearly $\Gamma(\XX(0)) \le \Gamma(\XX^*) + \delta$, which implies that returning $\XX(0)$ would be a valid output in this case. Thus, we may assume $\XX(0) \notin \sset$. 
    Now, suppose we found $C_1 < C_2$ so that $C_2 - C_1 < \frac{\delta}{R}$ and $\XX(C_1) \notin \sset$ and $\XX(C_2) \in \sset$. Then, returning $\XX(C_2)$ would be a valid output (i.e. $\Gamma(\XX(C_2)) \le \Gamma(\XX^*) + 2 \delta$. To see why this is true, note that by definition of $\XX(C_1)$, we have 
    \begin{equation}\label{ineq:C1toOPT_cond}
        \Gamma(\XX(C_1)) + C_1 \zeta(\XX(C_1)) \le \Gamma(\XX^*) + C_1 \zeta(\XX^*) + \delta.
    \end{equation}
    Since $\XX(C_1) \notin \sset, \XX^* \in \sset$, we have that $\zeta(\XX(C_1)) > \zeta(\XX^*)$. Hence, we obtain \[\Gamma(\XX(C_1)) \le \Gamma(\XX^*) + \delta.\]
    Next, by definition of $\XX(C_2)$, we have 
    \begin{equation}\label{ineq:C2toC1_cond}
        \Gamma(\XX(C_2)) + C_2 \zeta(\XX(C_2)) \le \Gamma(\XX(C_1)) + C_2 \zeta(\XX(C_1)) + \delta.
    \end{equation}
    Combining this with the previous inequality, we obtain \[\Gamma(\XX(C_2)) \le \Gamma(\XX^*) + C_2 [\zeta(\XX(C_1)) - \zeta(\XX(C_2))] + \delta.\]
    Using the bound on $|\zeta(\XX(C_1)) - \zeta(\XX(C_2))|$, and that $C_2 \le M$, we obtain 
    \[\Gamma(\XX(C_2)) \le \Gamma(\XX^*) + M \cdot L \cdot \sqrt{\frac{4 \delta + 2 (C_2 - C_1) R}{\alpha}} + \delta \le M L \sqrt{\frac{6 \delta}{\alpha}} + \delta.\]

    Hence, to find a $M L \sqrt{\frac{6 \delta}{\alpha}} + \delta$ approximate minimizer of $\Gamma$ in $\sset$, it suffices to find $C_1 < C_2$ with $C_2 - C_1 < \frac{\delta}{R}$ and $\XX(C_1) \notin \sset$ and $\XX(C_2) \in \sset$. To do so, we perform a binary search procedure, where at each step we keep an interval $[a, b]$ (starting with $a = 0, b = M$) with the property that $\XX(a) \notin \sset, \XX(b) \in \sset$. We query $\oracle(c)$ for $c = \frac{a+b}{2}$, and update $a = c$ if $\XX(c) \notin \sset$ and $b = c$ otherwise. This shrinks the size of the range $|a-b|$ by a factor of $2$ at each step. We stop when $|a-b| \le \frac{\delta}{R}$ and output $\XX(b)$. This procedure takes $\log (M \frac{R}{\delta})$ queries and by the argument above correctly outputs $\XX' \in \sset$ so that $\Gamma(\XX') \le \min_{\XX \in \sset} \Gamma(\XX) + M L \sqrt{\frac{6 \delta}{\alpha}} + \delta$. 
\end{proof}

\subsection{Approximate Best Response Oracle from High Accuracy Solvers}
\label{sec:BRFromAdditiveError}
In this section, we prove \Cref{lem:BRFromAdditiveError}.
\begin{proof}[Proof of \Cref{lem:BRFromAdditiveError}]
By the 1st order optimality of $\uu^{\star}$, we have that
\begin{align*}
    \l\gg + \g_{\uu} r(\uu^{\star}), \uu^{\star} - \uu\r + \psi(\uu^{\star}) - \psi(\uu) \le 0, \forall \uu \in \uset
\end{align*}
Since $\lambda_{\min} \II_n \preceq \g^2_{\uu, \uu} r(\uu, \vv)$, $r(\cdot, \vv)$ is $\lambda_{\min}$-strongly convex (with respect to $\ell_2$ norm) for any fixed value of $\vv$ and thus 
\begin{align*}
    r(\algoutput{\uu}, \vv) \ge r(\uu^{\star}, \vv) + \l\g_{\uu} r(\uu^{\star}, \vv), \algoutput{\uu} - \uu^{\star}\r + \frac{\lambda_{\min}}{2} \norm{\algoutput{\uu} - \uu^{\star}}_2^2
\end{align*}
Therefore, we have 
\begin{align*}
\delta 
&\ge \l\gg, \algoutput{\uu} - \uu^{\star}\r + r(\algoutput{\uu}, \vv) - r(\uu^{\star}, \vv) + \psi(\algoutput{\uu}) - \psi(\uu^{\star}) \\
&\ge \l\gg + \g_{\uu} r(\uu^{\star}, \vv), \algoutput{\uu} - \uu^{\star}\r + \frac{\lambda_{\min}}{2} \norm{\algoutput{\uu} - \uu^{\star}}_2^2 + \psi(\algoutput{\uu}) - \psi(\uu^{\star}) \\
&\ge \frac{\lambda_{\min}}{2} \norm{\algoutput{\uu} - \uu^{\star}}_2^2
\end{align*}
where the last inequality uses the 1st order optimality of $\uu^{\star}.$
This implies \[\|\algoutput{\uu} - \uu^{\star}\|_{2}^2 \le \frac{2 \delta}{\lambda_{\min}}\,.\]
To prove item 2, note that since $\|\g^2_{\vv, \uu} r(\cdot, \vv)\|_{\infty} \le \beta$, we obtain 
that $\g_{\vv} r(\uu, \vv)$ is $\beta$-Lipschitz with respect to $\ell_{\infty}$ and thus 
\begin{align*}
\|\g_{\vv}r(\uu^{\star}, \vv) - \g_{\vv}r(\algoutput{\uu}, \vv)\|_{\infty} 
&\le \beta \|\algoutput{\uu} - \uu^{\star}\|_{\infty} \\
&\le \frac{\beta \sqrt{2 \delta}}{\sqrt{\lambda_{\min}}}\,,
\end{align*}
where for the last inequality we bounded $\|\algoutput{\uu} - \uu^{\star}\|_{\infty}$ by $\|\algoutput{\uu} - \uu^{\star}\|_{2}.$

Next, to prove item 1, note that for every $\uu \in \uset$, 
\begin{align*}
\l\gg + \g_{\uu} r(\algoutput{\uu}, \vv), \algoutput{\uu} - \uu\r + \psi(\algoutput{\uu}) - \psi(\uu)
&\le \delta + \l\g_{\uu} r(\algoutput{\uu}, \vv) - \g_{\uu} r(\uu^{\star}, \vv), \algoutput{\uu} - \uu\r \\
&\le \delta + \lambda_{\max} \|\algoutput{\uu} - \uu^{\star}\|_2 \|\algoutput{\uu} - \uu\|_2 \\
&\le \delta + 2 \lambda_{\max} \sqrt{\frac{2 n \delta}{\lambda_{\min}}} \usetSize
\end{align*}
where for the second inequality we used the fact that $\g^2_{\uu, \uu} r(\cdot, \vv) \preceq \lambda_{\max} \II_n$ and for the third inequality that $\|\algoutput{\uu} - \uu\|_2 \le \sqrt{n} \|\algoutput{\uu} - \uu\|_{\infty} \le 2  \sqrt{n} \usetSize$. 
\end{proof}

\section{Computable and Decomposable Functions}

In this section, we present the omitted proofs for \Cref{coro:convexFlow}, as well as the proofs of certain functions that we work with being computable or decomposable. 

\subsection{Almost-Linear Time Convex Flow Solver}
\label{subsec:single_comm_unconstrained}

\begin{theorem}
\label{thm:convexFlow_gen}
Let $G$ be a graph, $\eta = \exp(-\log^{O(1)} m)$ be a granularity parameter, $\uu \in \eta\Z^E_{>0}$ be the vector of edge capacities and $\dd \in \eta\Z^V$ be a demand vector.
Given a $C > 0$, a collection of $m$-decomposable functions on edges $\{c_e(\cdot)\}_e$, and their barriers $\{\barrier_e(\cdot)\}_e$ (\Cref{def:computable}), there is an algorithm that runs in $\almostTime(m)$ time and outputs a flow $\algoutput{\ff} \in \R^E$ that routes $\dd$ and, 
\begin{align*}
    c(\algoutput{\ff}) \le \min_{\imbal(\ff^{\star}) = \dd, 0 \le \ff^{\star} \le \uu} c(\ff^{\star}) + \exp(-\log^C m)
        \text{ where }
    c(\ff) \defeq \sum_{e\in E} c_e(\ff_e)\,.
\end{align*}
\end{theorem}

Before proving \Cref{thm:convexFlow_gen}, we provide the following helper lemma that shows how to 

\begin{lemma}
\label{lem:min_split_Flow_gen}
Let $G$ be a graph, $\eta = \exp(-\log^{O(1)} m)$ be a granularity parameter, $\uu \in \eta\Z^E_{>0}$ be the vector of edge capacities and $\dd \in \eta\Z^V$ be a demand vector.
Given a $C > 0$, a collection of convex cost functions on edges $\{c_e(\cdot)\}_e$, written as $c_e(x) = \min_{a+b = x, a \in [\lowerend, \upperend]} c_{e, 1}(a) + c_{e, 2}(b)$, so that  
there is an algorithm that runs in $\almostTime(m)$ time and outputs a flow $\algoutput{\ff} \in \R^E$ that routes $\dd$ and, 
\begin{align*}
    c(\algoutput{\ff}) \le \min_{\imbal(\ff^{\star}) = \dd, 0 \le \ff^{\star} \le \uu} c(\ff^{\star}) + \exp(-\log^C m)
        \text{ where }
    c(\ff) \defeq \sum_{e\in E} c_e(\ff_e)\,.
\end{align*}
\end{lemma}

\begin{proof}[Proof of \Cref{lem:min_split_Flow_gen}]
    We construct a graph $G'$ by adding additional vertices and edges, such that every edge $e \in E(G')$ has a corresponding cost function $c'_e(x)$ that is IPM-friendly and computing a flow of additive error $\exp(- \log^C m)$ on $G'$ results in a flow of additive error $\exp(- \log^C m)$ on $G$. 
    For this, 
    we create a gadget, involving a constant number of vertices and edges, with which we replace each edge $e \in E(G)$. Specifically, 
    for each edge $e \in E(G), e = (u, v)$, we let $H(e)$ be the graph with vertex set $V(H(e)) = \{u, u^+, v^+, v\}$, and edge set $E(H(e)) = \{e_1 = (u, u^+), e_2 = (u^+, v^+), e_3 = (u^+, v^+), e_4 = (v^+, v)\}$. For these edges, we set the capacity vector $\uu'$ as $\uu'_{e_1} = \uu_e, \uu'_{e_2} = \upperend - \lowerend, \uu'_{e_3} = \uu_e - \lowerend, \uu'_{e_4} = \uu_e$ and cost functions $c'_{e_1}(x) = 0, c'_{e_2}(x) = c_{e, 1}(x + \lowerend), c'_{e_3}(x) = c_{e, 2}(x - \upperend), c_{e_4}(x) = 0$. 
    Finally, we set the demand vector for $H(e)$ as $\dd'_{u} = \dd_{u}, \dd'_{v} = \dd_{v}, \dd_{u^+} = \upperend - \lowerend, \dd_{v^+} = \lowerend - \upperend$. 
    We let $G'$ be the graph where we replace each edge $e$ with $H(e)$. Specifically, for each $e = (u, v)$, we eliminate edge $(u, v)$, add vertices $(u^+, v^+)$, along with edges $E(H(e))$. 
    For every $e = (u, v) \in E(G)$, we call pair $(u, v)$ in $G'$ an \emph{original pair} and call $H(u, v)$ the corresponding gadget, with added vertices $u^+(u, v), v^+(u, v)$ and the corresponding edges. When clear from context, we drop the dependence of the edges and vertices in $H(u, v)$ on pair $(u, v)$. 
    Note that the graph $G'$ has a valid set of demands, since we do not change the demands for any of the vertices in an original pair, we simply add new vertices in with their demands. We refer to $\dd', \uu'$ as the demand and capacity vectors and $\cc'(\ff')$ as the total cost of routing demand $\dd'$ on graph $G'$ (i.e., $\cc'(\ff') = \sum_{e' \in E(G')} \cc'_e(\ff'_e)$).

    Let $\algoutput{\ff'}$ be a flow on graph $G'$ that routes $\dd'$ so that 
    \[c'(\algoutput{\ff'}) \le \min_{\imbal(\ff^{\star}) = \dd', 0 \le \ff^{\star} \le \uu'} c'(\ff^{\star}) + \exp(-\log^{C} m)\,.\]
    We show that in $O(m)$ time, we can create a flow $\algoutput{\ff}$ on $G$ routing $\dd$ with 
    \[c(\algoutput{\ff}) \le \min_{\imbal(\ff^{\star}) = \dd, 0 \le \ff^{\star} \le \uu} c(\ff^{\star}) + \exp(-\log^C m)\,.\]
    For this, we create a flow $\algoutput{\ff}$ on $G$ by defining, for every original pair $(u, v) \in V(G') \times V(G')$, $\algoutput{\ff}_{(u, v)} = \algoutput{\ff'}_{(u, u^+(u, v))}$. Now, fix an $(u, v) = e$ in $G'$. Note that for any flow $\ff'$ routing $\dd'$ on $G'$, we have $\ff'_{u, u^+} = \ff'_{v^+ = v}$, since $u^+$ only receives flow from $u$ and $v^+$ only sends flow to $v$. Moreover, we show that for any $e = (u, v) \in G$, $\ff'_e \in [0, \uu_e]$, the minimum cost arising from sending $\ff'_e$ units of flow from $u$ to $v$ through gadget $H(e)$ (i.e., by sending $\ff'_e$ units of flow 
    is achieved by $\min_{a+b = x, a \in [\lowerend, \upperend]} c_{e, 1}(a) + c_{e, 2}(b) = c_e(x)$. For this, note that if the amount of flow sent on edge $e_2$ is $y \ge 0$, then the total cost incurred on edges of $H(e)$ is $c_{e, 1}(y + \lowerend) + c_{e, 2}(x - y + \upperend - \lowerend)$. Clearly, since $\uu'_{(e_2)} = \upperend - \lowerend$, we must have $y \in [0, \upperend - \lowerend]$. Hence, every flow sending $\ff'_e$ units of flow from $u$ to $v$ through gadget $H(e)$ has a cost of $c_1(a) + c_2(b)$, for some $a, b$ with $a+b = \ff'_e$ and $a \in [\lowerend, \upperend]$. 
    Conversely, for any $a, b$ with $a+b = \ff'_e$ and $a \in [\lowerend, \upperend]$, the cost $c_1(a) + c_2(b)$ can be achieved by sending $a - \lowerend$ units of flow from $u^+$ to $v^+$ via edge $e_2$ and the rest via edge $e_3$. 
    
    Thus, $\ff$ indeed routes $\dd$ on $G$, as 
    \[\sum_{(v, u) \in E(G)} \ff_{(v, u)} - \sum_{(u, v) \in E(G)} \ff_{(u, v)} = \sum_{(v, u) \in E(G)} \ff'_{(v^+(v,u), u)} - \sum_{(u, v) \in E(G)} \ff'_{(u, u^+(u, v))} = \dd_{u}\,.\]
    Additionally, we have that $\min_{\imbal(\ff^{\star}) = \dd', 0 \le \ff^{\star} \le \uu'} c'(\ff^{\star}) = \min_{\imbal(\ff^{\star}) = \dd, 0 \le \ff^{\star} \le \uu} c(\ff^{\star})$. 
    Finally, note that for every $(u, v) \in E(G)$, 
    \[c_{(u, v)}(\ff_{(u, v)}) \le \sum_{e' \in E(H(u, v))} c'_{e'}(\algoutput{\ff'}_{e'})\,.\]
    Thus, summing up these inequalities for every $(u, v) \in E(G)$ yields that 
    \[c(\algoutput{\ff}) \le \min_{\imbal(\ff^{\star}) = \dd', 0 \le \ff^{\star} \le \uu'} c'(\ff^{\star}) = \min_{\imbal(\ff^{\star}) = \dd, 0 \le \ff^{\star} \le \uu} c(\ff^{\star}) + \exp(- \log^C m)\,,\]
    as needed. 
    
    Consequently, since the number of edges in $G'$ is at most $4 |E(G)| = O(m)$, we apply \Cref{thm:convexFlow} to graph $G'$, and, in $\almostTime(m)$, obtain a flow $\algoutput{\ff'}$ routing $\dd'$ on $G'$ with 
    \[c'(\algoutput{\ff}) \le \min_{\imbal(\ff^{\star}) = \dd', 0 \le \ff^{\star} \le \uu'} c'(\ff^{\star}) + \exp(-\log^{C} m)\,.\] 
    We then translate, in $O(m)$, $\ff'$ into a flow $\ff$ on graph $G$ routing $\dd$ as described in the second paragraph, and thus obtain a flow of desired properties. 
\end{proof}

Now, we are ready to prove \Cref{thm:convexFlow_gen}. 

\begin{proof}[Proof of \Cref{thm:convexFlow_gen}]
    For each edge $e \in E$, let $c_e^{(1)}, \ldots c_e^{(\edgenumber_e)}$ be the IPM friendly cost functions so that $c_e = \sum_{i \in [\edgenumber_e]} c_e^{(i)}$. For each $e \in E$, we split edge $e$
    into a path of $\edgenumber_e$ edges, and associate, to each edge $i \in [\edgenumber_e]$, cost $c_e^{(i)}$. We let $G'$ be the resulting graph and define $\dd'$, the demand vector associated with $G'$ by $\dd'_v = \dd_v$ for every $v \in G$ and $\dd'_v$ for every newly created vertex $v$. 
    The resulting graph $G'$ has $\O(m)$ edges, each of which has a corresponding cost function that is computable. 
    Moreover, any flow $\ff'$ routing $\dd'$ for graph $G'$ has a corresponding flow $\ff$ routing $\dd$ for graph $G$ with the same cost. 
    Hence, by applying \Cref{thm:convexFlow}, we obtain the desired result. 
\end{proof}

We now present the following general lemma about approximately minimizing a single-variable convex function, which is Lipschitz over a closed interval, given an approximate function value oracle.

\begin{lemma}[Lemma 33~\cite{cohen2016geometric}]
\label{lem:1D_opt_apx_val}
Let $\Gamma: \R \to \R$ be an $L$-Lipschitz convex function defined on the $[R_1, R_2]$ interval and let $\oracle: \R \to \R$ be an oracle such that $|\oracle(y) - \Gamma(y)| \le \gamma$ for all $y.$
In $O(\log \frac{L(R_2 - R_1)}{\gamma})$ time and $O(\log \frac{L(R_2 - R_1)}{\gamma})$ calls to $\oracle$, there is an algorithm that outputs $x \in [R_1, R_2]$ such that
\begin{align*}
    \Gamma(x) \le \min_{y \in [R_1, R_2]} \Gamma(y) + 4\gamma
\end{align*}
\end{lemma}

We are now ready to prove \Cref{coro:convexFlow}. 

\begin{proof}[Proof of \Cref{coro:convexFlow}]

Let $C > 0$ be the given exponent in the error parameter and $C_1 > 10 C$ be a constant so that $\RdLower, \RdUpper, \max_e \uu_e$ and $\max_u |\dd_u|$ are all bounded by $\exp(\log^{C_1} m)$ and $v(\cdot)$ is $\exp(\log^{C_1} m)$-smooth.
Define $C_2 = 10C_1$,
\begin{align*}
    c'(\ff) &\defeq c(\ff) + \exp(-\log^{10C_2} m) \sum_{e \in E} \ff_e^2\text{, and} \\
    \Gamma(\beta) &\defeq v(\beta) + \min_{\ff: \imbal(\ff) = \beta \dd, 0 \le \ff \le \uu} c'(\ff)
\end{align*}
Because $c'(\ff)$ is $\exp(-\log^{10C_2} m)$-strongly convex w.r.t. $\|\ff\|_2$, we know that $\Gamma(\beta)$ is $\exp(\log^{C_1} m + \log^{10C_2} m)$-smooth.
Because $|\RdUpper - \RdLower| \le 2 \exp(\log^{C_1} m)$, we know that $\Gamma(\beta)$ is also $\exp(\log^{C_3} m)$-Lipschitz on the interval $[\RdLower, \RdUpper]$ for $C_3 > 4C_1 + 10C_2.$

Next, we use \Cref{lem:1D_opt_apx_val} to approximately minimize $\Gamma(\beta)$ over $\beta \in [\RdLower, \RdUpper]$ to an additive error of $\gamma = \exp(-\log^{C_4} m)$ for $C_4 = 10C.$
In order to apply \Cref{lem:1D_opt_apx_val}, we need to implement an approximate evaluation oracle $\oracle$ for $\Gamma(\beta).$
Given the input $\beta$, $\oracle$ works by first rounding the input $\beta$ to $\Bar{\beta}$, the nearest integral multiple of $\eta = \exp(-\log^{C_5} m)$ for some large enough constant $C_5 > 0.$
Then, $\oracle$ calls the single commodity flow solver (\Cref{thm:convexFlow}) to find a highly accurate solution $\Bar{\ff}$ to the problem $\min_{\ff: \imbal(\ff) = \Bar{\beta} \dd} c'(\ff).$
Finally, $\oracle$ outputs $\Bar{\Gamma} = v(\Bar{\beta}) + c'(\Bar{\ff})$ as its estimation of $\Gamma(\beta).$
By the Lispchitzness of $\Gamma(\cdot)$ and the convex flow solver guarantee, $\Bar{\Gamma}$ and $\Gamma(\beta)$ is at most $\exp(-\log^{C_5} m + \log^{C_3} m)$ apart.
We set $C_5$ large enough so that $4 \exp(-\log^{C_5} m + \log^{C_3} m) \le \gamma = \exp(-\log^{C_4} m).$
Combining $\oracle$ with \Cref{lem:1D_opt_apx_val} finds a $\algoutput{\beta}$ so that
\begin{align*}
    \Gamma(\algoutput{\beta}) \le \min_{\beta \in [\RdLower, \RdUpper]} \Gamma(\beta) + \exp(-\log^{C_4} m)
\end{align*}

Let $\algoutput{\beta}$ be the output of \Cref{lem:1D_opt_apx_val} using our oracle $\oracle$ and $\algoutput{\ff}$ be the corresponding flow.
We know that $\algoutput{\beta}$ is an integral multiple of $\eta = \exp(-\log^{C_5} m).$

Next, we argue the approximation quality of $(\algoutput{\beta}, \algoutput{\ff}).$
Observe that, for any feasible flow $\ff$, i.e., $0 \le \ff \le \uu$, the following holds
\begin{align*}
    c(\ff) \le c'(\ff) = c(\ff) + \exp(-\log^{10C_2} m) \|\ff\|_2^2 \le c(\ff) + \exp(-\log^{10C_2} m + \log^{C_1} m + \log m)
\end{align*}
because $\max_e \uu_e \le \exp(\log^{C_1} m).$

Let $(\beta^{\star}, \ff^{\star})$ be the optimal solution.
By the approximate optimality of $(\algoutput{\beta}, \algoutput{\ff})$ and the definition of $c'(\cdot)$, we have
\begin{align*}
v(\algoutput{\beta}) + c(\algoutput{\ff})
&\le v(\algoutput{\beta}) + c'(\algoutput{\ff}) \\
&\le v(\beta^{\star}) + c'(\ff^{\star}) + \exp(-\log^{C_4} m) \\
&\le v(\beta^{\star}) + c(\ff^{\star}) + \exp(-\log^{10C_2} m + \log^{C_1} m + \log m) + \exp(-\log^{C_4} m)
\end{align*}
where the last inequality comes from that $0 \le \ff^{\star} \le \uu$ and we have $\|\ff^{\star}\|_2^2 \le m \exp(\log^{C_1} m).$

Because $C_2 = 10C_1 > 100C$ and $C_4 = 10C$, we have that
\begin{align*}
v(\algoutput{\beta}) + c(\algoutput{\ff}) \le v(\beta^{\star}) + c(\ff^{\star}) + \exp(-\log^{C} m)
\end{align*}

The runtime is $\almostTime(m)$ because \Cref{lem:1D_opt_apx_val} takes only $\O(1)$ time and makes $\O(1)$ calls to \Cref{thm:convexFlow}, each of which takes $\almostTime(m)$ time.
This concludes the proof.
\end{proof}

\subsection{Proving computability}
\label{sec:proving_computability}

First, we start with the following folklore lemma about self-concordant barriers on intersection of two convex sets. 

\begin{lemma}[Proposition 3.1.1 ii) from \cite{nemirovski04}]\label{lem:inters_self_concord}
    Let $\barrier_i$ be a $\nu_i$-self-concordant barrier on convex set $K_i$ for every $i \in [m]$ and $\alpha_i \ge 1$ real numbers. Then, the function $\sum_{i \in [m]} \alpha_i \barrier_i$ is a $(\sum_{i \in [m]} \alpha_i \nu_i)$-self concordant barrier for set $K = \cap_{i \in [m]} K_i$. 
\end{lemma}

Next, we provide a lemma from \cite{nemirovski04} that helps with constructing barriers for epigraphs of convex functions. 

\begin{lemma}[Proposition 9.2.2 from \cite{nemirovski04}]\label{lem:gen_barrier_creation}
    Let $c(x)$ be a 3-times differentiable convex function such that 
    \[|c'''(x)| \le 3 \beta x^{-1} c''(x), \forall x>0 \,.\]
    Then the function $\barrier(x, t) := - \log(t - c(x)) - \max(1, \beta^2) \log(x)$ is a $(1+\max(1, \beta^2))$-self concordant barrier for the $2$-dimensional convex domain $\dset_c = \{(x, t) \mid c(x) \le t\}$. 
\end{lemma}

Now, we provide several lemmas about the computability of the single-variable convex functions we work with to prove \Cref{thm:approxL1InfMin}. 

\begin{lemma}\label{lem:xlogx_comp}
    Let $r: [0, R] \to \R$ defined by $r(x) = (x + \xi) \log (x + \xi)$, for some $\xi > 0$. Provided that $R, \frac{1}{\xi} \le \poly(m)$, $r$ is $(m, \log^{O(1)} m)$-computable. 
\end{lemma}

\begin{proof}[Proof of \Cref{lem:xlogx_comp}]
    First, note that for any $x \in [0, R]$, we have $|\log (x + \xi)| \le \max(|\log \xi|, |\log (\xi + R)|) = O(\log n) + O(|x|)$. 
    This proves property \ref{item:costQuasiPoly} of \Cref{def:computable}. 
    Next, note that $r'(x) = \frac{1}{x + \xi}$, $r''(x) = -\frac{1}{(x+\xi)^2}$, and 
    $r'''(x) = \frac{2(x + \xi)}{(x+\xi)^4} = \frac{2}{(x+\xi)^3}$. Further, since $\xi > 0$, we have $\frac{2}{(x+\xi)^3} \le \frac{2}{x} \cdot \frac{1}{(x+\xi)^2}$, which implies 
    \[|r'''(x)| \le 3 x^{-1} r''(x), \forall x>0\,.\]
    Thus, 
    by \Cref{lem:gen_barrier_creation}, with $\beta = 1$, $\barrier(x, t) := - \log(t - (x + \xi) \log (x + \xi)) - 4 \log(x)$ is a $2$-self-concordant barrier for $r$, which proves \ref{item:costSC}.  
    Note that $\barrier(x, t) = (\frac{1 + \log (x + \xi)}{t - (x + \xi) \log (x + \xi)}, - \frac{1}{t - (x + \xi) \log (x + \xi)} - \frac{1}{t})$ and that 
    \[\nabla^2 \barrier(x, t) = \begin{pmatrix}
        \frac{\frac{t - (x + \xi) \log (x + \xi)}{x + \xi} + (1 + \log (x + \xi))^2}{(t - (x + \xi) \log (x + \xi))^2} & \frac{-1}{(t - (x + \xi) \log (x + \xi))^2} \\
        \frac{1 + \log (x + \xi)}{(t - (x + \xi) \log (x + \xi))^2} & \frac{1}{(t - (x + \xi) \log (x + \xi))^2} + \frac{1}{t^2}
    \end{pmatrix}\,.\]
    Consequently, property \ref{item:costHessianCompute} of \Cref{def:computable} is satisfied. Lastly, note that 
    for any $|x|, |t| \le m^K$, $|\barrier(x,t)| \le \O(1)$ implies $t - (x + \xi) \log (x + \xi) \ge \frac{1}{\poly(m)}$ and $t \ge \frac{1}{\poly(m)}$. Whenever $t - (x + \xi) \log (x + \xi) \ge \frac{1}{\poly(m)}$, $t \ge \frac{1}{\poly(m)}$ and $|x|, |t| \le m^K$, 
    the quantities $\frac{1}{t - (x + \xi) \log (x + \xi)}, \frac{1}{t}, \frac{1}{x + \xi}, \log(x + \xi)$ are all bounded by $\poly(m) = \exp(\O(1))$. Thus, property \ref{item:costHessian} follows, as this implies that the entries of Hessian $\nabla^2 \barrier(x, t)$ are all bounded by $\exp(\O(1))$ for such value of $(x, t)$. 
\end{proof}

\begin{lemma}\label{lem:short_powers}
    Let $p, m \in \Z_{>0}$ and $w > 0$ with $p \ge 2, p = \otilde(1), w = \exp(\otilde(1))$\footnote{Throughout this subsection, $\otilde(1)$ hides $\log^c m$ factors for $c$ constant}.
    Then, $c: \R \to \R$ by $c(x) := w x^p$ is $(m, \log^{O(1)} m)$-computable. 
\end{lemma}

The proof is adapted from the proof of Theorem 10.14 in \cite{CKL+22}. 

\begin{proof}[Proof of \Cref{lem:short_powers}]
    First, note that since $w = \exp(\otilde(1))$, there exists $K = \otilde(1)$ so that $m^{K/2} \ge w$ and $K \ge 2p$. Hence, by AM-GM, we have that $c(x) \le m^K + |x|^K, \enspace \forall x \in \R$. 
    This proves property \ref{item:costQuasiPoly} of \Cref{def:computable}. 
    
    Second, we show that $\barrier(x, t) := -2 \log(t) - \log(t^{2/p} - w^{2/p} x^2)$ is a $4$-self concordant barrier for $c$ on $\mathcal{D}_c$. 
    Note that the first part of the lemma ($4$-self concordance of $\barrier(x, t)$) follows from Example 9.2.1 from \cite{nemirovski04}, which states that $-2 \log(t) - \log(t^{2/p} - x^2)$ is a self-concordant barrier for $t \ge x^p$, coupled with applying the linear map $x \to w^{1/p} x$. 
    This proves property \ref{item:costSC}. 
    Next, we show that both $\nabla \barrier, \nabla^2 \barrier$ can be computed and accessed in $\O(1)$ time. Note that $\nabla \barrier(x, t) = (\frac{2 w^{2/p} x}{t^{2/p} - w^{2/p} x^2}, - \frac{2}{t} - \frac{2 t^{2/p - 1} / p}{t^{2/p} - w^{2/p} x^2})$ and that 
    \[\nabla^2 \barrier(x, t) = \begin{pmatrix}
        \frac{2 w^{2/p}}{t^{2/p} - w^{2/p} x^2} + \frac{4 w^{4/p} x^2}{(t^{2/p} - x^2)^2} & - \frac{4 w^{2/p} x t^{2/p-1}}{p (t^{2/p} - w^{2/p} x^2)^2} \\
        - \frac{4 w^{2/p} x t^{2/p-1}}{p (t^{2/p} - w^{2/p} x^2)^2} & \frac{2}{t^2} + \frac{4 t^{4/p-2}}{p^2 (t^{2/p} - w^{2/p} x^2)^2} + \frac{(2p-4) t^{2/p-2}}{p^2 (t^{2/p} - w^{2/p} x^2)^2}
    \end{pmatrix}\,.\]
    Thus, $\nabla \barrier, \nabla^2 \barrier$ can indeed be computed and accessed in $O(1)$ time, which proves property \ref{item:costHessianCompute} of \Cref{def:computable}. 
    Lastly, note that whenever $|x|, |t| \le m^{K}$ and $-2 \log(t) - \log(t^{2/p} - w^{2/p} x^2) = \O(1)$, we have $t \ge \frac{1}{m^{\O(1)}}$ and $t^{2/p} - w^{2/p} x^2 \ge \frac{1}{m^{\O(1)}}$. Since the entries of $\nabla^2 \barrier(x, t)$ are sums of a constant number fractions where the numerator is a polynomial in $x, t$ and the denominator is a polynomial in $t, t^{2/p} - w^{2/p} x^2$ (with coefficients $\exp(\O(1))$ in both cases), for such values of $(x, t)$ with $|x|, |t| \le m^{K}, \barrier(x, t) = \O(1)$, we have 
    $\|\nabla^2 \barrier(x, t)\|_{\infty} \le \exp(\O(1))$, which proves property \ref{item:costHessian}. 
\end{proof}

We dedicate the rest of the section to proving the following lemma about the computability of odd (small) powers of the $\gamma_q(\cdot)$ function in \Cref{def:gamma}. 

\begin{lemma}\label{lem:gamma_big_p}
    Let $\gamma_q(\cdot; f)$ be the function defined in \Cref{def:gamma} for $1 \le q < 2$. Then, for any odd $p \ge 3, p = \O(1)$ and $w_1, w_2, f \le \exp(\O(1))$, $w_1 (\gamma_q(\cdot; f) + w_2)^p$ is $m$-decomposable. 
\end{lemma}

Before proving \Cref{lem:gamma_big_p}, we show that $w_1 (\gamma_q(x; f) + w_2)^p$ can be written in the form $\min_{a+b=x, a \in [\lowerend, \upperend]} c_1(a) + c_2(b)$ for suitable functions $c_1, c_2$, each of which we later show are $m, \otilde(\log^{O(1)} \allowbreak m)$-computable. 

\begin{lemma}\label{lem:gamma_split}
    Let $p \in \Z_{>0}$, $w_1, f \ge 0, w_2 \in \R$ and $c(x) := w_1 (\gamma_q(x) + w_2)^p$. Define $c_1(x) := w_1 (\frac{q}{2} f^{q-2} x^2 + w_2)^p, c_2(x) := w_1 \max((|x+f|^q - (1 - \frac{q}{2}) f^q + w_2)^p - (\frac{q}{2} f^q + w_2)^p, (|x-f|^q - (1 - \frac{q}{2}) f^q + w_2)^p - (\frac{q}{2} f^q + w_2)^p)$. Then, we have 
    \[c(x) = \min_{a + b = x, |a| \le f} c_1(a) + c_2(b)\,.\]
\end{lemma}

\begin{proof}[Proof of \Cref{lem:gamma_split}]
    First, we show that $c(x) \ge \min_{a + b = x, |a| \le f} c_1(a) + c_2(b)$. To see why this is the case, first note that if $|x| \le f$, $c(x) = c_1(x) \ge \min_{a + b = x, |a| \le f} c_1(a) + c_2(b)$, as $c_2(0) = 0$. 
    Second, for $x > f$, $c(x) = c_1(f) + c_2(x-f)$, since $c_2(x-f) + c_1(f) = w_1 (|x|^q - (1 - \frac{q}{2}) f^q + w_2)^p - w_1 (\frac{q}{2} f^q + w_2)^p + w_1 (\frac{q}{2} f^q + w_2)^p = w_1 (\gamma_q(x) + w_2)^p$. Third, for $x < -f$, we have $c(x) = c_1(-f) + c_2(x+f)$, since $c_2(x+f) + c_1(-f) = w_1 (|x|^q - (1 - \frac{q}{2}) f^q + w_2)^p - w_1 (\frac{q}{2} f^q + w_2)^p + w_1 (\frac{q}{2} f^q + w_2)^p = w_1 (\gamma_q(x) + w_2)^p$. 

    Next, we show that for every $x \in \R$, $\min_{a + b = x, |a| \le f} c_1(a) + c_2(b) = c(x)$. For this, first note that $c_2$ is a differentiable function everywhere except at $x = 0$, with the property that 
    \begin{equation}\label{ineq:c1_c2_derivs}
        |c_2'(x)| > |c_1'(y)|, \forall x \ne 0, |y| < f\,.
    \end{equation}
    To see why this is true, note that \[c_2'(x) = w_1 p q (x+f)^{q-1} \left((x+f)^q - \left(1 - \frac{q}{2}\right) f^q + w_2\right)^{p-1}, \forall x > 0\,,\]
    and \[c_2'(x) = w_1 p q (-x+f)^{q-1} \left((-x+f)^q - \left(1 - \frac{q}{2}\right) f^q + w_2\right)^{p-1}, \forall x < 0\,.\]
    Hence, for any $x \ne 0$, $|c_2'(x)| > w_1 p q f^{q-1} (f^q - (1 - \frac{q}{2}) f^q + w_2)^{p-1}$. 
    Additionally, \[c_1'(y) = w_1 p q f^{q-2} y \left(\frac{q}{2} f^{q-2} y^2 + w_2\right)^{p-1}\,.\] 
    Hence, for $|y| < f$, $|c_1'(y)| < w_1 p q f^{q-1} (\frac{q}{2} f^{q} + w_2)^{p-1}$, which implies \ref{ineq:c1_c2_derivs}. 
    Now, by \ref{ineq:c1_c2_derivs}, for any $x \in \R$, if $(a^*(x), b^*(x)) = \argmin_{a + b = x, |a| \le f} c_1(a) + c_2(b)$, then either $b^* = 0$, or $a^* = \pm f$. Thus, for $|x| > f$, we must have either $a^*(x) = f$ or $a^*(x) = -f$. Since $c_2$ is strictly increasing on $R_{>0}$ and strictly decreasing on $R_{<0}$, this implies $a^*(x) = \mathrm{sign}(x) f$ and $b^*(x) = x - \mathrm{sign}(x) f$, so $\min_{a + b = x, |a| \le f} c_1(a) + c_2(b) = c(x)$. 
    Finally, for $|x| \le f$, we must have $(a^*(x), b^*(x)) = (x, 0)$. To see why this is the case, note that $b^*(x) \ne 0$ would imply that $a^*(x) = \pm f$. Since $c_2(-f) = c_2(f) = \max_{|y|\le f} c_2(y)$ and $c_2(x) \ge 0, \enspace \forall x \in R$, with equality if and only if $x = 0$, we would obtain a contradiction. 
\end{proof}

Next, we show how to construct self-concordant barriers for each of the functions $c_1(x) := w_1 (\frac{q}{2} f^{q-2} x^2 + w_2)^p, c_2(x) := w_1 \max((|x+f|^q - (1 - \frac{q}{2}) f^q + w_2)^p - (\frac{q}{2} f^q + w_2)^p, (|x-f|^q - (1 - \frac{q}{2}) f^q + w_2)^p - (\frac{q}{2} f^q + w_2)^p)$. 

\begin{lemma}\label{lem:q_power}
    Let $c(x) := w_1 (|x - \offset|^q + w_2)^p$ for $w_1 > 0$ and $w_2, \offset \in \R$. Then, the function $\barrier(x, t) := -\log\big(\big(\frac{t^{1/p}}{w_1^{1/p}} - w_2\big)^{1/q} - x - \offset\big) - \log\big(\big(\frac{t^{1/p}}{w_1^{1/p}} - w_2\big)^{1/q} + x + \offset\big) - 2 \log (t - w_1 w_2^p)$ 
    is a $4$-self-concordant barrier for $\dset_c$, the epigraph of $c(x)$. 
\end{lemma}

\begin{proof}[Proof of \Cref{lem:q_power}]
    The epigraph of $c(x)$, $\{(x, t) : t \scgeq w_1 (|x - \offset|^q + w_2)^p\}$, can be rewritten as $\{(x, t) : t > w_1 w_2^p, \frac{t^{1/p}}{w_1^{1/p}} - w_2 > |x + \offset|^q\}$, or, equivalently, as $\{(x, t) : t \scgeq w_1 w_2^p, (\frac{t^{1/p}}{w_1^{1/p}} - w_2)^{1/q} \scgeq |x + \offset|\}$. This set is equal to the intersection of $K_1, K_2$, where $K_i = \{(x, t) : t \scgeq w_1 w_2^p, (\frac{t^{1/p}}{w_1^{1/p}} - w_2)^{1/q} \scgeq (-1)^i (x + \offset)\}$. 
    We will show that for each $K_i$, $\barrier_i(x, t) = -\log((\frac{t^{1/p}}{w_1^{1/p}} - w_2)^{1/q} - (-1)^i (x + \offset)) + \log (t - w_1 w_2^p)$ is a $2$-self-concordant barrier for $K_i$. Note that by \Cref{lem:inters_self_concord}, this suffices for the proof. 

    We build a barrier for $K_1$, as doing so for $K_2$ is analogous. In particular, we will apply Theorem 9.1.1 from \cite{nemirovski04}. To do so, define $\varphi(t) := (\frac{t^{1/p}}{w_1^{1/p}} - w_2)^{1/q}$. 
    To do so, we define sets $G^+ := \R_{>0}$, $G^{-} := \{(x, t): t \scgeq w_1 w_2^p\}$, along with functions $F^+(y) = -\log y$, which is $1$-self concordant on $G^+$, $F^{-}(y) := -\log (y - w_1 w_2^p)$, which is $1$-self concordant on $G^-$, and $A: G^- \to \R$ with $A(x, t) := (\frac{t^{1/p}}{w_1^{1/p}} - w_2)^{1/q} + x + \offset$. 
    To apply the theorem, it suffices to prove that the function $A(x, t)$ is $\beta$-appropriate (Definition 9.1.1 from \cite{nemirovski04}), for $G^+$, with $\beta = 1$. For this, since the recessive cone $\calR^+(G^+)$ (defined in Section 9.1 of \cite{nemirovski04}) of $G^+$ is $\calR^+(G^+) = \R_{>0} = G^+$, and since $A(x, t)$ is separable in $x, t$ (implying that $\frac{\partial^2}{\partial t \partial x} A(x, t) = \frac{\partial^2}{\partial x \partial t} A(x, t) = 0$), it suffices to show the following two conditions: 
    \begin{enumerate}
        \item $A(x, t)$ is concave.
        \item $D^3 A(x, t) [h, h, h] \le -3 \beta D^2 A(x, t) [h, h]$ for any $h \in \R^2$ with $(x, t) \pm h \in G^-$.
    \end{enumerate}
    First, the concavity of $A(x, t)$ follows from the fact that $\pm x$ is concave and $\varphi(t)$ is concave on $\{t > w_1 w_2^p\}$, as it is the composition of a non-decreasing concave function ($x \mapsto x^{1/q}$) with a concave function that is everywhere positive ($t \mapsto \frac{t^{1/p}}{w_1^{1/p}} - w_2$). To prove the condition $2$, note that it suffices to only consider terms of the form $\frac{\partial^{\ell}}{\partial t ... \partial t} A(x, t)$ and $\frac{\partial^{\ell}}{\partial x ... \partial x} A(x, t)$, as $\frac{\partial^2}{\partial t \partial x} A(x, t) = 0$. Since 
    the second and third derivatives of $x + \offset$ are $0$ with respect to either $x$ or $t$, it suffices to show that $|\varphi(t)'''| \le -3 (t - w_1 w_2^p)^{-1} \beta \varphi(t)'', \forall (x, t) \in K_1$. 
    Note that $\varphi(t)''' = \frac{t^{3/p - 3}}{w_1^{3/p} p^3} \frac{1}{q} (\frac{1}{q} - 1) (\frac{1}{q} - 2) (\frac{t^{1/p}}{w_1^{1/p}} - w_2)^{\frac{1}{q} - 3}$. 
    Hence, it suffices to show that for all $t > w_1 w_2^p$
    \[\frac{t^{3/p - 3}}{p^3 w_1^{3/p}} \frac{1}{q} \left(\frac{1}{q} - 1\right) \left(\frac{1}{q} - 2\right) \left(\frac{t^{1/p}}{w_1^{1/p}} - w_2\right)^{\frac{1}{q} - 3} \le - \frac{3 \beta t^{2/p - 2}}{p^2 w_1^{2/p}} \frac{1}{q} \left(\frac{1}{q} - 1\right) \left(t - w_1 w_2^p\right)^{-1} \left(\frac{t^{1/p}}{w_1^{1/p}} - w_2\right)^{\frac{1}{q} - 2}\,.\]
    For this, it suffices to show $ \le (t^{1/p} - w_1^{1/p} w_2) + (t - w_1 w_2^p) t^{1/p - 1}$ for all $t > w_1 w_2^p$, which is obviously true. 
    
    By Theorem 9.1.1, the function $\barrier_1(x, t) = - \log (t - w_1 w_2^p) - \log (\varphi(t) + x + \offset)$, which can be written as $\barrier_1(x, t) = - \log(t - w_1 w_2^p) - \log((\frac{t^{1/p}}{w_1^{1/p}} - w_2)^{1/q} + x + \offset)$, is $2$-self concordant on $\{(x, t) \in G^-: A(x, t) \in G^+\}$. Note that by definition, $A(x, t) > 0, \forall (x, t) \in G^-$. Thus, the set 
    $\{(x, t) \in G^-: A(x, t) \in G^+\} = G^-$, which implies that $\{(x, t) \in G^-: A(x, t) \in G^+\} = G^- = K_1$. 
    Analogously, we obtain that $\barrier_2(x, t) = - \log (t - w_1 w_2^p) - \log ((\frac{t^{1/p}}{w_1^{1/p}} - w_2)^{1/q} - x - \offset)$ is $2$-self concordant on $K_2$. This suffices for our proof. 
\end{proof}

Finally, we are ready to prove \Cref{lem:gamma_big_p}.

\begin{proof}[Proof of \Cref{lem:gamma_big_p}]
    By \Cref{lem:gamma_split}, we can write
    \[w_1 (\gamma_q(x; f) + w_2)^p = \min_{a + b = x, |a| \le f} c_1(a) + c_2(b)\,,\]
    where $c(x) := w_1 (\gamma_q(x) + w_2)^p$. Define $c_1(x) := w_1 (\frac{q}{2} f^{q-2} x^2 + w_2)^p, c_2(x) := w_1 \max((|x+f|^q - (1 - \frac{q}{2}) f^q + w_2)^p - (\frac{q}{2} f^q + w_2)^p, (|x-f|^q - (1 - \frac{q}{2}) f^q + w_2)^p - (\frac{q}{2} f^q + w_2)^p)$. 
    We already know that $c_1(x)$ is $(m, \otilde(\log^{O(1)} m)$-computable, by \Cref{lem:short_powers}. Hence, it suffices to show that $c_2(x)$ is $(m, \otilde(\log^{O(1)} m)$-computable as well. Property \ref{item:costQuasiPoly} of \Cref{def:computable} is trivially satisfied, as $w_1, w_2, f \le \exp(\O(1))$. 
    Next, by \Cref{lem:q_power}, 
    we have that $\barrier(x, t) := -\log\big(\big(\frac{t^{1/p}}{w_1^{1/p}} - w_2\big)^{1/q} - x - \eff\big) - \log\big(\big(\frac{t^{1/p}}{w_1^{1/p}} - w_2\big)^{1/q} + x + \eff\big) - 2 \log (t - w_1 w_2^p)$ 
    is a $4$-self-concordant barrier for $\dset_{c_2}$. This implies property \ref{item:costSC} of \Cref{def:computable}. 
    Moreover, we have \[\nabla_x \barrier(x, t) = \frac{w_1^{1/pq}}{(t^{1/p} - w_1^{1/p} w_2)^{1/q} - w_1^{1/pq} (x + \eff)} - \frac{w_1^{1/pq}}{(t^{1/p} - w_1^{1/p} w_2)^{1/q} + w_1^{1/pq} (x + \eff)}\,,\]
    \begin{align*}
        \nabla_t \barrier(x, t) = & - \frac{t^{1/p - 1} (t^{1/p} - w_1^{1/p} w_2)^{1/q - 1}}{pq \left((t^{1/p} - w_1^{1/p} w_2)^{1/q} - w_1^{1/pq} (x + \eff)\right)} - \frac{t^{1/p - 1} (t^{1/p} - w_1^{1/p} w_2)^{1/q - 1}}{pq \left((t^{1/p} - w_1^{1/p} w_2)^{1/q} + w_1^{1/pq} (x + \eff)\right)} \\ 
        & - \frac{2}{t - w_1 w_2^p}\,.
    \end{align*}
    To compute $\nabla^2 \barrier(x, t)$, we simply take derivatives of the terms above, and obtain that each entry of $\nabla^2 \barrier(x, t)$ is a sum of a constant number of fractions, each of which is represented by $O(1)$ terms and easily accessed. Hence, property \ref{item:costHessianCompute} of \Cref{def:computable} follows. Lastly, to prove property \ref{item:costHessian}, note that by the chain rule of differentiation,  
    each entry of $\nabla^2 \barrier(x, t)$ will be a sum of a constant number of 
    fractions where the numerator is a sum of constantly many powers of $x, t$ and the denominator is a sum of constantly many powers of $\frac{(t^{1/p} - w_1^{1/p} w_2)^{1/q}}{w_1^{1/pq}} + x + \eff, \frac{(t^{1/p} - w_1^{1/p} w_2)^{1/q}}{w_1^{1/pq}} - x - \eff, t - w_1 w_2^p$ (with coefficients bounded by $\O(1)$ in both cases). Hence, 
    whenever $|x|, |t| \le m^{K}$ and $\barrier(x, t) = \O(1)$, we have $\frac{(t^{1/p} - w_1^{1/p} w_2)^{1/q}}{w_1^{1/pq}} + x + \eff, \frac{(t^{1/p} - w_1^{1/p} w_2)^{1/q}}{w_1^{1/pq}} - x - \eff, t - w_1 w_2^p \ge \frac{1}{m^{\O(1)}}$. Thus, for such values of $(x, t)$ with $|x|, |t| \le m^{K}, \barrier(x, t) = \O(1)$, we have 
    $\|\nabla^2 \barrier(x, t)\|_{\infty} \le \exp(\O(1))$, which proves property \ref{item:costHessian}. 
\end{proof}

\end{document}